\documentclass[final,5p,times,twocolumn]{elsarticle}

\usepackage{amsmath,amssymb,amsfonts,amsthm}
\usepackage{bm}
\usepackage{tabularx,booktabs,colortbl}
\usepackage{graphicx}
\usepackage{hyperref}
\usepackage{breqn}
\usepackage{physics}
\usepackage{subcaption}
\usepackage{caption}
\usepackage{enumitem}
\usepackage[toc, title]{appendix}
\usepackage[acronym]{glossaries}
\usepackage{multirow}
\usepackage{tikz}
\usepackage{placeins}
\usepackage{float}

\usetikzlibrary{positioning,calc,fit}
\biboptions{sort&compress}

\newacronym{bdi}{BDI}{Belief-Desire-Intention}
\newacronym{npe}{LPE}{Latent Perceptual Estimate}
\newacronym{dcm}{DCM}{Dynamic Causal Modelling}
\newacronym{sa}{SA}{Stationary Audio}
\newacronym{da}{DA}{Dynamic Audio}
\newacronym{sv}{SV}{Stationary Visual}
\newacronym{dv}{DV}{Dynamic Visual}
\newacronym{io}{I/O}{Input-Output}
\newacronym{is}{I/S}{Input-to-State}
\newacronym{mpc}{MPC}{Model Predictive Control}
\newacronym{hri}{HRI}{human-robot interaction}
\newacronym{iss}{ISS}{input-to-state stability}

\DeclareRobustCommand{\bigO}{%
  \text{\usefont{OMS}{cmsy}{m}{n}O}%
}

\newtheorem{lemma}{Lemma}
\newtheorem{proposition}{Proposition}
\newtheorem{definition}{Definition}
\newtheorem{assumption}{Assumption}
\newtheorem{corollary}{Corollary}
\newtheorem{remark}{Remark}

\newcommand{\zasl}{z_{\ell}^{\textup{\textsc{ats}}}}
\newcommand{\bu}{\boldsymbol{u}}
\newcommand{\fasl}{F_{\ell}^{\textup{\textsc{ats}}}}
\newcommand{\fpil}{F_{\ell}^{\textup{\textsc{pi}}}}
\newcommand{\gasl}{G_{\ell}^{\textup{\textsc{ats}}}}
\newcommand{\gpil}{G_{\ell}^{\textup{\textsc{pi}}}}
\newcommand{\alphapil}{\alpha_{\ell}^{\textup{\textsc{pi}}}}
\newcommand{\betaasl}{\beta_{\ell}^{\textup{\textsc{ats}}}}
\newcommand{\gammaaslm}{\gamma_{\ell m}^{\textup{\textsc{ats}}}}
\newcommand{\faslmax}{\overline{F_{\ell}^\textup{\textsc{ats}}}}
\newcommand{\gaslmin}{\underline{G_{\ell}^\textup{\textsc{ats}}}}
\newcommand{\zaslmax}{\overline{z_{\ell}^\textup{\textsc{ats}}}}
\newcommand{\gasmtilde}{\tilde{G}_{m}^{\textup{\textsc{ats}}}}
\newcommand{\uset}{\mathcal{U}}
\newcommand{\ulset}{\mathcal{U}_{\ell}}
\newcommand{\umset}{\mathcal{U}_{m}}
\newcommand{\xnpel}{x_{\ell}^{\textup{\textsc{lpe}}}}
\newcommand{\bxnpe}{\bx^{\textup{\textsc{lpe}}}}
\newcommand{\bzi}{\bz^{\textup{\textsc{i}}}}
\newcommand{\bzimini}{\bz_{-i}^{\textup{\textsc{i}}}}
\newcommand{\xnpelset}{\mathcal{X}_{\ell}^{\textup{\textsc{lpe}}}}
\newcommand{\xgiset}{\mathcal{X}_{i}^{\textup{\textsc{g}}}}
\newcommand{\xbset}{\mathcal{X}^{\textup{\textsc{b}}}}
\newcommand{\xnpeset}{\mathcal{X}^{\textup{\textsc{lpe}}}}
\newcommand{\ziiset}{\mathcal{Z}_{i}^{\textup{\textsc{i}}}}
\newcommand{\zaiset}{\mathcal{Z}_{i}^{\textup{\textsc{a}}}}
\newcommand{\zimset}{\mathcal{Z}_{m}^{\textup{\textsc{i}}}}
\newcommand{\ziminiset}{\mathcal{Z}_{-i}^{\textup{\textsc{i}}}}
\newcommand{\ziset}{\mathcal{Z}^{\textup{\textsc{i}}}}
\newcommand{\zaset}{\mathcal{Z}^{\textup{\textsc{a}}}}
\newcommand{\bxb}{\bx^{\textup{\textsc{b}}}}
\newcommand{\xb}{x^{\textup{\textsc{b}}}}
\newcommand{\xg}{x^{\textup{\textsc{g}}}}
\newcommand{\bxg}{\bx^{\textup{\textsc{g}}}}
\newcommand{\thetapil}{\theta_{\ell}^{\textup{\textsc{pi}}}}
\newcommand{\chipil}{\chi_{\ell}^{\textup{\textsc{pi}}}}
\newcommand{\phipil}{\phi_{\ell}^{\textup{\textsc{pi}}}}
\newcommand{\gisi}{G_{i}^{\textup{\textsc{if}}}}
\newcommand{\fisi}{F_{i}^{\textup{\textsc{if}}}}

\newcommand{\gisimax}{\overline{{G}_{i}^\textup{\textsc{if}}}}
\newcommand{\fisimax}{\overline{{F}_{i}^\textup{\textsc{if}}}}
\newcommand{\phiisi}{\phi_{i}^{\textup{\textsc{if}}}}
\newcommand{\phiasi}{\phi_{i}^{\textup{\textsc{as}}}}
\newcommand{\zii}{z_{i}^{\textup{\textsc{i}}}}
\newcommand{\um}{u_{m}}
\newcommand{\ul}{u_{\ell}}
\newcommand{\ziimax}{\overline{z_{i}^\textup{\textsc{i}}}}
\newcommand{\zai}{z_{i}^{\textup{\textsc{a}}}}
\newcommand{\gasi}{G_{i}^{\textup{\textsc{as}}}}
\newcommand{\fasi}{F_{i}^{\textup{\textsc{as}}}}
\newcommand{\hasi}{H_{i}^{\textup{\textsc{as}}}}
\newcommand{\sasi}{S_{i}^{\textup{\textsc{as}}}}
\newcommand{\gasimax}{\overline{G_{i}^\textup{\textsc{as}}}}
\newcommand{\hasimax}{\overline{H_{i}^\textup{\textsc{as}}}}
\newcommand{\fiic}{F_{ii}^{\textup{\textsc{c}}}}
\newcommand{\fc}{F^{\textup{\textsc{c}}}}
\newcommand{\gimc}{G_{im}^{\textup{\textsc{c}}}}
\newcommand{\card}[1]{\left\lvert #1 \right\rvert}
\renewcommand\norm[1]{\left\lVert#1\right\rVert}
\renewcommand\abs[1]{\left\lvert #1 \right\rvert}
\newcommand{\enpe}{e^{\textup{\textsc{lpe}}}}
\newcommand{\latsl}{L^{\textup{\textsc{ats}}}_{\ell}}
\newcommand{\bx}{\bm{x}}

\newcommand{\xigoals}{x_i^\textup{\textsc{g}}}
\newcommand{\xbeliefs}{\bm{x}^\textup{\textsc{b}}}
\newcommand{\uraw}{u^{\textup{\textsc{raw}}}}
\newcommand{\xnpebar}{\overline{x^{\textup{\textsc{lpe}}}}}
\newcommand{\xnpebartraj}{\sup_{i\ge 0}\norm{\bxnpe(i)}}
\newcommand{\betaiss}{\beta^{\textup{\textsc{iss}}}}
\newcommand{\gammaiss}{\gamma^{\textup{\textsc{iss}}}}
\newcommand{\bz}{\boldsymbol{z}}
\newcommand{\srel}{S^{\mathrm{rel}}}
\newcommand{\ypatient}{y^{\mathrm{patient}}}
\newcommand{\yhatpatient}{\hat{y}^{\mathrm{patient}}}
\newcommand{\dperform}{d^{\mathrm{perform}}}
\newcommand{\dhatperform}{\hat{d}^{\mathrm{perform}}}
\newcommand{\dsuggest}{d^{\mathrm{suggest}}}
\newcommand{\dtarget}{d^{\mathrm{target}}}

\newcommand{\urobot}{\boldsymbol{u}^{\mathrm{robot}}}
\newcommand{\xhatb}[1]{\hat{x}^{\textup{\textsc{b}}}_{\mathrm{#1}}}
\newcommand{\xbarbcomfort}{\bar{x}^{\textup{\textsc{b}}}_{\mathrm{comfort}}}
\newcommand{\zperception}{Z^{\textup{\textsc{p}}}}
\newcommand{\zhatperception}{\tilde{Z}^{\textup{\textsc{p}}}}
\newcommand{\zcognition}{Z^{\textup{\textsc{c}}}}
\newcommand{\zhatcognition}{\tilde{Z}^{\textup{\textsc{c}}}}
\newcommand{\zdecision}{Z^{\textup{\textsc{d}}}}
\newcommand{\zhatdecision}{\tilde{Z}^{\textup{\textsc{d}}}}
\newcommand{\bza}{\boldsymbol{z}^{\textup{\textsc{a}}}}
\newcommand{\bxnpehat}{\hat{\bx}^{\textup{\textsc{lpe}}}}
\newcommand{\bxhat}{\hat{\bx}}
\newcommand{\bzahat}{\hat{\boldsymbol{z}}^{\textup{\textsc{a}}}}

\journal{Journal of Mathematical Psychology}

\begin{document}
\glsdisablehyper

\begin{frontmatter}



\title{A modular state-space model of human perception, cognition, and decision dynamics}


\author[1,2]{Sven Schoonebeek} 
\author[2]{Carlo Cenedese}
\author[1]{Anahita Jamshidnejad}

\affiliation[1]{organization={Department of Intelligent Systems, Faculty of Electrical Engineering, Mathematics, and Computer Science, TU Delft},
            addressline={Van Mourik Broekmanweg 6}, 
            city={Delft},
            postcode={2628 XE}, 
            state={Zuid-Holland},
            country={The Netherlands}}

\affiliation[2]{organization={Delft Center for Systems and Control, Faculty of Mechanical Engineering, TU Delft},
            addressline={Mekelweg 2}, 
            city={Delft},
            postcode={2628 CE}, 
            state={Zuid-Holland},
            country={The Netherlands}}

\begin{abstract}
Human-centered adaptive systems require behavioral models that are both psychologically interpretable and mathematically analyzable. Many existing predictors either operate as black-box input–output mappings or provide limited access to latent internal dynamics. This paper addresses this gap by modeling behavior as a perception–cognition–decision pipeline. We propose a modular state-space model in which attentional selection, predictive inference, cognitive-state evolution, intention formation, and action selection are represented by coupled mathematical mappings. The model links sensory inputs to observable behavior through latent internal states while retaining interpretable connections to neuro-cognitive mechanisms.

We establish sufficient conditions for boundedness, Lipschitz regularity, forward invariance, contraction of perceptual inference under constant input, and input-to-state stability of the cognitive state dynamics. Numerical sensitivity analyses show that the model yields interpretable changes in perceptual tracking, cognitive amplification, intention expression, and action decisiveness. We further demonstrate a closed-loop rehabilitation case study in which a receding-horizon controller uses the model to adapt movement difficulty from partial feedback. In this proof-of-concept setting, the model-based controller sustains simulated task participation and achieves lower realized cumulative cost than target-following and random baselines.

Overall, the framework provides a white-box dynamical structure for estimation, validation, and model-based control in human-centered settings.
\end{abstract}


\end{frontmatter}

\section{Introduction}
\noindent
Modern intelligent systems in healthcare, mobility, and education increasingly rely on models of human behavior to enable context-aware and personalized interventions. This creates a need for formal, cognitively grounded mathematical models that capture not only observable actions, but also the latent processes through which people perceive, reason, adapt, and decide over time.

Many existing approaches remain dominated by data-driven input–output predictors that infer preferences or actions directly from observed behavior \cite{b2}. Although such models can be effective in well-sampled domains, their reliance on large representative datasets limits generalization to changing environments and individual users \cite{b2,b12}. Moreover, black-box predictors often provide limited access to intermediate cognitive variables, making it difficult to explain, constrain, or control the latent mechanisms underlying behavior. In parallel, cognitive-science models have often emphasized qualitative explanation over mathematical analysis, partly because latent mental states are difficult to measure and formalize \cite{b12,b13}.

These limitations are especially important in adaptive human-centered systems that provide support, guidance, or intervention over time. In such settings, decisions should depend not only on observed behavior, but also on changing latent factors such as capacity, confidence, fatigue, motivation, and engagement. In post-stroke rehabilitation, for example, robotic or digital support should adapt task difficulty to patient state: assistance should increase when demand becomes excessive, but decrease when active engagement should be preserved \cite{b84}. Pure input–output predictors can struggle in such settings because clinically meaningful personalization depends on unobserved internal factors and their evolution over time. Recent machine-learning-based behavior analysis methods report similar limitations, including risks of bias that hinder reliable deployment at the individual level \cite{b83}.

Addressing these challenges calls for mathematical models that combine predictive utility with interpretable latent-state structure. Frameworks such as the \gls{bdi} paradigm \cite{b10,b11} describe how beliefs, emotions, goals, and intentions interact to give rise to decisions. Building on this perspective, we propose a modular state-space framework that decomposes human decision-making into perception, cognition, and decision-making processes, and represents latent cognitive variables as internal state variables. This formulation exposes intermediate states that can be estimated, constrained, and used in feedback control, while also enabling systems-theoretic analysis through sufficient conditions for boundedness, contraction, and \gls{iss} under specified parameter regimes \cite{b7,b9}. Table~\ref{tab:comparison_existing_models} summarizes the positioning of the proposed framework relative to representative model classes. Our contribution lies in the integrated formulation and analysis of these established computational mechanisms, rather than in any individual model component.

\paragraph{\textbf{Core contributions}}
This paper develops a principle-based dynamical model that links sensory inputs to observable actions through structured latent-state mappings. The main contributions are:
\begin{enumerate}[leftmargin=*]
\item We propose a modular state-space model that represents human behavior as a coupled perception–cognition–decision pipeline with explicit inputs, outputs, and latent internal states.
\item We establish sufficient conditions for boundedness, Lipschitz regularity, forward invariance, contraction of the perceptual inference dynamics under constant input, and \gls{iss} of the cognitive state dynamics.
\item We provide sensitivity analyses showing how model parameters regulate perceptual tracking, cognitive amplification and damping, intention expression, and action decisiveness.
\item We demonstrate how the full architecture can be embedded in a receding-horizon controller for an illustrative rehabilitation \gls{hri} scenario.
\end{enumerate}

\begin{table*}[!t]
\centering
\caption{Qualitative positioning of the proposed framework relative to representative model classes. Entries summarize common modeling and analytical emphases rather than exhaustive properties or all possible extensions. Personalization mechanism denotes the common means of fitting or updating a model for a particular agent or participant.}
\label{tab:comparison_existing_models}
\resizebox{\textwidth}{!}{%
\begin{tabular}{lllll}
\toprule
Model class
& Latent states
& Neuro-cognitive grounding
& Personalization mechanism
& Analytical emphasis \\
\midrule

Black-box predictors \cite{b8,b83}
& Implicit or absent
& Usually limited
& Data-driven fitting or adaptation
& Prediction and generalization \\

BDI architectures \cite{b105}
& Explicit symbolic
& Cognitive and agent-based
& Belief, goal, or rule updating
& Symbolic practical reasoning \\

Mathematical models of mind \cite{b10,b11,b102}
& Explicit latent states
& Cognitive or theory-of-mind based
& State and parameter estimation
& State-space modeling and identification \\

Predictive coding models \cite{b47,b49,b104}
& Latent predictions and errors
& Computational and neuro-cognitive
& Latent-state or precision updating
& Prediction-error and variational dynamics \\

Dynamic causal modelling \cite{b75,b76}
& Explicit neural states
& Biophysically and neurally grounded
& Subject-specific Bayesian inversion
& Model inversion and effective connectivity \\

Diffusion decision models \cite{b103}
& Explicit decision variable
& Cognitive decision-process grounding
& Individual or hierarchical fitting
& Accumulation and first-passage analysis \\

Proposed model
& Explicit latent state-space
& Neuro-cognitively inspired
& State and parameter estimation
& Boundedness, invariance, ISS, and control \\
\bottomrule
\end{tabular}%
}
\end{table*}

The paper is organized as follows. Section~\ref{sec:modulararchitecture} introduces the proposed modular architecture. Section~\ref{sec:mathreprhumandecisionmaking} develops the mathematical formulation of the perception, cognition, and decision-making modules. Section~\ref{sec:stabilityanalysis} establishes analytical properties of the model. Section~\ref{sec:simandanalysis} presents parameter sensitivity analyses and the closed-loop rehabilitation showcase.

\section{Modular architecture}
\label{sec:modulararchitecture}
\noindent
Modularity is a central organizing principle in both neuroscience and cognitive architecture. Neural systems process information through specialized and interacting regions, supporting parallel processing, adaptability, and efficient wiring \cite{b70}. Cognitive architectures similarly describe cognition as the coordination of specialized components through structured integration and control \cite{b68}. Motivated by these perspectives, we decompose human behavior into three interacting modules: perception, cognition, and decision-making.

In the proposed architecture, the perception module processes sensory input and forms perceptual estimates, the cognition module updates latent internal states such as beliefs, goals, emotions, and biases, and the decision-making module maps these states to behavioral responses. Each module has explicit inputs, outputs, and internal states, and the modules interact with the environment in a closed loop from sensory observations to action. Figure~\ref{fig:modulararchitecture} summarizes this perception–cognition–decision pipeline.

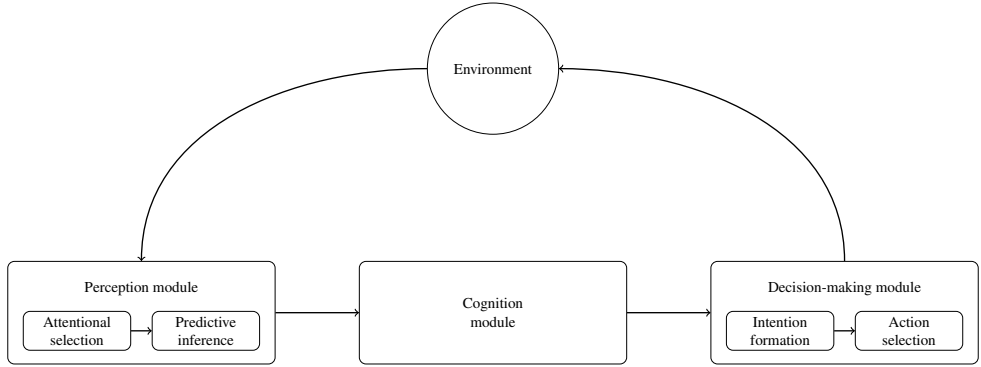
\begin{figure*}[!t]
    \centering
    \resizebox{.7\linewidth}{!}{%
    \begin{tikzpicture}[
        font=\small,
        module/.style={
            draw,
            rounded corners,
            align=center,
            minimum width=5.60cm,
            minimum height=2.15cm,
            inner sep=6pt
        },
        layer/.style={
            draw,
            rounded corners,
            align=center,
            minimum width=2.25cm,
            minimum height=0.62cm,
            inner sep=3pt
        },
        env/.style={
            circle,
            draw,
            align=center,
            minimum size=2.75cm,
            inner sep=4pt
        },
        arrow/.style={->, thick}
    ]

    \node[module] (perception) {};
    \node[module, right=1.75cm of perception] (cognition)
    {Cognition\\module};
    \node[module, right=1.75cm of cognition] (decision) {};

    \node[env, above=2.65cm of cognition] (environment)
    {Environment};

    \node[above=0.25cm of perception.center] (perceptiontitle)
    {Perception module};

    \node[layer] (asl) at ([xshift=-1.35cm,yshift=-0.38cm]perception.center)
    {Attentional\\selection};

    \node[layer] (pi) at ([xshift=1.35cm,yshift=-0.38cm]perception.center)
    {Predictive\\inference};

    \draw[arrow] (asl) -- (pi);

    \node[above=0.25cm of decision.center] (decisiontitle)
    {Decision-making module};

    \node[layer] (iflayer) at ([xshift=-1.35cm,yshift=-0.38cm]decision.center)
    {Intention\\formation};

    \node[layer] (aslayer) at ([xshift=1.35cm,yshift=-0.38cm]decision.center)
    {Action\\selection};

    \draw[arrow] (iflayer) -- (aslayer);

    \draw[arrow] (perception.east) -- (cognition.west);

    \draw[arrow] (cognition.east) -- (decision.west);

    \draw[arrow]
    (environment.west) to[out=180,in=90] (perception.north);

    \draw[arrow]
    (decision.north) to[out=90,in=0] (environment.east);

    \end{tikzpicture}%
    }
    \caption{High-level modular architecture of the proposed closed-loop 
    human behavior model as a perception-cognition-decision pipeline. 
    The environment provides sensory observations to the perception module, whose outputs are propagated through the cognition module to the decision-making module. The resulting behavioral output acts on the environment, closing the cycle of perception, cognition, and decision-making.}
    \label{fig:modulararchitecture}
\end{figure*}

We next summarize the role of each module and the neuro-cognitive motivation for its mathematical structure.

\subsection{Perception module}
\label{sec:perceptionmoduleneurotheory}
\noindent
Perception is modeled as a two-stage process consisting of attentional selection and predictive inference. The first stage determines which sensory cues are most relevant at the current time, whereas the second stage integrates these selected cues with prior expectations to form a perceptual estimate. This decomposition is motivated by perceptual hierarchies in which early sensory encoding and feature extraction precede context-dependent interpretation through recurrent feedforward and feedback interactions \cite{b47,b65,b72}. It is also consistent with two-stage models of attention, which distinguish initial feature registration from later focused integration of selected features into coherent object representations \cite{b73,b74}.

Sensory input is organized into channels, where each channel represents one cue type, such as visual intensity, motion, auditory salience, or task-relevant context. Attentional selection assigns a non-negative weight to each channel, indicating how strongly that cue influences downstream inference. We model these weights using divisive normalization, a standard computational form in sensory competition and attention \cite{b61,b62,b63}. This captures saturation and contextual competition without claiming a direct neural implementation \cite{b20,b61}.

Predictive inference maps the selected sensory information to a \gls{npe}, interpreted as a perceptual estimate of the causes of sensory inputs. This estimate evolves recursively as new evidence arrives \cite{b81}. Consistent with predictive coding, it is updated by prediction errors, i.e., mismatches between selected sensory evidence and the previous perceptual estimate \cite{b49}. The update size depends on sensory precision, prior precision, and attention gain, representing confidence in current evidence, confidence in the previous estimate, and the emphasis placed on the corresponding cue, respectively \cite{b47,b48,b49}. Together, these factors determine how strongly prediction errors shift the \gls{npe} \cite{b46,b49}.

\subsection{Cognition module}
\noindent
The cognition module describes how perceptual estimates influence interacting latent cognitive variables, such as beliefs, goals, emotions, and biases. Its mathematical structure is inspired by \gls{dcm}, which represents distributed neural dynamics as nonlinear input–state–output processes with directed effective couplings among state variables \cite{b75,b76}.

Although \gls{dcm} was originally developed for neural populations and biophysically grounded neural-mass models, we use its state-space coupling structure at a higher level of abstraction. Specifically, each state variable represents a latent cognitive quantity rather than a neural population. This abstraction is motivated by network-level accounts in which beliefs, goals, and emotions emerge from distributed interacting systems whose directed influences can be summarized at the network level \cite{b75,b78,b79,b80}. The resulting \gls{dcm}-inspired terms provide a template for modeling directed interactions among latent cognitive states.

\subsection{Decision-making module}
\label{sec:decisionmakingmoduleneurotheory}
\noindent
Decision-making is modeled as two sequential stages: intention formation and action selection \cite{b16,b18,b65}. Intention formation integrates goal- and belief-related cognitive states into a prospective behavioral tendency. We define intention as a value-weighted expression of goal salience gated by beliefs about feasibility, desirability, and contextual constraints. This is broadly consistent with accounts in which intentions depend on expected outcomes and perceived behavioral control \cite{b29,b38,b54,b55}.

Action selection then maps intentions to behavioral outputs through competition among candidate actions \cite{b18,b65}. This abstraction is motivated by cortico-basal-ganglia accounts in which selected actions are facilitated while competitors are suppressed, and conflict can increase the threshold for commitment \cite{b57,b58,b59,b60,b98,b99}. These mechanisms motivate the excitatory, inhibitory, and threshold components of the action-selection mapping introduced below.

We next formulate these modules as an interconnected discrete-time state-space model. The following section defines the admissible sets, module mappings, and internal state updates that propagate sensory observations to perceptual estimates, latent cognitive states, intentions, and actions.

\section{Mathematical modeling of human decision-making}
\label{sec:mathreprhumandecisionmaking}
\noindent
This section introduces the notation and then formulates the perception, cognition, and decision-making modules as an interconnected discrete-time state-space model.

\paragraph{\textbf{Notation}}
Superscripts indicate the associated process or concept, and boldface denotes vector-valued quantities. Latin italic letters denote dynamic signals: $u$ for inputs, $z$ for auxiliary variables, $x$ for state variables, and $y$ for outputs. State variables evolve recursively from their past values, whereas auxiliary variables are instantaneous mappings.

Sets are denoted by calligraphic letters, e.g., $\mathcal X$ for an admissible state set. For $\bx\in\mathbb R^n$, $x_i$ denotes the $i$-th component; for $A\in\mathbb R^{n\times n}$, $A_{im}$ denotes the $(i,m)$-th entry. We write $I_n$ for the identity matrix, omitting dimensions when clear from the context. For an index set $\mathcal S$, $[x_j]_{j\in\mathcal S}$ denotes the vector collecting the corresponding components in a fixed order. The symbol $\odot$ denotes the Hadamard product. Unless stated otherwise, vector norms are Euclidean norms and matrix norms are induced spectral norms.

For reference, Table~\ref{tab:variables} summarizes the main variables, associated equations, and admissible domains used in the model formulation.

\subsection{Perception module}
\noindent
We formalize perception as two coupled stages: attentional selection, which maps sensory inputs to bounded cue weights, and predictive inference, which updates latent perceptual estimates from weighted prediction errors.

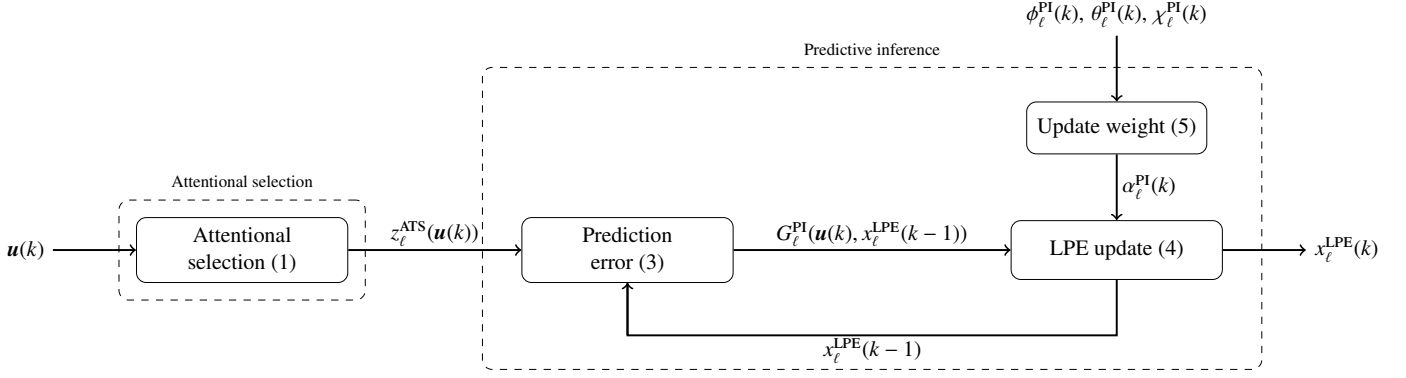
\begin{figure*}[!t]
    \centering
    \resizebox{\linewidth}{!}{%
    \begin{tikzpicture}[
    font=\small,
    node distance=1.15cm and 1.25cm,
    block/.style={
    draw,
    rounded corners,
    align=center,
    minimum width=3.05cm,
    minimum height=0.85cm,
    inner sep=4pt
    },
    smallblock/.style={
    draw,
    rounded corners,
    align=center,
    minimum width=2.60cm,
    minimum height=0.75cm,
    inner sep=4pt
    },
    arrow/.style={->, thick},
    plain/.style={thick},
    dashedbox/.style={
    draw,
    dashed,
    rounded corners,
    inner sep=7pt
    }
    ]
    
    \node (input)
    {$\bu(k)$};
    
    \node[block, right=1.2cm of input] (asl)
    {Attentional\\selection \eqref{eqn:perceptualaccessgeneral}};
    
    \node[block, right=2.5cm of asl] (pe)
    {Prediction\\error \eqref{eqn:predictiveinferenceg}};
    
    \node[block, right=4cm of pe] (update)
    {LPE update \eqref{eq:NPE_recursion}};
    
    \node[smallblock, above=0.95cm of update] (alpha)
    {Update weight \eqref{eqn:predictiveinferencealpha}};
    
    \node[above=0.95cm of alpha] (precision)
    {$\phipil(k),\,\thetapil(k),\,\chipil(k)$};
    
    \node[right=1.2cm of update] (output)
    {$\xnpel(k)$};
    
    \draw[arrow] (input) -- (asl);
    
    \draw[arrow] (asl) -- node[above, fill=white, inner sep=1pt]
    {$\zasl(\bu(k))$}
    (pe);
    
    \draw[arrow] (pe) -- node[above, fill=white, inner sep=1pt]
    {$\gpil(\bu(k),\xnpel(k-1))$}
    (update);
    
    \draw[arrow] (update) -- (output);
    
    \draw[arrow] (precision) -- (alpha);
    
    \draw[arrow] (alpha) -- node[right, fill=white, inner sep=2pt]
    {$\alphapil(k)$}
    (update);
    
    \draw[plain] 
    (update.south) -- ++(0,-0.80)
    coordinate (npefeedbackbottom)
    -| coordinate[pos=0.85] (npefeedbackleft)
    (pe.south);
    
    \node[below, inner sep=1pt] at
    ($(npefeedbackbottom)!0.50!(npefeedbackbottom -| pe.south)$)
    {$\xnpel(k-1)$};
    
    \draw[arrow]
    ([yshift=-0.75cm]pe.south)
    --
    ([yshift=-0cm]pe.south);
    
    \node[dashedbox, fit=(asl)] (asldash) {};
    \node[above=0.05cm of asldash] {\scriptsize Attentional selection};
    
    \node[
    dashedbox,
    inner xsep=16pt,
    inner ysep=14pt,
    fit=(pe)(update)(alpha)(npefeedbackbottom)(npefeedbackleft)
    ] (pidash) {};
    \node[above=0.05cm of pidash] {\scriptsize Predictive inference};
    
    \end{tikzpicture}%
    }
    \caption{
    Block diagram of the perception module.
    }
    \label{fig:perception_module_block_diagram}
\end{figure*}

\subsubsection{Attentional selection}
\noindent
Let $\mathcal L$ be a finite set of perceptual channels, where each channel $\ell\in\mathcal L$ represents a distinct sensory or contextual cue. At time step $k$, channel $\ell$ provides a scalar feature score $\ul(k)\in\ulset\subseteq\mathbb R$, such as an intensity, salience, or confidence value obtained from upstream preprocessing. The stacked sensory input is denoted by $\bu(k)\in\uset\coloneqq\prod_{\ell\in\mathcal L}\ulset$, where $\prod_{\ell\in\mathcal L}$ denotes the Cartesian product. Each admissible set $\uset_\ell$ is compact and convex (implying the same for $\uset$), and represents the finite operating range of the feature score extracted for channel $\ell$. The range reflects sensory saturation and normalization applied during upstream pre-processing \cite{b61}. We assume that $\mathbf{0}\in\uset$, with the zero input vector representing the absence of any stimulus drive.

\begin{definition}
\label{def:perceptualaccess}
    Suppose that the stacked input to the perceptual model at time step $k$ is $\bu(k) \in \uset$. To each channel-wise input $\ul(k) \in \ulset$, we assign an attentional selection weight via the mapping $\zasl: \uset \to \left[ 0, \zaslmax \right]$ defined by
    
    \begin{align}
    \label{eqn:perceptualaccessgeneral}
        \zasl \left(\bu(k)\right) \coloneqq \dfrac{\fasl \left(\ul(k)\right)}{\gasl \left( \bu(k) \right)},
    \end{align}
    with $\fasl: \ulset \to \left[0, \faslmax\right]$,
    $\gasl: \uset \to \left[\gaslmin, \infty\right)$, $\zaslmax = \faslmax/\gaslmin$. The contextual normalization function $\gasl(\cdot)$ is defined via
    
    \begin{align}
    \label{eqn:attentionalselectionG}
        \gasl \left( \bu(k) \right) \coloneqq \betaasl + \sum_{m \in \mathcal{L}} \gammaaslm \gasmtilde \left( \um(k) \right),
    \end{align}
    where $\betaasl \in \left(0, \infty \right)$ is a semi-saturation constant, $\gasmtilde : \umset \to [0, \infty)$ for $m \in \mathcal{L}$, and $\gammaaslm \in [0, \infty)$ is a pooling weight. \hfill $\square$
\end{definition}

The functions and constants in Definition~\ref{def:perceptualaccess} are assumed to have the following properties:

\begin{enumerate}[leftmargin=*]
    \item \textbf{Boundedness of $\fasl(\cdot)$}: The channel-wise attentional drive function $\fasl(\cdot)$ is assumed to be non-negative and uniformly upper-bounded. Non-negativity of $\fasl(\cdot)$ is a natural choice, as it represents a gain that scales the contribution of channel $\ell$ to perceptual access. 
    Negative values invert the sign of the drive and thereby implement active suppression, a role that is already captured through the divisive normalization term (cf.~\eqref{eqn:perceptualaccessgeneral}--\eqref{eqn:attentionalselectionG}). Mathematically, boundedness prevents unbounded growth of the numerator in \eqref{eqn:perceptualaccessgeneral}. Finally, we require $\fasl(0) = 0$ to ensure that zero input produces no attentional drive.
    
    \item \textbf{Boundedness of $\gasl(\cdot)$}: The contextual normalization function $\gasl(\cdot)$ in \eqref{eqn:attentionalselectionG} implements divisive normalization in \eqref{eqn:perceptualaccessgeneral}. Here, $\gasmtilde$ is a non-negative pooling function that maps the $m$-th sensory channel into its contribution to the normalization pool, and $\gammaaslm$ weights the contribution of channel $m$ to the normalization term for channel $\ell$. Similar channels typically exert greater suppression on each other; for instance, if two channels both encode luminance-related features (e.g., raw intensity and local contrast), they may suppress each other more strongly than channels belonging to different feature classes (e.g., luminance versus auditory pitch) \cite{b82}. 
    The semi-saturation constant $\betaasl$ ensures a uniform positive lower bound on $\gasl(\cdot)$ over $\uset$, so the denominator in \eqref{eqn:perceptualaccessgeneral} never vanishes. These assumptions guarantee that $\zasl(\cdot)$ is well-defined on $\uset$. 

\end{enumerate}

\noindent
The ratio in \eqref{eqn:perceptualaccessgeneral} defines the attentional selection weight $\zasl(\cdot)$ for channel $\ell$ at time step $k$, capturing both local stimulus strength and global contextual competition. The collection of channel-wise attentional selection weights is used as the input to the predictive inference stage described next.

\glsreset{npe}
\subsubsection{Predictive inference}
\noindent
The second stage models predictive inference as a bounded error-driven recursion, where the prediction error is the discrepancy between the attentional selection weight and the previous \gls{npe} (see Figure~\ref{fig:perception_module_block_diagram}).

\begin{definition}
\label{def:predictiveinference}
    Suppose that the attentional selection stage produces an attentional selection weight $\zasl\left(\bu(k)\right) \in \left[ 0, \zaslmax \right]$ 
    per channel $\ell\in\mathcal{L}$ at time step $k$. Let the \gls{npe} for channel $\ell$ be a state variable denoted by $\xnpel(k) \in \xnpelset \subseteq \mathbb{R}$, where the admissible set $\xnpelset$ is specified below. Define the prediction error function $\gpil : \uset \times \xnpelset \to \mathbb{R}$ of channel $\ell$ as
    
    \begin{align}
    \label{eqn:predictiveinferenceg}
        \gpil \left( \bu(k), \xnpel(k - 1) \right) \coloneqq \zasl \left( \bu(k) \right) - \xnpel (k - 1).
    \end{align}
    
    The \gls{npe} update is given by the recursion
    \begin{align}
    \label{eq:NPE_recursion}
        \xnpel(k) = \xnpel(k - 1) + \alphapil(k) \; \gpil \left( \bu(k), \xnpel(k - 1) \right).
    \end{align} 
    
    \hfill $\square$
\end{definition}

The next definition specifies how the effective update weight $\alphapil(k)$ of the \gls{npe} is computed from sensory precision, prior precision, and attention gain.

\begin{definition}
\label{def:predictiveinference_alpha}
    For each channel $\ell\in\mathcal{L}$ at time step $k$, let the sensory precision $\phipil(k) \in (0,\infty)$, prior precision $\thetapil(k) \in (0,\infty)$, and attention gain $\chipil(k) \in [0,\infty)$ be given. The effective update weight at time step $k$ is given by
    
    \begin{align}
    \label{eqn:predictiveinferencealpha}
        \alphapil(k) = \fpil \left( \phipil(k), \thetapil(k), \chipil(k) \right), 
    \end{align}
    where $\fpil : (0,\infty)^2\times[0,\infty) \to [0,1)$. \hfill $\square$
\end{definition}

The quantities $\phipil(k)$ and $\thetapil(k)$ denote sensory and prior precision, respectively. As precisions represent inverse uncertainty, they are required to be strictly positive \cite{b48,b49}. The attention gain $\chipil(k)$ is assumed non-negative, as it modulates the impact of prediction errors without reversing their sign \cite{b47,b49}. Furthermore, $\fpil(\cdot)$ is designed to satisfy:

\begin{enumerate}[leftmargin=*]
    \item \textbf{Boundedness:} The co-domain of $\fpil(\cdot)$ is restricted to $[0,1)$, so that the effective update weight satisfies 
    $\alphapil(k)<1$ for all $k$, and the recursion in \eqref{eq:NPE_recursion} can be re-written as a convex-combination update. Such bounded learning-rate assumptions are also consistent with predictive coding formulations, where precision-weighted prediction-error updates are regulated to avoid unstable perceptual inference \cite{b47}.
    
    \item \textbf{Monotonicity:} Function $\fpil(\cdot)$ is non-decreasing in $\phipil(k)$ and $\chipil(k)$, and non-increasing in $\thetapil(k)$. This captures the empirically supported role of neuro-modulators in precision-weighted prediction error processing, i.e., increased sensory precision or attention gain amplifies learning, whereas stronger prior precision dampens estimate updating \cite{b45,b46,b47}.
    
    \item \textbf{Vanishing conditions:}
    Conditions \eqref{eq:perception_vanishing_1}--\eqref{eq:perception_vanishing_3} reflect the fact that inference halts whenever (i) attention is absent, (ii) sensory input is unreliable, or (iii) prior estimates dominate. These are consistent with predictive coding and empirical findings on confidence-weighted \gls{npe} updating \cite{b45,b46,b48}:
    
    \begin{subequations}
    \label{eq:perception_vanishing}
        \begin{align}
        \label{eq:perception_vanishing_1}
           &\fpil \left( \phipil(k), \thetapil(k), 0 \right) = 0, \\
        \label{eq:perception_vanishing_2}
            &\lim_{\phipil(k) \to 0^+} \fpil \left( \phipil(k), \thetapil(k), \chipil(k) \right) = 0, \\
        \label{eq:perception_vanishing_3}
            &\lim_{\thetapil(k) \to \infty} \fpil \left( \phipil(k), \thetapil(k), \chipil(k) \right) = 0.
        \end{align}
    \end{subequations}
    
\end{enumerate}

\begin{remark}
    \label{remark:admissiblesetchoicexnpel}
    \textbf{Choice of $\xnpelset$.}
    Since $\alphapil(k)\in[0,1)$ and $\zasl(\bu(k))\in[0,\zaslmax]$, the update \eqref{eq:NPE_recursion} can be written as
    \begin{align}
    \label{eqn:nperecursionconvex}
    \xnpel(k) = (1-\alphapil(k)) \xnpel(k-1) + \alphapil(k) \zasl(\bu(k)).
    \end{align}
    Hence, we choose $\xnpelset\coloneqq[0,\zaslmax]$. Forward invariance of this set is proved in Lemma~\ref{lem:npe_forward_invariance}. \hfill $\square$
\end{remark}

The predictive inference stage produces the vector
$\bxnpe(k)\in\xnpeset\coloneqq\prod_{\ell\in\mathcal L}\xnpelset$ of channel-wise \glspl{npe} at each time step,
with input–output relation
$\bxnpe(k)\coloneqq\zperception(\bu(k),\bxnpe(k-1))$.
This vector summarizes the current latent perceptual estimate and serves as input to the cognition module. 

\begin{figure*}[!t]
    \centering
    \resizebox{.7\linewidth}{!}{%
    \begin{tikzpicture}[
    font=\small,
    node distance=1.15cm and 1.25cm,
    block/.style={
    draw,
    rounded corners,
    align=center,
    minimum width=3.35cm,
    minimum height=0.95cm,
    inner sep=5pt
    },
    smallblock/.style={
    draw,
    rounded corners,
    align=center,
    minimum width=2.75cm,
    minimum height=0.75cm,
    inner sep=4pt
    },
    arrow/.style={->, thick},
    dashedbox/.style={
    draw,
    dashed,
    rounded corners,
    inner xsep=16pt,
    inner ysep=14pt
    }
    ]
    
    \node[block] (self)
    {Self-inhibition \eqref{eqn:cognitionmodulefiic}\\Parameters: $\kappa,\,\gamma_i,\,\Lambda$};
    
    \node[left=2.15cm of self] (npe)
    {$\left[ \xnpel(k)\right]_{\ell \in \mathcal{L}}$};
    
    \node[block, below=1.25cm of self] (drive)
    {Additive perceptual drive\\Parameter: $\Theta$};
    
    \node[block, below=1.25cm of drive] (coupling)
    {Cross-coupling \eqref{eqn:gimc}\\Parameters: $\Phi,\,\{\Psi_\ell\},\,\{\Xi_q\}$};
    
    \node[left=2.02cm of drive] (npedrive)
    {$\left[ \xnpel(k)\right]_{\ell \in \mathcal{L}}$};
    
    \node[left=2cm of coupling, yshift=-0.0cm] (npecoupling)
    {$\left[ \xnpel(k)\right]_{\ell \in \mathcal{L}}$};
    
    \node[
    block,
    right=2.2cm of drive,
    minimum height=1.5cm
    ] (update)
    {Cognitive state update \eqref{eqn:dcmnonlinear}};
    
    \node[right=1.50cm of update] (output)
    {$\bx(k+1)$};
    
    \draw[arrow] (npe) -- (self);
    
    \draw[arrow] (npedrive) -- (drive);
    
    \draw[arrow] (npecoupling.east) -- ++(0.55,0) -- 
    ([yshift=-0.0cm]coupling.west);
    
    \draw[arrow] (self.east) -- ++(0.99,0) |- 
    ([yshift=0.6cm]update.west);
    
    \node[right=0.15cm] at ([xshift=0.85cm,yshift=-0.8cm]self.east)
    {$\fiic(\bxnpe(k))$};
    
    \draw[arrow] (drive.east) -- (update.west);
    
    \node[above=-0.06cm] at ($(drive.east)!0.5!(update.west)$)
    {$\Theta_{i\ell}\xnpel(k)$};
    
    \draw[arrow] (coupling.east) -- ++(0.85,0) |- 
    ([yshift=-0.6cm]update.west);
    
    \node[right=0.05cm] at ([xshift=0.85cm,yshift=0.8cm]coupling.east)
    {$\gimc(\bx(k),\bxnpe(k))$};
    
    \draw[arrow] (update) -- (output);
    
    \draw[arrow]
    (update.south) -- ++(0,-2.70)
    coordinate (feedbackbottom)
    -| coordinate[pos=0.85] (feedbackleft)
    (coupling.south);
    
    \node[below, fill=white, inner sep=1pt] at
    ($(feedbackbottom)!0.55!(feedbackbottom -| coupling.south)$)
    {$\bx(k)$};
    
    \node[
    dashedbox,
    fit=(self)(drive)(coupling)(update)(feedbackbottom)(feedbackleft)
    ] (cognitiondash) {};
    \node[above=0.05cm of cognitiondash] {\scriptsize Cognition module};
    
    \end{tikzpicture}%
    }
    \caption{
    Block diagram of the cognition module.
    }
    \label{fig:cognition_module_block_diagram}
\end{figure*}
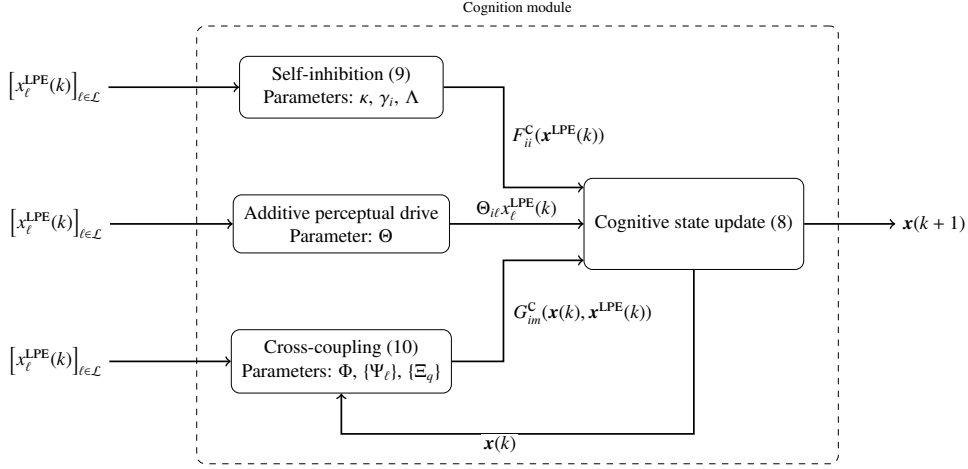

\subsection{Cognition module}
\label{sec:cognitionmodule}
\noindent
The cognition module transforms latent perceptual estimates into interacting cognitive states. Let $\mathcal I$ be a finite index set of cognitive variables (e.g., representations of emotions, beliefs, or goals), and let $i \in \mathcal{I}$ be an index (not necessarily numeric) for an individual cognitive variable. Each cognitive variable at time step $k$ is represented by a scalar intensity $x_i(k) \in \mathcal{X}_i \subseteq \mathbb{R}$, where $\mathcal{X}_i$ denotes the admissible range of state $i$. Section~\ref{sec:stabilityanalysis} derives sufficient conditions under which cognitive state trajectories remain in a bounded operating set. We stack all cognitive variables via the vector $\bx(k) \in \mathcal{X}$, where $\mathcal{X}\coloneqq \prod_{i \in \mathcal{I}} \mathcal{X}_i$.

\begin{remark}\textbf{Latent state interpretation.}
    The labels assigned to these latent states indicate their intended psychological interpretation rather than properties implied by the state equations. In empirical applications, this interpretation is supported by construct-relevant measurements, targeted manipulations, and identifiability analysis, since different parameterizations may yield similar observable behavior.
\end{remark}

\bigskip
\noindent
Let $\card{\mathcal{I}}$ and $\card{\mathcal{L}}$ represent the cardinalities of the sets $\mathcal{I}$ and $\mathcal{L}$. 
Starting from the nonlinear, continuous-time \gls{dcm}-style ordinary differential equation in \cite{b76} and using 
a forward-Euler discretization with fixed step size $\Delta t>0$, we obtain the following discrete-time update model for the cognitive 
states: 

\begin{align}
\label{eqn:dcmnonlinear}
    x_i(k+1) = & \Bigg(1 + \Delta t\; \fiic \left( \bxnpe(k)\right)\Bigg)\; x_i(k) + \\ 
    &\Delta t \sum_{\substack{m \in \mathcal{I}\\ m\neq i}} \gimc \bigg(\bx(k), \bxnpe(k)\bigg)\; x_m(k) + \nonumber\\ 
    &\Delta t\sum_{\ell \in \mathcal{L}} \Theta_{i\ell} \xnpel(k). \nonumber
\end{align}

The state update equation \eqref{eqn:dcmnonlinear} consists of an input-modulated self-term governed by $\fiic(\cdot)$, cross-state couplings governed by $\gimc(\cdot, \cdot)$, and an additive perceptual drive mediated by $\Theta$. A detailed description and admissible parameterization of these terms follow below. The cognition module is visualized in Figure~\ref{fig:cognition_module_block_diagram}.

\subsubsection{Self-inhibition}
\noindent
The first term on the right-hand side of \eqref{eqn:dcmnonlinear} describes the intrinsic evolution of the $i$-th state in the absence of cross-state interactions: it consists of identity carry-over and input-modulated self-dynamics governed by $\fiic :\xnpeset\to\mathbb{R}$. Collecting these self-effects across all states yields the diagonal matrix $\fc (\cdot)\in\mathbb{R}^{\card{\mathcal{I}}\times\card{\mathcal{I}}}$. 
If $\fiic (\cdot)<0$, the self-term acts as a leak, tending to attenuate $x_i(k)$ in the absence of other terms. In particular, we use an exponential link for the self-inhibition term to guarantee $\fiic(\cdot) < 0$ on $\xnpeset$ without imposing hard inequality constraints during estimation. The resulting parameterization is stated next.

\begin{definition}
\label{def:cognitionmoduleselfinhibition}
    Let $\bxnpe(k) \in \xnpeset \subseteq \mathbb{R}^{|\mathcal{L}|}$ be the exogenous perceptual input at time step $k$. For each $i \in \mathcal{I}$, we parameterize the self-inhibition as follows:
    
    \begin{align}
        \label{eqn:cognitionmodulefiic}
            &\fiic \left(\bxnpe(k)\right) \coloneqq -\kappa \exp\left(\gamma_i+\sum_{\ell\in\mathcal{L}}\xnpel(k)\Lambda_{i\ell}\right),
    \end{align}
    with $\kappa \in (0, \infty)$ a scalar gain, $\gamma_i$ the $i$-th element of $\boldsymbol{\gamma} \in \mathbb{R}^{\card{\mathcal{I}}}$, which sets the baseline log-leak, and $\Lambda_{i \ell}$ the $(i,\ell)$-entry of $\Lambda \in \mathbb{R}^{\card{\mathcal{I}}\times \card{\mathcal{L}}}$, which specifies how perceptual channel $\ell$ modulates the log-leak for the $i$-th state. \hfill $\square$
\end{definition}

Furthermore, we assume there exist finite constants $\overline{\boldsymbol{\gamma}}$ and $\overline{\Lambda_{\ell}}$ such that
$\norm{\boldsymbol{\gamma}}_\infty \leq \overline{\boldsymbol{\gamma}}$ and $\abs{\Lambda_{i\ell}} \leq \overline{\Lambda_\ell}$ for all $i \in \mathcal{I}$ and $\ell \in \mathcal{L}$. These assumptions, together with the boundedness of $\xnpeset$, yield uniform bounds and local Lipschitz constants for $\fiic (\cdot)$ over $\xnpeset$. 

\subsubsection{Cross-coupling}
\noindent
The cross-coupling, modeled via the second term on the right-hand side of \eqref{eqn:dcmnonlinear}, quantifies the directed influence of all other cognitive states $x_m(k)$ (with $m\neq i$) on cognitive state $x_i(k)$ at time step $k$. This influence is input and state-modulated through the entry-wise coupling coefficient $\gimc : \mathcal{X} \times \xnpeset \to \mathbb{R}$ defined by:

\begin{align}
\label{eqn:gimc}
    \gimc \left(\bx(k), \bxnpe(k)\right) \coloneqq  &\Phi_{im} + \\
    &\sum_{\ell \in \mathcal{L}} \xnpel(k)\; \Psi_{\ell,im}
    + \nonumber \\
    &\sum_{q \in \mathcal{I}} x_q(k)\; \Xi_{q,im}  \nonumber.
\end{align}

The structure of $\gimc(\cdot, \cdot)$ mirrors the canonical \gls{dcm} decomposition into a baseline coupling term, $\Phi_{im}$, a channel-wise input-gated perceptual modulation 
via $\Psi_{\ell,im}$, and a state-gated endogenous modulation via $\Xi_{q,im}$. 
We impose that $\Phi$, $\Psi_\ell$, and $\Xi_q$  have zero diagonal entries, so that self-dynamics arises solely from the identity carry-over and the diagonal self-inhibition matrix $\fc(\cdot)$. We next elaborate the interpretation of each component in \eqref{eqn:gimc}.

\paragraph{\textbf{Baseline cross-state influence}} 
Baseline interactions are encoded by a matrix $\Phi\in\mathbb{R}^{\card{\mathcal{I}}\times\card{\mathcal{I}}}$ (cf. \eqref{eqn:gimc}), whose entries $\Phi_{im}$ represent directed influences from state $x_m(k)$ to state $x_i(k)$ for $m\neq i$ \cite{b77}. Matrix $\Phi$ represents directed influences among cognitive variables that persist in the absence of perceptual or endogenous modulation. For instance, it may encode long-standing associations among beliefs or emotions, thereby providing the structural backbone on which input-gated and state-gated modulations act.

\paragraph{\textbf{Perceptual (input-gated) modulation}} Matrix $\Psi_\ell \in \mathbb{R}^{\card{\mathcal{I}} \times \card{\mathcal{I}}}$ encodes how the \gls{npe} component corresponding to channel $\ell$ modulates the directed influence from state $x_m(k)$ to state $x_i(k)$ (with $m \neq i$) through the bilinear input-state interaction $\xnpel(k)\Psi_{\ell,im}x_m(k)$ (cf. \eqref{eqn:dcmnonlinear} and \eqref{eqn:gimc}). 
This modulation specifies how new perceptual inputs strengthen or weaken the baseline coupling among the cognitive variables.

\paragraph{\textbf{Endogenous (state-gated) modulation}} Higher-order state-dependent modulation is captured via matrices $\Xi_q \in \mathbb{R}^{\card{\mathcal{I}} \times \card{\mathcal{I}}}$. In \eqref{eqn:dcmnonlinear}, this endogenous modulation enters as the bilinear term $x_q(k)\Xi_{q,im}x_m(k)$ and corresponds to entry $(i,m)$ of matrix $\Xi_q$ scaled by the current values of states $x_q(k)$ (i.e., the gating state) and $x_m(k)$ (i.e., the sender state). This term specifies how state $x_q(k)$ gates the directed influence from the $m$-th state to the $i$-th state (with $m \neq i$), consistent with nonlinear expansions of causal dynamics in \gls{dcm} formulations \cite{b76}. For example, heightened arousal may amplify the effect of perceived threat on fear.

\subsubsection{Additive perceptual drive}
\label{sec:additiveperceptualdrive}
\noindent
The last term in \eqref{eqn:dcmnonlinear} provides an additive driving input to state $x_i(k)$ from exogenous (perceptual) inputs $\bxnpe(k)$. It injects perceptual evidence into the cognition module independently of the current state $\bx(k)$ and enters the update linearly in $\bxnpe(k)$ and additively in the state equation. Matrix $\Theta \in \mathbb{R}^{\card{\mathcal{I}} \times \card{\mathcal{L}}}$ encodes which perceptual channels (or which components of $\bxnpe(k)$) directly drive which states. This allows only selected perceptual channels to directly influence specific latent cognitive variables. Additive driving terms are standard in \gls{dcm}-based dynamic formulations, where external inputs enter independently of the state and shape its evolution.

\subsubsection{Parameterization of cross-coupling and drive}
\noindent
To parameterize the coupling terms and perceptual drive in an interpretable and structurally constrained way, we impose an admissible construction with two design goals: (i) separating interaction structure from parameter values via binary adjacency masks, and (ii) enforcing sign and range constraints via smooth link functions applied entry-wise \cite{b77}. Specifically, we introduce unconstrained (free) parameter matrices whose entries can take arbitrary real values during estimation, and map them into admissible coupling and drive matrices through the masks and link functions. The resulting construction of $\Phi$, $\left\{\Psi_\ell | \ell \in \mathcal{L}\right\}$, $\left\{\Xi_q | q \in \mathcal{I} \right\}$, and $\Theta$ is given formally below.

\begin{definition}
\label{def:cognitionmodulecrosscoupling}
    Let $\tilde{\Phi}, \tilde{\Psi}_{\ell}, \tilde{\Xi}_q \in \mathbb{R}^{\card{\mathcal{I}} \times \card{\mathcal{I}}}$, and $\tilde{\Theta} \in \mathbb{R}^{\card{\mathcal{I}} \times \card{\mathcal{L}}}$ be unconstrained parameter matrices. Let $g_1, g_3: \mathbb{R} \to \mathbb{R}$ and, for each $\ell \in \mathcal{L}$ and $q \in \mathcal{I}$, let $g_{2,\ell}, g_{4,q} : \mathbb{R} \to \mathbb{R}$ be scalar link functions, and let $G_1, \{ G_{2,\ell} | \ell \in \mathcal{L} \}, G_3, \{ G_{4,q} | q \in \mathcal{I} \}$ denote their entry-wise matrix lifts, i.e., $[G(A)]_{im}=g(A_{im})$ for any matrix $A=[A_{im}]$. Let $\mathrm{M}_1, \{ \mathrm{M}_{2,\ell} | \ell \in \mathcal{L} \}, \mathrm{M}_3, \{ \mathrm{M}_{4,q} | q \in \mathcal{I} \}$ be binary masks of compatible dimensions. Define the effective coupling matrices by:
    
    \begin{subequations}
    \label{eq:coupling_matrices_cognition}
        \begin{align}
            &\Phi \coloneqq \mathrm{M}_1 \odot G_1 \left( \tilde{\Phi} \right), \\
            &\Psi_{\ell} \coloneqq \mathrm{M}_{2,\ell} \odot G_{2,\ell} \left( \tilde{\Psi}_{\ell} \right) , \\ 
            &\Theta \coloneqq \mathrm{M}_3 \odot G_3 \left( \tilde{\Theta} \right), \\ 
            &\Xi_q \coloneqq \mathrm{M}_{4,q} \odot G_{4,q} \left( \tilde{\Xi}_q \right). 
        \end{align}
    \end{subequations}
    
    \hfill $\square$
\end{definition}

The binary adjacency masks $\mathrm{M}_1, \{ \mathrm{M}_{2,\ell} | \ell \in \mathcal{L} \}, \mathrm{M}_3, \{ \mathrm{M}_{4,q} | q \in \mathcal{I} \}$ separate interaction structure from parameter values, encode prior connectivity assumptions, and, when sparse, reduce the number of free parameters to improve interpretability and practical identifiability. Without such structural restrictions, the number of cognition module parameters scales as $\bigO\left(\card{\mathcal{I}}^3 + \card{\mathcal{L}} \card{\mathcal{I}}^2\right)$, with the dominant contributions arising from the state-gated matrices $\{\Xi_q\}_{q\in\mathcal I}$ and input-gated matrices $\{\Psi_\ell\}_{\ell\in\mathcal L}$. The parameterization in Definition~\ref{def:cognitionmodulecrosscoupling} is assumed to satisfy:

\begin{enumerate}[leftmargin=*]

    \item \textbf{Link function regularity and constraints.}
    The scalar link functions $g_1, g_{2,\ell},g_3$ are continuous and locally Lipschitz. For each $q \in \mathcal{I}$, $g_{4,q}$ is continuous, locally Lipschitz, and bounded: $g_{4,q}(\cdot) \in \left[\underline{g_{4,q}}, \overline{g_{4,q}}\right].$ This  bounds $x_q(k)\Xi_{q,im}x_m(k)$ on any bounded cognition state set.
    \item \textbf{Masks and sparsity.}
    The adjacency masks configure the network structure of the state variables and specify the permitted couplings (entries $1$) and the excluded couplings (entries $0$):
    
    \begin{subequations}
        \begin{align}
            &\;\mathrm{M}_1\in\{0,1\}^{\card{\mathcal{I}}\times\card{\mathcal{I}}}, \\ 
            &\left\{\mathrm{M}_{2,\ell} \mid \ell\in\mathcal{L}\right\}\subseteq\{0,1\}^{\card{\mathcal{I}}\times\card{\mathcal{I}}}, \\ 
            &\mathrm{M}_3\in\{0,1\}^{\card{\mathcal{I}}\times\card{\mathcal{L}}}, \\
            &\left\{\mathrm{M}_{4,q} \;\middle|\; q\in\mathcal{I} \right\}\subseteq\{0,1\}^{\card{\mathcal{I}}\times\card{\mathcal{I}}},  
        \end{align}
    \end{subequations}
    
    \item \textbf{Hollowness.}
    The square masks are hollow:
    
    \begin{subequations}
        \begin{align}
            \mathrm{diag}\left(\mathrm{M}_1 \right) &= \boldsymbol{0},\\
            \mathrm{diag}\left(\mathrm{M}_{2,\ell}\right) &= \boldsymbol{0}, \quad \forall \ell\in\mathcal{L},\\
            \mathrm{diag}\left(\mathrm{M}_{4,q} \right) &= \boldsymbol{0}, \quad \forall q\in\mathcal{I}.
        \end{align}
    \end{subequations}
    
    The hollowness constraints encoded via the masks imply $\mathrm{diag}(\Phi)=\boldsymbol{0}$, $\mathrm{diag}(\Psi_{\ell})=\boldsymbol{0}$, $\forall \ell\in\mathcal{L}$, and $\mathrm{diag}(\Xi_q)=\boldsymbol{0}$, $\forall q\in\mathcal{I}$. This ensures no self-effects are introduced via baseline, input‑gated, or state‑gated couplings. 
\end{enumerate}

The cognition module maps the current perceptual state and previous cognitive state to the updated state vector,
$\bx(k)\coloneqq\zcognition(\bxnpe(k),\bx(k-1))$,
which is forwarded to the decision-making module to drive downstream behavioral outputs. The next section explains how cognitive states give rise to intentions and actions.

\begin{remark}
\label{remark:goal_belief_state_partition}
    \textbf{Goal–belief state partition.}
    For the integration of the decision-making module below, we assume that a subset of the cognitive state indices corresponds to goal-related variables and another subset corresponds to belief-related variables. Specifically, let $\mathcal{G}\subseteq\mathcal{I}$ denote the set of goal state indices and let $\mathcal{B}\subseteq\mathcal{I}$ denote the set of belief state indices. Additional states in $\mathcal{I}\setminus(\mathcal{G}\cup\mathcal{B})$ may represent emotions, appraisals, contextual variables, memories, or other latent cognitive quantities that influence the dynamics but are not mapped directly into intentions.
\end{remark}

\subsection{Decision-making module}
\noindent
The decision-making module maps the cognitive state vector $\bx(k)$ to behavioral output. It first forms intention intensities from goal and belief states, and then maps these intentions to actions through a bounded competitive gate.

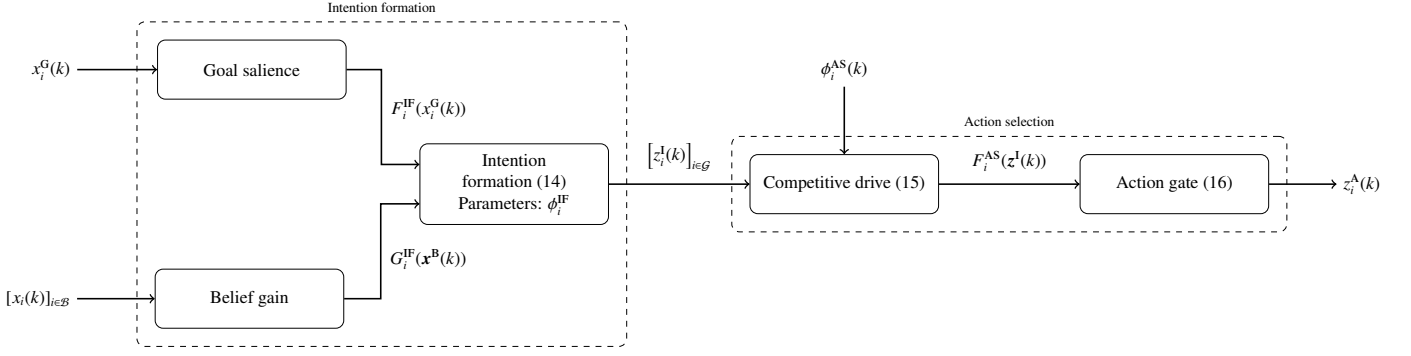
\begin{figure*}[!t]
    \centering
    \resizebox{\linewidth}{!}{%
    \begin{tikzpicture}[
    font=\small,
    node distance=1.15cm and 1.25cm,
    block/.style={
    draw,
    rounded corners,
    align=center,
    minimum width=3.35cm,
    minimum height=1.05cm,
    inner sep=5pt
    },
    smallblock/.style={
    draw,
    rounded corners,
    align=center,
    minimum width=2.60cm,
    minimum height=0.75cm,
    inner sep=4pt
    },
    tinyblock/.style={
    draw,
    rounded corners,
    align=center,
    minimum width=1.25cm,
    minimum height=0.55cm,
    inner sep=2pt
    },
    arrow/.style={->, thick},
    dashedbox/.style={
    draw,
    dashed,
    rounded corners,
    inner sep=9pt
    }
    ]
    
    \node[] (beliefinput)
    {$\left[x_i(k)\right]_{i \in \mathcal{B}}$};
    
    \node[above=3.5cm of beliefinput, xshift=0.23cm] (goalinput)
    {$\xigoals(k)$};
    
    \node[block, right=1.65cm of goalinput, xshift=-0.23cm] (goal)
    {Goal salience};
    
    \node[block, right=1.38cm of beliefinput] (belief)
    {Belief gain};
    
    \node[block, right=3cm of $(goal)!0.5!(belief)$] (intention)
    {Intention\\formation \eqref{eqn:intentionselection}\\Parameters: $\phiisi$};
    
    \node[block, right=2.5cm of intention] (drive)
    {Competitive drive \eqref{eqn:decisionmakingfasi}};
    
    \node[block, right=2.5cm of drive] (action)
    {Action gate \eqref{eqn:actionselection}};
    
    \node[above=1.2cm of drive] (threshold)
    {$\phiasi(k)$};
    
    \node[right=1.20cm of action] (output)
    {$\zai(k)$};
    
    \draw[arrow] (goalinput) -- (goal);
    
    \draw[arrow] (beliefinput) -- (belief);
    
    \draw[arrow] (goal.east) -- ++(0.62,0)
    coordinate (goalturn)
    |- ([yshift=0.35cm]intention.west);
    
    \node[right=0.05cm] at ([yshift=-0.70cm]goalturn)
    {$\fisi(\xigoals(k))$};
    
    \draw[arrow] (belief.east) -- ++(0.65,0)
    coordinate (beliefturn)
    |- ([yshift=-0.35cm]intention.west);
    
    \node[right=0.05cm] at ([yshift=0.70cm]beliefturn)
    {$\gisi(\xbeliefs(k))$};
    
    \draw[arrow] (intention) -- node[above, yshift=0.10cm]
    {$\left[ \zii(k)\right]_{i\in \mathcal{G}}$}
    (drive);
    
    \draw[arrow] (threshold) -- (drive);
    
    \draw[arrow] (drive) -- node[above, yshift=0.10cm]
    {$\fasi(\bzi(k))$}
    (action);
    
    \draw[arrow] (action) -- (output);
    
    \node[
    dashedbox,
    fit=(goal)(belief)(intention)
    ] (ifdash) {};
    \node[above=0.05cm of ifdash] {\scriptsize Intention formation};
    
    \node[
    dashedbox,
    fit=(drive)(action)
    ] (asdash) {};
    \node[above=0.05cm of asdash] {\scriptsize Action selection};
    
    \end{tikzpicture}%
    }
    \caption{
    Block diagram of the decision-making module.
    }
    \label{fig:decision_module_block_diagram}
\end{figure*}

\subsubsection{Intention formation}
\noindent

We model intention formation by separating goal salience from its belief-dependent expression. 
Specifically, intention formation is factored into (i) a non-negative, value-weighted goal salience function $\fisi(\cdot)$ and (ii) a belief-dependent gain function $\gisi(\cdot)$ that modulates how strongly salience is expressed as intention. Accordingly, we model each intention intensity as a baseline plus the product of these two terms.

Consistent with the partition introduced in Remark~\ref{remark:goal_belief_state_partition}, we collect the goal-related and belief-related components of the cognitive state as $\bxb(k) \coloneqq \left[x_j(k)\right]_{j\in \mathcal{B}}$ and $\bxg(k) \coloneqq \left[x_j(k)\right]_{j\in \mathcal{G}}$. For any index subset $\mathcal{S} \subseteq \mathcal{I}$, the corresponding sub-vector inherits its admissible range $\mathcal{X}^{\textup{s}}$ through a coordinate projection defined via $\Gamma\left(\mathcal{S}\right): \mathcal{X} \to \mathbb{R}^{\card{\mathcal{S}}}$. Accordingly, each individual element $s$ of the sub-vector inherits its admissible range through the mapping $\Gamma\left(\{s\}\right): \mathcal{X} \to \mathbb{R}$ applied to the singleton $\{s\}$. In the following definition, $\mathcal{S}$ denotes either $\mathcal{B}$ or $\mathcal{G}$. For clarity, when referring to a specific belief or goal component, we write it, respectively, as $\xb(k)  \equiv x_i(k)$ with $i \in \mathcal{B}$ and as $\xg(k)  \equiv x_i(k)$ with $i \in \mathcal{G}$. In this paper, each goal index $i \in \mathcal{G}$ defines one corresponding intention channel $z_i^{\mathrm {I}}$. More general mappings, where one intention depends on multiple goals or one goal contributes to multiple intentions, are left for future extensions.

\begin{definition}
\label{def:intention_formation}
    Suppose that at time step $k$ the cognition module outputs belief vector $\bxb(k)$ and goal vector $\bxg(k)$. For a given goal index $i \in \mathcal{G}$, the intention intensity mapping $\zii : \xgiset \times \xbset \to \ziiset$, where $\ziiset \coloneqq \left[0, \ziimax \right]$, is defined as:
    \begin{align}
    \label{eqn:intentionselection}
        \zii\left(\xigoals(k), \xbeliefs(k)\right) \coloneqq \phiisi + \fisi \left(\xigoals(k)\right) \; \gisi \left(\xbeliefs(k)\right).
    \end{align}
    Here, $\phiisi \in [0,\infty)$ is a baseline intensity, $\fisi: \xgiset \to \left[0, \fisimax \right]$ is a goal salience mapping that encodes how strongly goal $\xigoals$ drives intention formation, and $\gisi: \xbset \to \left[ 0, \gisimax \right]$ is a belief-dependent gain mapping that modulates how strongly beliefs affect the intention. The constants $\ziimax, \fisimax, \gisimax$ are finite bounds specified below. \hfill $\square$
\end{definition}

We assume that $\fisi(\cdot)$ and $\gisi(\cdot)$ satisfy the following conditions:
    
\begin{enumerate}[leftmargin=*]
    \item \textbf{Boundedness:} 
    Since $\zii(\cdot, \cdot)$ is modeled as an intensity, we restrict $\gisi(\cdot)$ to be non-negative, so that beliefs modulate the goal-to-intention gain without directly inducing negative intention intensity. Furthermore, there exist finite constants $\fisimax$ and $\gisimax$ such that
    
    \begin{align*}
    \begin{split}
        0 \le \fisi\left(\xigoals(k)\right) \le \fisimax, \quad \forall \xigoals(k)\in \xgiset, \\
        0 \le \gisi\left(\bxb(k)\right) \le \gisimax, \quad \forall \bxb(k) \in \xbset.
    \end{split}
    \end{align*}

    These bounds reflect saturating gain effects observed in neural valuation and selection processes \cite{b20,b55,b56}, and ensure that $\zii(\cdot,\cdot)$ remains uniformly bounded. Finally, we choose $\ziimax \coloneqq \phiisi + \fisimax \gisimax$ so that $\zii(\cdot, \cdot)\in \left[0,\ziimax \right]$ holds. 
    
    \item \textbf{Monotonicity:} 
    We assume $\fisi(\cdot)$ is non-decreasing on $\xgiset$. This is a modeling choice for intention $i$ associated with goal $i$: the direct goal-to-intention contribution is facilitatory, while reductions in intention salience are captured by explicit suppressive pathways. For the belief-dependent gain, fix a partition of belief indices for goal $i$ into supportive and suppressive sets $\mathcal{B}_i^{+}$ and $\mathcal{B}_i^{-}$ with $\mathcal{B}_i^{+}\cap \mathcal{B}_i^{-}=\varnothing$ and $\mathcal{B}_i^{+}\cup \mathcal{B}_i^{-} \subseteq \mathcal{B}$. The remaining indices $\mathcal{B}_i^{0}\coloneqq \mathcal{B}\setminus(\mathcal{B}_i^{+}\cup \mathcal{B}_i^{-})$ are interpreted as neutral (neither supportive nor suppressive) for goal $i$. Then $\gisi(\cdot)$ is chosen to be coordinate-wise monotone in the following sense: for any two belief vectors $\bxb_1, \bxb_2\in \xbset$, if
    
    \begin{align*}
        \xb_{1,j} &\le \xb_{2,j} && \forall j\in\mathcal{B}_i^{+}, \\
        \xb_{1,j} &\ge \xb_{2,j} && \forall j\in\mathcal{B}_i^{-}, \\
        \xb_{1,j} &= \xb_{2,j} && \forall j\in \mathcal{B}_i^{0}.
    \end{align*}
    
    then $\gisi(\bxb_1)\le \gisi(\bxb_2)$. This captures that increasing supportive beliefs and/or decreasing suppressive beliefs does not reduce the gain, consistent with belief-based modulation of intention formation \cite{b54,b29,b32}.
\end{enumerate}

The collection of intention intensities for all goals is used as the input to the action selection stage described next.

\subsubsection{Action selection}
\noindent
We model action selection using a self-facilitating drive, a competition-induced suppression term, and a bounded gating nonlinearity with conflict-dependent thresholding.

\begin{definition}
\label{def:actionselection}
    Suppose that the intention selection stage at time step $k$ produces a vector $\bzi(k) \in \ziset$ of intention intensities, where $\ziset \coloneqq \prod_{i \in \mathcal{G}} \ziiset$. Fix $i \in \mathcal{G}$ and decompose $\bzi(k) = \left( \zii(k), \bzimini(k) \right)$ with $\bzimini(k) \in \ziminiset$ and 
    $\ziminiset\coloneqq \prod_{\substack{m \in \mathcal{G} \\ m \neq i}}  \zimset$, such that $\ziset = \ziiset \times \ziminiset$. Define $\gasi: \ziiset \to \left[ 0,\gasimax \right]$, which captures self-facilitation, and $\hasi: \ziminiset \to \left[ 0,\hasimax \right]$, which captures suppression due to competing intentions. Define the competitive drive  $\fasi : \ziset \to \mathbb{R}$ for the action induced by goal $i$ as:
    \begin{align}
    \label{eqn:decisionmakingfasi}
        \fasi\left(\bzi(k)\right) \coloneqq \gasi \left(\zii(k)\right) - \phiasi(k) - \hasi \left(\bzimini(k) \right),
    \end{align}
    where $\phiasi(k) \in [0, \infty)$ is a action threshold. The  action selection stage outputs an action intensity, defined by the sigmoid gate $\sasi : \mathbb{R} \to (0,1)$:
    \begin{align}
    \label{eqn:actionselection}
        \zai\left(\bzi(k)\right) \coloneqq \sasi \left(\fasi\left(\bzi(k)\right)\right).
    \end{align} \hfill $\square$
\end{definition}

The exogenous action threshold $\phiasi(k)$ is time-dependent and may increase under conflict, delaying commitment when competing intentions are similar in strength. Furthermore, we assume the following properties of $\gasi(\cdot), \hasi(\cdot)$, and $\sasi(\cdot)$:

\begin{enumerate}[leftmargin=*]
    \item \textbf{Monotonicity:} Functions $\gasi(\cdot)$ and $\hasi(\cdot)$ are monotone non-decreasing in, respectively, $\zii$ and each component of $\bzimini$. The gate $\sasi(\cdot)$ is monotone increasing. These monotonicity conditions encode that stronger intention-related drive for action $i$ 
     only increases the gating tendency, while stronger competing drives only increase the net suppression exerted on action $i$ \cite{b57,b58}.
    \item \textbf{Boundedness:} 
    Both $\gasi(\cdot)$ and $\hasi(\cdot)$ are bounded by construction with finite upper bounds $\gasimax$ and $\hasimax$, respectively. The boundedness of $\sasi(\cdot)$ reflects the saturating nature of the action selection gate and prevents unbounded action activation \cite{b57,b58,b60}.
    \item \textbf{Asymptotic behavior:} Function $\sasi(\cdot)$ follows 
    \begin{align*}
        \lim_{\zeta \to -\infty} \sasi(\zeta) = 0, \quad \lim_{\zeta \to +\infty} \sasi(\zeta) = 1.
    \end{align*}
    These asymptotic limits ensure a saturating gating regime in which sufficiently negative and sufficiently positive drives yield, respectively, near-zero and near-one outputs \cite{b58,b59,b60}.
\end{enumerate}

Together, intention formation and action selection define the decision-making module, whose input–output relation is summarized by $\bza(k)=\zdecision(\bx(k))$. Here, $\bza(k)\in\zaset\coloneqq\prod_{i\in\mathcal G}\zaiset$ denotes the vector of action intensities. Figure~\ref{fig:decision_module_block_diagram} shows the corresponding block diagram.

\subsubsection{Regularity assumptions}
The rest of the paper is based on the following assumptions:

\begin{assumption}\textbf{Regularity of attentional selection mappings:} 
\label{ass:attentionalselection}
    For each $\ell \in \mathcal{L}$, the functions 
    $\fasl: \ulset \to \left[0, \faslmax\right]$ and 
    $\gasl: \uset \to \left[\gaslmin, \infty\right)$ are continuously differentiable on open neighborhoods of $\ulset$ and $\uset$, respectively. In particular, they are Lipschitz on the compact admissible sets $\ulset$ and $\uset$.
\end{assumption}

\begin{assumption}\textbf{Regularity of intention formation mappings:}
\label{ass:intentionformation}
    For each $i \in \mathcal{G}$, the functions $\fisi(\cdot)$ and $\gisi(\cdot)$ are locally Lipschitz on open neighborhoods of $\xgiset$ and $\xbset$, respectively.
\end{assumption}

\begin{assumption}\textbf{Regularity of action selection mappings:}
\label{ass:actionselection}
    For each $i \in \mathcal{G}$, the functions $\gasi(\cdot)$ and $\hasi(\cdot)$ are locally Lipschitz on open neighborhoods of $\ziiset$ and $\ziminiset$, respectively. The action gate $\sasi(\cdot)$ is continuously differentiable, and hence locally Lipschitz on $\mathbb{R}$.
\end{assumption}

\begin{table}[!t]
\centering
\caption{Main variables used in the model formulation. The domain column gives the admissible value set or structural parameter space.}
\label{tab:variables}
\begin{minipage}{\columnwidth}

\hrule
\vspace{0.5em}

\noindent
\makebox[0.20\columnwidth][l]{\textbf{Variable}}%
\makebox[0.50\columnwidth][l]{\textbf{Description}}%
\makebox[0.09\columnwidth][l]{\textbf{Eq.}}%
\makebox[0.20\columnwidth][r]{\textbf{Domain}}

\vspace{0.4em}
\hrule
\vspace{0.4em}

\noindent
\makebox[0.20\columnwidth][l]{$\ul(k)$}%
\parbox[t]{0.50\columnwidth}{Input to perception module \\ through channel $\ell$}%
\makebox[0.09\columnwidth][l]{\eqref{eqn:perceptualaccessgeneral}}%
\makebox[0.20\columnwidth][r]{$\ulset \subseteq \mathbb{R}$}

\vspace{0.4em}
\hrule
\vspace{0.4em}

\noindent
\makebox[0.20\columnwidth][l]{$\zasl(\bu(k))$}%
\parbox[t]{0.50\columnwidth}{Channel-wise \\ attentional selection weight}%
\makebox[0.09\columnwidth][l]{\eqref{eqn:perceptualaccessgeneral}}%
\makebox[0.20\columnwidth][r]{$\left[0, \zaslmax \right]$}

\vspace{0.4em}
\hrule
\vspace{0.4em}

\noindent
\makebox[0.20\columnwidth][l]{$\betaasl$}%
\makebox[0.50\columnwidth][l]{Semi-saturation constant}%
\makebox[0.09\columnwidth][l]{\eqref{eqn:attentionalselectionG}}%
\makebox[0.20\columnwidth][r]{$(0,\infty)$}

\vspace{0.4em}
\hrule
\vspace{0.4em}

\noindent
\makebox[0.20\columnwidth][l]{$\gammaaslm$}%
\makebox[0.50\columnwidth][l]{Pooling weight}%
\makebox[0.09\columnwidth][l]{\eqref{eqn:attentionalselectionG}}%
\makebox[0.20\columnwidth][r]{$[0, \infty)$}

\vspace{0.4em}
\hrule
\vspace{0.4em}

\noindent
\makebox[0.20\columnwidth][l]{$\xnpel(k)$}%
\makebox[0.50\columnwidth][l]{\gls{npe} state}%
\makebox[0.09\columnwidth][l]{\eqref{eq:NPE_recursion}}%
\makebox[0.20\columnwidth][r]{$\left[0, \zaslmax \right]$}

\vspace{0.4em}
\hrule
\vspace{0.4em}

\noindent
\makebox[0.20\columnwidth][l]{$\alphapil(k)$}%
\makebox[0.50\columnwidth][l]{Effective update weight}%
\makebox[0.09\columnwidth][l]{\eqref{eq:NPE_recursion}}%
\makebox[0.20\columnwidth][r]{$[0,1)$}

\vspace{0.4em}
\hrule
\vspace{0.4em}

\noindent
\makebox[0.20\columnwidth][l]{$\phipil(k)$}%
\makebox[0.50\columnwidth][l]{Sensory precision}%
\makebox[0.09\columnwidth][l]{\eqref{eqn:predictiveinferencealpha}}%
\makebox[0.20\columnwidth][r]{$(0,\infty)$}

\vspace{0.4em}
\hrule
\vspace{0.4em}

\noindent
\makebox[0.20\columnwidth][l]{$\thetapil(k)$}%
\makebox[0.50\columnwidth][l]{Prior precision}%
\makebox[0.09\columnwidth][l]{\eqref{eqn:predictiveinferencealpha}}%
\makebox[0.20\columnwidth][r]{$(0,\infty)$}

\vspace{0.4em}
\hrule
\vspace{0.4em}

\noindent
\makebox[0.20\columnwidth][l]{$\chipil(k)$}%
\makebox[0.50\columnwidth][l]{Attention gain}%
\makebox[0.09\columnwidth][l]{\eqref{eqn:predictiveinferencealpha}}%
\makebox[0.20\columnwidth][r]{$[0,\infty)$}

\vspace{0.4em}
\hrule
\vspace{0.4em}

\noindent
\makebox[0.20\columnwidth][l]{$x_i(k)$}%
\makebox[0.50\columnwidth][l]{Cognitive state}%
\makebox[0.09\columnwidth][l]{\eqref{eqn:dcmnonlinear}}%
\makebox[0.20\columnwidth][r]{$\mathcal{X}_i \subseteq \mathbb{R}$}

\vspace{0.4em}
\hrule
\vspace{0.4em}

\noindent
\makebox[0.20\columnwidth][l]{$\Theta$}%
\makebox[0.50\columnwidth][l]{Perceptual drive}%
\makebox[0.09\columnwidth][l]{\eqref{eqn:dcmnonlinear}}%
\makebox[0.20\columnwidth][r]{$\mathbb{R}^{\card{\mathcal{I}}\times\card{\mathcal{L}}}$}

\vspace{0.4em}
\hrule
\vspace{0.4em}

\noindent
\makebox[0.20\columnwidth][l]{$\kappa$}%
\makebox[0.50\columnwidth][l]{Self-inhibition}%
\makebox[0.09\columnwidth][l]{\eqref{eqn:cognitionmodulefiic}}%
\makebox[0.20\columnwidth][r]{$(0,\infty)$}

\vspace{0.4em}
\hrule
\vspace{0.4em}

\noindent
\makebox[0.20\columnwidth][l]{$\gamma_i$}%
\makebox[0.50\columnwidth][l]{Baseline log-leak}%
\makebox[0.09\columnwidth][l]{\eqref{eqn:cognitionmodulefiic}}%
\makebox[0.20\columnwidth][r]{$\left[-\overline{\boldsymbol{\gamma}}, \overline{\boldsymbol{\gamma}} \right]$}

\vspace{0.4em}
\hrule
\vspace{0.4em}

\noindent
\makebox[0.20\columnwidth][l]{$\Lambda_{i \ell}$}%
\makebox[0.50\columnwidth][l]{Perception-leak coupling}%
\makebox[0.09\columnwidth][l]{\eqref{eqn:cognitionmodulefiic}}%
\makebox[0.20\columnwidth][r]{$\left[-\overline{\Lambda_{\ell}}, \overline{\Lambda_{\ell}} \right]$}

\vspace{0.4em}
\hrule
\vspace{0.4em}

\noindent
\makebox[0.20\columnwidth][l]{$\Phi$}%
\makebox[0.50\columnwidth][l]{Baseline coupling}%
\makebox[0.09\columnwidth][l]{\eqref{eqn:gimc}}%
\makebox[0.20\columnwidth][r]{$\mathbb{R}^{\card{\mathcal{I}}\times\card{\mathcal{I}}}$}

\vspace{0.4em}
\hrule
\vspace{0.4em}

\noindent
\makebox[0.20\columnwidth][l]{$\Psi_{\ell}$}%
\makebox[0.50\columnwidth][l]{Input-gated coupling}%
\makebox[0.09\columnwidth][l]{\eqref{eqn:gimc}}%
\makebox[0.20\columnwidth][r]{$\mathbb{R}^{\card{\mathcal{I}}\times\card{\mathcal{I}}}$}

\vspace{0.4em}
\hrule
\vspace{0.4em}

\noindent
\makebox[0.20\columnwidth][l]{$\Xi_q$}%
\makebox[0.50\columnwidth][l]{State-gated coupling}%
\makebox[0.09\columnwidth][l]{\eqref{eqn:gimc}}%
\makebox[0.20\columnwidth][r]{$\mathbb{R}^{\card{\mathcal{I}}\times\card{\mathcal{I}}}$}

\vspace{0.4em}
\hrule
\vspace{0.4em}

\noindent
\makebox[0.20\columnwidth][l]{$\phiisi$}%
\makebox[0.50\columnwidth][l]{Baseline intention intensity}%
\makebox[0.09\columnwidth][l]{\eqref{eqn:intentionselection}}%
\makebox[0.20\columnwidth][r]{$[0,\infty)$}

\vspace{0.4em}
\hrule
\vspace{0.4em}

\noindent
\makebox[0.20\columnwidth][l]{$\zii(k)$}%
\makebox[0.50\columnwidth][l]{Intention intensity}%
\makebox[0.09\columnwidth][l]{\eqref{eqn:intentionselection}}%
\makebox[0.20\columnwidth][r]{$\left[ 0, \ziimax \right]$}

\vspace{0.4em}
\hrule
\vspace{0.4em}

\noindent
\makebox[0.20\columnwidth][l]{$\phiasi(k)$}%
\makebox[0.50\columnwidth][l]{Baseline action threshold}%
\makebox[0.09\columnwidth][l]{\eqref{eqn:decisionmakingfasi}}%
\makebox[0.20\columnwidth][r]{$[0,\infty)$}

\vspace{0.4em}
\hrule
\vspace{0.4em}

\noindent
\makebox[0.20\columnwidth][l]{$\zai(k)$}%
\makebox[0.50\columnwidth][l]{Action intensity}%
\makebox[0.09\columnwidth][l]{\eqref{eqn:decisionmakingfasi}}%
\makebox[0.20\columnwidth][r]{$(0,1)$}

\vspace{0.4em}
\hrule

\end{minipage}
\end{table}

\section{Stability analysis}
\label{sec:stabilityanalysis}
\noindent
This section establishes boundedness, regularity, and input–state stability properties of the proposed modular model. Since the internal states represent latent cognitive variables, admissible sensory perturbations should propagate through the perception–cognition–decision pipeline without inducing unbounded internal responses.

The results are sufficient conditions for stable model operation, not global guarantees for all parameter choices. The model may still exhibit strongly amplified or unstable behavior under parameter regimes intended to represent dysregulated responses, high uncertainty, or conflicting internal states. The conditions below identify operating regimes in which the model remains suitable for bounded latent-state prediction and model-based control.

The analysis proceeds compositionally. We first bound the perception-to-\gls{npe} mapping, then analyze the cognition dynamics, and finally propagate the resulting bounds through intention formation and action selection. Proofs are given in \ref{appendix}.

\subsection{Perception module}
\noindent
We first establish well-posedness, boundedness, and regularity of the perception module under admissible sensory inputs. The results show that attentional selection is Lipschitz, that the \gls{npe} state set is forward invariant, and that the predictive inference recursion is both contractive under constant selected evidence and Lipschitz as a one-step update.

Lemma~\ref{lemma:attentionalselectionsinglestep} shows that the attentional selection mapping is Lipschitz in the sensory input, ensuring that small input perturbations cannot produce abrupt changes in attentional weights.

\begin{lemma}\textbf{Single-step Lipschitz property of attentional selection:} 
\label{lemma:attentionalselectionsinglestep} 
    Under Assumption~\ref{ass:attentionalselection}, for each channel $\ell \in \mathcal{L}$, the attentional selection weight mapping \eqref{eqn:perceptualaccessgeneral} is Lipschitz continuous on the compact set $\uset$. Hence, there exists a finite constant $\latsl > 0$, such that for all $\bu_1,\; \bu_2 \in \uset$, 
    \begin{align}
    \label{eqn:lemmalipschitzattentionalselection}
        \abs{\zasl(\bu_1)-\zasl(\bu_2)} \leq  \latsl  
        \norm{\bu_1-\bu_2}.
    \end{align}
\end{lemma}

Proposition~\ref{prop:attentionalselectiontrajectory} lifts this single-step bound to a trajectory-wise bound, demonstrating that the cumulative variation of attentional selection along an input trajectory is controlled by the total variation of sensory inputs.

\begin{proposition}\textbf{Trajectory-wise \gls{io} bound of attentional selection:} 
\label{prop:attentionalselectiontrajectory}
    Suppose that Assumption~\ref{ass:attentionalselection} holds, and let $\latsl$ be the Lipschitz constant from Lemma~\ref{lemma:attentionalselectionsinglestep}. Then, for any admissible input trajectory and any $k_2\ge k_1$ the trajectory-wise \gls{io} variation of the attentional selection is bounded by  
    \begin{align}
    \label{eqn:propattnselection1}
        \sum_{i=k_1}^{k_2-1}\abs{\zasl(\bu(i+1))-\zasl(\bu(i))}\le  \latsl \sum_{i=k_1}^{k_2-1}\norm{\bu(i+1) - \bu(i)}.
    \end{align}
\end{proposition}

Lemma~\ref{lem:npe_forward_invariance} guarantees forward invariance of the \gls{npe} admissible set, so that \glspl{npe} remain within a cognitively meaningful bounded range under all admissible inputs.

\begin{lemma}\textbf{Forward invariance of the \gls{npe} state set $\xnpelset$:} 
\label{lem:npe_forward_invariance}
    Let $\xnpelset \coloneqq \left[0,\zaslmax \right]$, with $\zaslmax$ a uniform upper bound for $\zasl(\cdot)$.  
    If $\xnpel(0) \in \xnpelset$, then $\xnpel(k)$
    remains in $\xnpelset$ for all $k>0$.
\end{lemma}

Lemma~\ref{lem:npe_contraction_constant} characterizes exponential convergence of the \gls{npe} towards the steady-state perceptual estimate under a constant attentional selection weight. This shows that predictive inference removes dependence on the initial perceptual estimate under stationary selected evidence rather than accumulating error.

\begin{lemma} \textbf{Contraction of the \gls{npe} under constant input:}
\label{lem:npe_contraction_constant} 
    Let $\bar{\bu}$ be a constant input with corresponding attentional selection weight $\zasl(\bar{\bu})$.  
    If the effective update weights in \eqref{eq:NPE_recursion} are uniformly bounded away from zero, i.e., 
    $\alphapil(k) \geq \underline{\alphapil}>0$, then the \gls{npe} tracking error 
    \begin{align}
    \label{eqn:npetrackingerror}
        \enpe_\ell(k)\coloneqq \xnpel(k)-\zasl(\bar{\bu})
    \end{align}
    decays geometrically, i.e., converges exponentially to zero:  
    \begin{align*}
        \left|\enpe_\ell(k)\right| \le \left(1 - \underline{\alphapil}\right)^k \left|\enpe_\ell(0)\right|.
    \end{align*}
\end{lemma}

The condition $\alphapil(k)\geq\underline{\alphapil}>0$ is a persistence-of-updating condition. If $\alphapil(k)$ can remain zero for arbitrarily long intervals, the recursion remains bounded, but exponential convergence to the constant attentional-selection target is not guaranteed.

Finally, Lemma~\ref{lem:npe_update_lipschitz} establishes that the one-step \gls{npe} update is Lipschitz in both the sensory input and the current \gls{npe} state, which ensures that the perceptual recursion responds smoothly to admissible perturbations. 

\begin{lemma}\textbf{Lipschitz continuity of the \gls{npe} update mapping:} 
\label{lem:npe_update_lipschitz}
    For channel $\ell\in\mathcal{L}$ and a fixed time step $k$, define the one-step \gls{npe} update induced by \eqref{eqn:predictiveinferenceg} and \eqref{eq:NPE_recursion} via the mapping  
    $T_{\ell} : \uset \times \xnpelset \to \xnpelset$ given by:
    \begin{align}
    \label{eq:def_Tlk_npe}
        T_{\ell}(\bu, \xi) \coloneqq \xi + \alphapil(k) \left(\zasl(\bu) - \xi \right),
    \end{align}
    where $\xi \in \xnpelset$. If $\zasl(\cdot)$ is Lipschitz on $\uset$ with constant $\latsl$ as per Lemma~\ref{lemma:attentionalselectionsinglestep}, then $T_{\ell}(\cdot,\cdot)$ is Lipschitz on $\uset \times \xnpelset$ with respect to the product norm
    $\norm{\left(\bu, \xi \right)} \coloneqq \norm{\bu} + \abs{\xi}$. 
\end{lemma}

\subsection{Cognition module}
\noindent
The cognition module yields three complementary properties that collectively underpin the stability of the overall coupled model. 

Lemma~\ref{lem:uniformselfinhibition} provides uniform upper and lower bounds on the self-inhibition coefficients for the cognition module, 
i.e., the input-modulated terms that govern the intrinsic damping of each state 
over the admissible domain. Therefore, it quantifies the weakest and strongest intrinsic damping that the cognition dynamics can exhibit under admissible perceptual conditions. 

\begin{lemma}\textbf{Uniform self-inhibition bounds for cognition}:
\label{lem:uniformselfinhibition}
    Let $\Lambda_i \coloneqq \left[ \Lambda_{i\ell} \right]_{\ell \in \mathcal{L}}$ denote the vector of perceptual coupling coefficients affecting 
    cognitive state $i$, with $\Lambda_{i\ell}$ as in \eqref{eqn:cognitionmodulefiic}, and let
    $\xnpebar \coloneqq \max_{\boldsymbol{\xi} \in \xnpeset} \norm{\boldsymbol{\xi}}$.
    Define the constants
    $\underline{\Lambda} \coloneqq \kappa \exp \left( \min_{i \in \mathcal{I}} \left( \gamma_i - \norm{\Lambda_i}
    \xnpebar \right) \right)$ and $\overline{\Lambda} \coloneqq \kappa \exp \left( \max_{i \in \mathcal{I}} \left( \gamma_i + \norm{\Lambda_i}
    \xnpebar \right) \right)$, where $\kappa > 0$ and $\gamma_i$ are as in \eqref{eqn:cognitionmodulefiic}.
    Then, for all $k$ and $i \in \mathcal{I}$,
    \begin{align*}
        -\overline{\Lambda} \leq \fiic \left(\bxnpe(k)\right) \leq -\underline{\Lambda}.
    \end{align*}
\end{lemma}

Proposition~\ref{prop:cognition_forward_invariance} combines these uniform bounds with induced norm bounds on the coupling matrices to establish forward invariance of a bounded cognition set. In particular, if a cognition state is initialized 
within a cognitively meaningful range, the internal couplings cannot drive it outside that range under admissible \gls{npe} trajectories. 

\begin{proposition}\textbf{Forward invariance of the cognition state set:}
\label{prop:cognition_forward_invariance}
    Assume that
    \begin{itemize}[leftmargin=*]
        \item the \gls{npe} state set $\xnpeset$ is bounded, so that $\norm{\bxnpe(k)} \leq \xnpebar$ for all $k$;
        \item the cognition self-inhibition satisfies uniform bounds of Lemma~\ref{lem:uniformselfinhibition};
        \item induced matrix norms in \eqref{eq:coupling_matrices_cognition} are bounded, i.e., $\norm{\Phi} \leq \overline{\Phi}$, $\norm{\Psi_{\ell}} \leq \overline{\Psi_{\ell}}$, $\norm{\Xi_q} \leq \overline{\Xi_q}$, and $\norm{\Theta} \leq \overline{\Theta}$; also define 
        $\overline{\Xi} \coloneqq \sum_{q \in \mathcal{I}} \overline{\Xi_{q}}$ and $
        \overline{\Psi} \coloneqq \sup_{\boldsymbol{\xi} \in \xnpeset} \sum_{\ell \in \mathcal{L}} \abs{\xi_\ell} \norm{\Psi_{\ell}}$;
        \item the discrete-time step size $\Delta t$ used in the forward-Euler approximation satisfies $\Delta t \in \left(0, \tfrac{1}{\overline{\Lambda}} \right]$;
        \item there exists $R\in[0,\infty)$ such that
        \begin{align}
        \label{eqn:cognitionmodulelemmaquadraticinequality}
        \max_{0\le \norm{\bx(k)}\le R}
        \Bigg[&
        \norm{\bx(k)} +\Delta t
        \bigg(
        \overline{\Xi} \norm{\bx(k)}^2
        + \nonumber \\
        &
        \left(
        \overline{\Phi}+\overline{\Psi}-\underline{\Lambda}
        \right) \norm{\bx(k)}
        + \overline{\Theta}\;\xnpebar
        \bigg)
        \Bigg]
        \le R.
        \end{align}
    \end{itemize}
    Then the following set of cognitive states is forward invariant:
    \begin{align}
    \label{eq:forward_invariant_XR}
        \mathcal{X}_{R} = \Big\{ \bx : \norm{\bx} \leq R \Big\}.
    \end{align}
\end{proposition}

The assumptions in Proposition~\ref{prop:cognition_forward_invariance} are sufficient for forward invariance. The key requirement is the dissipativity condition \eqref{eqn:cognitionmodulelemmaquadraticinequality}, which ensures that self-inhibition dominates the combined effects of baseline coupling, input-gated coupling, endogenous state-gated amplification, and bounded perceptual drive on $\mathcal{X}_R$. The following corollary provides explicit conditions under which this norm-ball certificate is feasible.

\begin{corollary}
\textbf{Explicit feasibility conditions for cognition-state invariance.}
\label{cor:cognition_forward_invariance_explicit}
Under the assumptions of Proposition~\ref{prop:cognition_forward_invariance},
define
\begin{align*}
    \sigma \coloneqq \overline{\Phi}+\overline{\Psi}-\underline{\Lambda}, \quad 
    \tau \coloneqq \overline{\Theta}\xnpebar.
\end{align*}
If $\overline{\Xi}=0$, a feasible invariant radius exists whenever $\sigma<0$, and any radius $R$ satisfying
\begin{align*}
R
\geq
\max\left\{
\Delta t\tau,
\frac{\tau}{-\sigma}
\right\}
\end{align*}
is admissible. 
If $\overline{\Xi}>0$, define
\begin{align*}
R_{\pm} \coloneqq
\frac{-\sigma\pm\sqrt{\sigma^2-4\overline{\Xi} \tau}}{2\overline{\Xi}}.
\end{align*}
A feasible invariant radius exists if
\begin{align*}
\sigma<0, \quad
\sigma^2-4\overline{\Xi} \tau \geq0,\quad
\Delta t\tau\leq R_{+}.
\end{align*}
In that case, every radius $R$ satisfying
\begin{align*}
\max\left\{
\Delta t \tau,
R_{-}
\right\}
\leq R\leq R_{+}
\end{align*}
satisfies \eqref{eqn:cognitionmodulelemmaquadraticinequality}.
\end{corollary}

This certificate is sufficient and may be conservative because it is obtained from uniform induced-norm bounds and the triangle inequality. Failure of the certificate does not imply that trajectories are unbounded.

Proposition~\ref{prop:cognition_iss_bound} strengthens this forward-invariance result by establishing an explicit input-to-state bound on $\mathcal{X}_R$. It characterizes how bounded perceptual inputs propagate through the cognition dynamics and identifies parameter regimes in which self-inhibition dominates recurrent and state-dependent amplification.

\begin{proposition}\textbf{Input-state bound on the forward invariant set $\mathcal{X}_{R}$:}
    \label{prop:cognition_iss_bound}
    Suppose that $\mathcal{X}_R$ is defined via \eqref{eq:forward_invariant_XR} under the conditions of Proposition~\ref{prop:cognition_forward_invariance} and is thus forward invariant. Define
    \begin{align*}
        \beta \coloneqq \Delta t \; \overline{\Theta}, \qquad \alpha_R \coloneqq 
        1 + \Delta t \left( \overline{\Xi} R + \overline{\Phi} + \overline{\Psi} - \underline{\Lambda} \right).
    \end{align*}
    If $0\leq \alpha_R < 1$, then for all $\bx(0) \in \mathcal{X}_R$ and all \gls{npe} sequences $\left\{\bxnpe(i) \right\}_{i\ge 0}$ (sequence of inputs to the cognition module)
satisfying $\xnpebartraj < \infty$, the cognition state satisfies 
    \begin{align}
        \label{eqn:iss_bound_optionA}
        \norm{\bx(k)} \leq \alpha_R^k \norm{\bx(0)} + \frac{\beta}{1 - \alpha_R} \xnpebartraj.
    \end{align}
\end{proposition}

The contraction condition $\alpha_R<1$ is equivalent to
\begin{align*}
    \underline{\Lambda}
    >
    \overline{\Phi}
    +
    \overline{\Psi}
    +
    \overline{\Xi}R.
\end{align*}
Thus, on $\mathcal{X}_R$, the weakest admissible self-inhibition must dominate the baseline, input-gated, and state-gated coupling. Proposition~\ref{prop:cognition_iss_bound} directly leads to the following \gls{iss} result:

\glsreset{iss}
\begin{corollary}\textbf{Input-to-state stability on $\mathcal{X}_R$:}
    \label{cor:cognition_iss}
    Under the conditions of Proposition~\ref{prop:cognition_iss_bound}, the cognition dynamics are \gls{iss} on $\mathcal{X}_R$. 
    Specifically, the estimate in \eqref{eqn:iss_bound_optionA} implies an input-to-state bound of the form 
    \begin{align*}
        \norm{\bx(k)} \le \betaiss \left(\norm{\bx(0)},k \right) + \gammaiss\left(\sup_{i\ge 0}\norm{\bxnpe(i)}\right),
    \end{align*}
with $\betaiss(s,k)=\alpha_R^k s$ and $\gammaiss(s)=\dfrac{\beta}{1-\alpha_R}s$.%
\end{corollary}

\subsection{Decision-making module}
\noindent
Lemmas~\ref{lem:intention_bibo_lip} and~\ref{lem:action_bibo_lip} establish the key regularity properties of the decision-making mappings. 
Lemma~\ref{lem:intention_bibo_lip} shows that each intention intensity is Lipschitz continuous in the goal and belief inputs, so small changes in these cognition variables cannot produce abrupt changes in intention formation. 

\begin{lemma}\textbf{Single-step Lipschitz continuity of the intention formation mapping:}
    \label{lem:intention_bibo_lip}
    Fix $i\in\mathcal G$ and suppose that $\mathcal{X}_{R}$ is defined via \eqref{eq:forward_invariant_XR} under the conditions of Proposition~\ref{prop:cognition_forward_invariance} and is thus forward invariant. 
    Let $\xgiset$ and $\xbset$ denote the projected compact subsets of $\mathcal{X}_R$ corresponding to the 
    goal state and the belief state vector, respectively. 
    Under Assumption~\ref{ass:intentionformation}, $\zii: \xgiset \times \xbset \to \ziiset$ defined by \eqref{eqn:intentionselection} is Lipschitz continuous on the admissible compact goal-belief domain inherited from $\mathcal{X}_R$.
\end{lemma}

Lemma~\ref{lem:action_bibo_lip} provides an analogous single-step Lipschitz bound for the action selection stage, showing that each action activation varies smoothly with the intention vector. Together, these bounds imply that the decision-making module behaves as a regular bounded-gain mapping.

\begin{lemma}\textbf{Single-step Lipschitz continuity of the action selection mapping:}
    \label{lem:action_bibo_lip}
    Fix $i\in\mathcal G$. Under Assumption~\ref{ass:actionselection}, the action selection mapping $\bzi \mapsto \zai \left(\bzi \right)$ 
    defined via \eqref{eqn:actionselection} is Lipschitz continuous on $\ziset$, as introduced in Definition~\ref{def:actionselection}. 
\end{lemma}

\section{Numerical analysis and closed-loop implementation}
\label{sec:simandanalysis}
\noindent
This section has two aims. First, we test whether the instantiated model exhibits the qualitative dependencies predicted by the neuro-cognitive mechanisms in Section~\ref{sec:modulararchitecture}. Second, we illustrate how the model can be embedded in a closed-loop controller that adapts interaction dynamics from partial feedback.

\subsection{Working hypotheses for model behavior}
\label{sec:numericalanalysistestablepredictions}
\noindent
We formulate six working hypotheses to assess whether the numerical instantiation behaves consistently with the mechanisms introduced in Section~\ref{sec:modulararchitecture}; these hypotheses do not constitute independent empirical validation. Hypotheses~\ref{hypothesis1}–\ref{hypothesis2}, \ref{hypothesis3}–\ref{hypothesis4}, and \ref{hypothesis5}–\ref{hypothesis6} concern perception, cognition, and decision-making, respectively.

\begin{enumerate}[leftmargin=*,  label=\textbf{H\arabic*}]
    \item \label{hypothesis1} \textbf{Contextual normalization attenuates attentional gain:}
    Normalization-based accounts of attention and sensory competition predict that contextual suppression reduces the influence of perturbations in individual sensory channels \cite{b61,b62,b63}. Increasing the normalization offset $\betaasl$ or pooling weights $\gammaaslm$ in \eqref{eqn:attentionalselectionG} is therefore expected to reduce sensory sensitivity to perturbations and produce smoother attentional selection trajectories $\zasl(\cdot)$.
    \item \label{hypothesis2} \textbf{Relative precision modulates perceptual update gain:} Predictive-coding accounts propose that perceptual estimates are updated by prediction errors weighted by sensory precision, prior precision, and attention-related gain \cite{b45,b46,b47,b48,b49}. This is captured by the convex-combination update in \eqref{eqn:nperecursionconvex} and the effective update weight in \eqref{eqn:predictiveinferencealpha}. Since $\alphapil(k)$ increases with sensory precision $\phipil$ and attention gain $\chipil$, but decreases with prior precision $\thetapil$, stronger sensory precision or attention should improve tracking while increasing reactivity, whereas stronger prior precision should yield smoother but slower perceptual estimates.
    \item \label{hypothesis3} \textbf{Recurrent coupling and self-inhibition oppositely modulate cognitive stability:} Network-level accounts of cognition, including \gls{dcm}, emphasize that directed couplings can amplify distributed state activity, whereas intrinsic self-dynamics regulate stability \cite{b75,b76,b77,b78,b79,b80}. Recurrent and input-gated couplings are represented by $\Phi$ and $\Psi_\ell$ in \eqref{eqn:gimc}, while self-inhibition is represented by \eqref{eqn:cognitionmodulefiic}. Stronger recurrent coupling is expected to amplify cognitive state excursions and reduce local stability margins, whereas stronger self-inhibition should damp trajectories and improve stability. Input-gated coupling is expected to modulate these effects based on the current perceptual context.
    \item \label{hypothesis4} \textbf{Additive perceptual drive modulates cognitive response magnitude:} In \gls{dcm}, exogenous inputs can enter additively and shift the forced response of the state dynamics \cite{b75,b76}. Increasing the scale of $\Theta\bxnpe(k)$ in \eqref{eqn:dcmnonlinear} is expected to increase the forced cognitive response under fixed perceptual input, with weaker effects on local stability than recurrent coupling $\Phi$ or self-inhibition $\kappa$.
    \item \label{hypothesis5} \textbf{Belief gating modulates goal-to-intention conversion:} Accounts of intention formation emphasize that goals are expressed as intentions depending on beliefs about feasibility, desirability, and control \cite{b16,b29,b32,b54}. Shifting the belief gate in $\gisi(\cdot)$ in \eqref{eqn:intentionselection} toward activation should increase the expression of goal salience as intention, whereas shifting it toward suppression should reduce this expression. Scaling supportive and suppressive belief weights is expected to modulate this conversion in an operating-regime-dependent manner.
    \item \label{hypothesis6} \textbf{Competition and thresholding modulate action commitment:} Basal-ganglia-inspired accounts of action selection propose that selected actions are facilitated, competing alternatives are suppressed, and conflict can increase the threshold for commitment \cite{b57,b58,b59,b60,b98,b99}. These effects are represented by the self-facilitation, competition, and threshold terms in \eqref{eqn:decisionmakingfasi}. Stronger competition should suppress co-activation of competing actions, while a higher threshold $\phiasi$ should reduce action activation and delay or weaken commitment, as reflected by lower winner margins and smoother action trajectories.
\end{enumerate}
Section~\ref{sec:parametersensitivityanalyses} evaluates these hypotheses using synthetic inputs, candidate functions, baseline parameters, and trajectory-level metrics.

\subsection{Synthetic input signals}
\glsreset{sa}
\glsreset{da}
\glsreset{sv}
\glsreset{dv}
\noindent
We use four synthetic input families to evaluate the model under controlled sensory regimes, specified in Table~\ref{tab:environment_signal_families}. They vary by modality, auditory versus visual, and temporal structure, stationary versus dynamic, yielding \gls{sa}, \gls{da}, \gls{sv}, and \gls{dv} inputs. These inputs are not intended as full generative models of natural scenes, but as controlled excitations for parameter sensitivity analysis. The \gls{sa} and \gls{sv} families model stationary mean-reverting inputs, the \gls{da} family combines channel-specific and shared sinusoidal modulation, and the \gls{dv} family represents time-varying visual input driven by motion, occlusion, or scene changes \cite{b86,b87,b88,b89,b90}.

\begin{table}[t]
\centering
\caption{Synthetic input families used in the sensitivity analyses. Each family generates a raw channel-wise trajectory $\uraw_\ell(k)$, which is mapped into the bounded admissible range $[0,\overline{u_\ell}]$ before attentional selection via
$u_{\ell}(k)=\overline{u_{\ell}}\left(1-\exp[-a\max\{\uraw_{\ell}(k),0\}]\right)$.
The \gls{sa} and \gls{sv} families are mean-reverting stochastic processes, the \gls{da} family combines channel-specific and shared sinusoidal modulation, and the \gls{dv} family tracks a sinusoidally varying target with persistent perturbations.}
\label{tab:environment_signal_families}
\small
\setlength{\tabcolsep}{6pt}
\renewcommand{\arraystretch}{1.25}
\begin{tabularx}{\linewidth}{@{}p{0.34\linewidth} X@{}}
\toprule
\textbf{Input family} & \textbf{Representative synthetic form} \\
\midrule

Stationary Audio (SA)
&
\(\displaystyle
\begin{aligned}
\uraw_{\ell}(k+1)
&=
\mu_{\ell}^{\mathrm{SA}}
+
\rho^{\mathrm{SA}}
\left(
\uraw_{\ell}(k)-\mu_{\ell}^{\mathrm{SA}}
\right) \\
&+ 
\epsilon_{\ell}^{\mathrm{SA}}(k)
\end{aligned}
\)
\\

\midrule

Dynamic Audio (DA)
&
\(\displaystyle
\begin{aligned}
\uraw_{\ell}(k)
&=
\mu_{1,\ell}^{\mathrm{DA}}
+
A_{\ell}^{\mathrm{DA}}
\sin \left(2\pi f_{1,\ell}^{\mathrm{DA}}k+\varphi_{1,\ell}^{\mathrm{DA}}\right)
\\
&+ \mu_{2}^{\mathrm{DA}} A_{\ell}^{\mathrm{DA}} \sin \left(2\pi f_{2}^{\mathrm{DA}}k +\varphi_{2}^{\mathrm{DA}}\right) \\
&+
\epsilon_{\ell}^{\mathrm{DA}}(k)
\end{aligned}
\)
\\

\midrule

Stationary Visual (SV)
&
\(\displaystyle
\begin{aligned}
\uraw_{\ell}(k+1)
&=
\mu_{\ell}^{\mathrm{SV}}
+
\rho^{\mathrm{SV}}
\left(
\uraw_{\ell}(k)-\mu_{\ell}^{\mathrm{SV}}
\right) \\
&+
\epsilon_{\ell}^{\mathrm{SV}}(k)
\end{aligned}
\)
\\

\midrule

Dynamic Visual (DV)
&
\(\displaystyle
\begin{aligned}
\uraw_{\ell}(k+1)
&=
\left(1-\rho^{\mathrm{DV}}\right)\left[
\mu_{\ell}^{\mathrm{DV}} \right.
\\ &+ \left.
A_{\ell}^{\mathrm{DV}}
\sin \left(2\pi f_{\ell}^{\mathrm{DV}}k+\varphi_{\ell}^{\mathrm{DV}}\right)
\right]
\\
&+
\rho^{\mathrm{DV}} \uraw_{\ell}(k)
+
\epsilon_{\ell}^{\mathrm{DV}}(k)
\end{aligned}
\)
\\

\bottomrule
\end{tabularx}
\end{table}

\subsection{Model instantiation}
\label{sec:modelinstantiation}
\noindent
Table~\ref{tab:candidate_functions} summarizes the candidate functions used in the numerical instantiation. These choices satisfy the analytical assumptions in Section~\ref{sec:mathreprhumandecisionmaking}, preserve bounded nonlinear responses, and remain consistent with the mechanisms introduced in Section~\ref{sec:modulararchitecture}. All variables are interpreted in normalized model coordinates. Raw synthetic inputs are mapped component-wise into $[0,\overline{u_\ell}]$ before attentional selection. The resulting \gls{npe} trajectories are bounded by the convex-combination structure \eqref{eqn:nperecursionconvex} of predictive inference, whereas cognition states are not explicitly saturated during simulation. The vector $\boldsymbol{\eta}$ collects the swept numerical parameters used in the sensitivity analyses. Depending on the module, these parameters control response shape, operating range, effective precision, offsets, or coupling scales.

\begin{table}[!t]
\centering
\caption{Candidate functions used in the numerical model instantiations and the corresponding one-at-a-time swept parameters.}
\label{tab:candidate_functions}
\begin{minipage}{\columnwidth}

\hrule
\vspace{0.5em}

\noindent
\makebox[0.67\columnwidth][l]{\textbf{Candidate functions}}%
\makebox[0.04\columnwidth][l]{\textbf{Eqs.}}%
\makebox[0.29\columnwidth][r]{\textbf{Sweep}}

\vspace{0.5em}
\hrule
\vspace{0.8em}

\noindent
\makebox[0.67\columnwidth][l]{\textbf{Attentional selection}}%
\makebox[0.04\columnwidth][l]{\eqref{eqn:perceptualaccessgeneral}--\eqref{eqn:attentionalselectionG}}%
\makebox[0.29\columnwidth][r]{%
  \vtop{%
    \hbox{\(\eta_1,\eta_2,\)}%
    \kern 0.15em
    \hbox{\(\eta_3,\betaasl\)}%
  }%
}

\vspace{-1.0em}
\noindent\(\displaystyle
\begin{aligned}
\fasl(\ul)
&\coloneqq
\frac{
\faslmax \ul^{\eta_1}
}{\ul^{\eta_1}
+ \eta_2^{\eta_1}
}, \\
\gasmtilde(u_m)
&\coloneqq
u_m^{p_m},
\quad p_m\geq 1,\\
\betaasl &\coloneqq \betaasl, \\
\gammaaslm(\eta_3)
&\coloneqq
\eta_3 \gamma_{\ell m}^{\mathrm{ATS,base}} .
\end{aligned}
\)

\vspace{0.8em}
\hrule
\vspace{0.8em}

\noindent
\makebox[0.67\columnwidth][l]{\textbf{Predictive inference}}%
\makebox[0.04\columnwidth][l]{\eqref{eqn:predictiveinferencealpha}}%
\makebox[0.29\columnwidth][r]{\(\eta_4,\thetapil\)}

\vspace{0.3em}
\noindent\(\displaystyle
\begin{aligned}
\fpil(\eta_4,\thetapil)
&\coloneqq
\frac{\eta_4}
{\eta_4+\thetapil},\\
\thetapil &\coloneqq \thetapil, \\
\eta_4
&\coloneqq
\chipil\phipil .
\end{aligned}
\)

\vspace{0.8em}
\hrule
\vspace{0.8em}

\noindent
\makebox[0.67\columnwidth][l]{\textbf{Cognition module}}%
\makebox[0.04\columnwidth][l]{\eqref{eqn:dcmnonlinear}--\eqref{eqn:gimc}}%
\makebox[0.29\columnwidth][r]{%
  \vtop{%
    \hbox{\(\eta_5,\eta_6,\)}%
    \kern 0.15em
    \hbox{\(\eta_7,\kappa\)}%
  }%
}

\vspace{-1.0em}
\noindent\(\displaystyle
\begin{aligned}
\Phi(\eta_5)
&\coloneqq
\eta_5\Phi^{\mathrm{base}},\\
\Psi_\ell(\eta_6)
&\coloneqq
\eta_6\Psi_\ell^{\mathrm{base}},\\
\Theta(\eta_7)
&\coloneqq
\eta_7\Theta^{\mathrm{base}}, \\
\kappa &\coloneqq \kappa.
\end{aligned}
\)

\vspace{0.8em}
\hrule
\vspace{0.8em}

\noindent
\makebox[0.67\columnwidth][l]{\textbf{Intention formation}}%
\makebox[0.04\columnwidth][l]{\eqref{eqn:intentionselection}}%
\makebox[0.29\columnwidth][r]{%
  \vtop{%
    \hbox{\(\eta_8,\eta_9,\)}%
    \kern 0.15em
    \hbox{\(\eta_{10},\eta_{11}\)}%
  }%
}

\vspace{-1.0em}
\noindent\(\displaystyle
\begin{aligned}
\fisi(\xigoals)
&\coloneqq
\frac{
\fisimax \bigl(\max\{0,\xigoals\}\bigr)^{\eta_8}
}{
\bigl(\max\{0,\xigoals\}\bigr)^{\eta_8}
+ \eta_9^{\eta_8}
},\\
\gisi(\xbeliefs)
&\coloneqq
\gisimax S_i^{\mathrm{IF}}
\bigl(H_i^{\mathrm{IF}}(\xbeliefs)\bigr),\\
H_i^{\mathrm{IF}}(\xbeliefs)
&\coloneqq
\eta_{10}
+
\sum_{j \in \mathcal{B}_i^+} w_{ij}^+ x_j^{\mathrm{B}}
-
\sum_{j \in \mathcal{B}_i^-} w_{ij}^- x_j^{\mathrm{B}},\\
w_{ij}^{\pm}(\eta_{11})
&\coloneqq
\eta_{11} w_{ij}^{\pm,\mathrm{base}},
\quad
w_{ij}^{\pm}\geq 0 .
\end{aligned}
\)

\vspace{0.8em}
\hrule
\vspace{0.8em}

\noindent
\makebox[0.67\columnwidth][l]{\textbf{Action selection}}%
\makebox[0.04\columnwidth][l]{\eqref{eqn:decisionmakingfasi}}%
\makebox[0.29\columnwidth][r]{\(\eta_{12}, \phiasi\)}

\vspace{0.3em}
\noindent\(\displaystyle
\begin{aligned}
\gasi(\zii)
&\coloneqq
\gasimax S_i^{\mathrm G}(\zii),\\
\hasi(\bzimini)
&\coloneqq
\hasimax S_i^{\mathrm H} \left(
\sum_{\substack{m\in\mathcal G\\m\neq i}}
\gamma_{im}^{\mathrm{AS}}z_m^{\mathrm I} \right),\\
\gamma_{im}^{\mathrm{AS}}(\eta_{12})
&\coloneqq
\eta_{12}\gamma_{im}^{\mathrm{AS,base}},
\qquad
\gamma_{im}^{\mathrm{AS}}\geq 0, \\
\phiasi &\coloneqq \phiasi.
\end{aligned}
\)

\vspace{0.8em}
\hrule

\end{minipage}
\end{table}

The attentional drive $\fasl(\cdot)$ is modeled as a Naka--Rushton function, and contextual normalization uses rectified power-law pooling with $p_m=2$, consistent with sensory contrast-response and normalization models \cite{b20,b61,b62,b93,b94}. The predictive-inference weight $\fpil(\cdot)$ is instantiated as a relative precision weight, where attended sensory precision $\eta_4=\chipil\phipil$ is normalized against prior precision, following reliability-weighted cue integration \cite{b48,b95,b96}.

The cognition sweeps vary global scales for recurrent coupling, input-gated coupling, additive perceptual drive, and self-inhibition. To isolate these effects, state-gated endogenous couplings $\Xi_q$ are set to zero in the sensitivity analyses and used only in the closed-loop rehabilitation showcase. Intention formation uses a saturating goal-salience function and a bounded sigmoid belief gate with supportive beliefs weighted by $w_{ij}^{+}$ and suppressive beliefs weighted by $w_{ij}^{-}$, while action selection uses bounded monotone self-facilitation and competition terms. Together, these candidate functions instantiate the abstract mappings introduced in Section~\ref{sec:mathreprhumandecisionmaking} and satisfy the assumptions used in Section~\ref{sec:stabilityanalysis}.

\subsection{Parameter sensitivity analyses}
\label{sec:parametersensitivityanalyses}
\noindent
We next examine how the swept parameters affect sensitivity, variability, tracking, and decision consistency. Parameters are held fixed within each simulation run, so the analyses isolate static parameter effects rather than time-dependent adaptation. One-at-a-time sweeps vary only the parameter under study, with all remaining parameters fixed at the baseline values reported in Table~\ref{tab:baseline_simulation_parameters} in \ref{appendix:baselinesensitivityparams}.

\subsubsection{Evaluation metrics}
\glsreset{io}
\noindent
Table~\ref{tab:sensitivity_metrics} in \ref{appendix:sensitivitymetrics} defines the metrics used to evaluate Hypotheses~\ref{hypothesis1}–\ref{hypothesis6}. For static mappings, we report a sampled \gls{io} gain and a mean step increment, which measure perturbation transmission and temporal variability, respectively. For predictive inference, we additionally report mean and terminal tracking errors between the \gls{npe} trajectory and the attentional selection target, capturing the trade-off between tracking accuracy and smoothing. For cognition, we report sampled finite-horizon \gls{is} gain, maximum state norm, and mean step increment, which summarize input-to-state propagation, state excitation, and temporal variability. For action selection, winner margin and switch count quantify decisiveness and temporal consistency.

Parameter-sweep results are summarized using the endpoint-change score
\begin{align}
\label{eqn:srel}
    \srel = 100 \left( \frac{m^{\mathrm{f}}-m^{\mathrm{s}}}{\max_q \abs{m_q}}\right)
\end{align}
where $m^{\mathrm{s}}$ and $m^{\mathrm{f}}$ are the metric values at the first and last sampled parameter values, and $m_q$ is the value at the $q$-th sample. Positive and negative scores indicate endpoint increases and decreases, respectively, while magnitudes near $100$ indicate that the endpoint change is comparable to the largest metric magnitude observed within the sweep. The score thus provides a directional effect-size measure without becoming artificially large when $m^{\mathrm{s}}$ is near zero. Because it is most informative for approximately monotone sweeps, curves containing both increasing and decreasing segments are marked as non-monotone (NM) after changes below a numerical tolerance are ignored by the trend classifier.

\subsubsection{Attentional selection}
\label{sec:attentionalselectionnumericalanalysis}
\noindent
We vary the Naka--Rushton parameters $\eta_1$ and $\eta_2$, the contextual-normalization scale $\eta_3$, and the normalization offset $\betaasl$ (Table~\ref{tab:candidate_functions}) to evaluate Hypothesis~\ref{hypothesis1}. Figure~\ref{fig:montecarloattentionalselection_heatmap} shows that increasing either $\betaasl$ or $\eta_3$ consistently reduces both sampled \gls{io} gain and mean step increment across all input families. Thus, the divisive contextual-normalization term in \eqref{eqn:attentionalselectionG} robustly attenuates sensitivity and temporal reactivity in the attentional selection map.

The Naka–Rushton parameters $\eta_1$ and $\eta_2$ tune the steepness and operating range of the sensory-drive response. Increasing $\eta_2$ generally reduces sampled \gls{io} gain by shifting the response curve toward higher inputs. By contrast, the sampled \gls{io} gain is non-monotone in $\eta_1$ (Figure~\ref{fig:montecarloattentionalselection_ngain}). Low-to-intermediate values of $\eta_1$ yield the lowest sampled amplification, whereas larger values make the response more switch-like. Inputs outside the half-saturation region around $\eta_2$ produce little change, while inputs near this transition region can be amplified. Thus, increasing $\eta_1$ can either suppress small fluctuations or amplify perturbations near the steep part of the response curve. Most sampled gains remain below unity, indicating overall perturbation attenuation.

Overall, these results support Hypothesis~\ref{hypothesis1}. Contextual normalization makes attentional selection more conservative, whereas $\eta_1$ and $\eta_2$ regulate the sharpness and operating range of sensory-drive sensitivity.

\begin{figure*}[!t]
    \centering
    \begin{subfigure}[t]{0.45\linewidth}
        \centering
        \includegraphics[width=.9\linewidth]{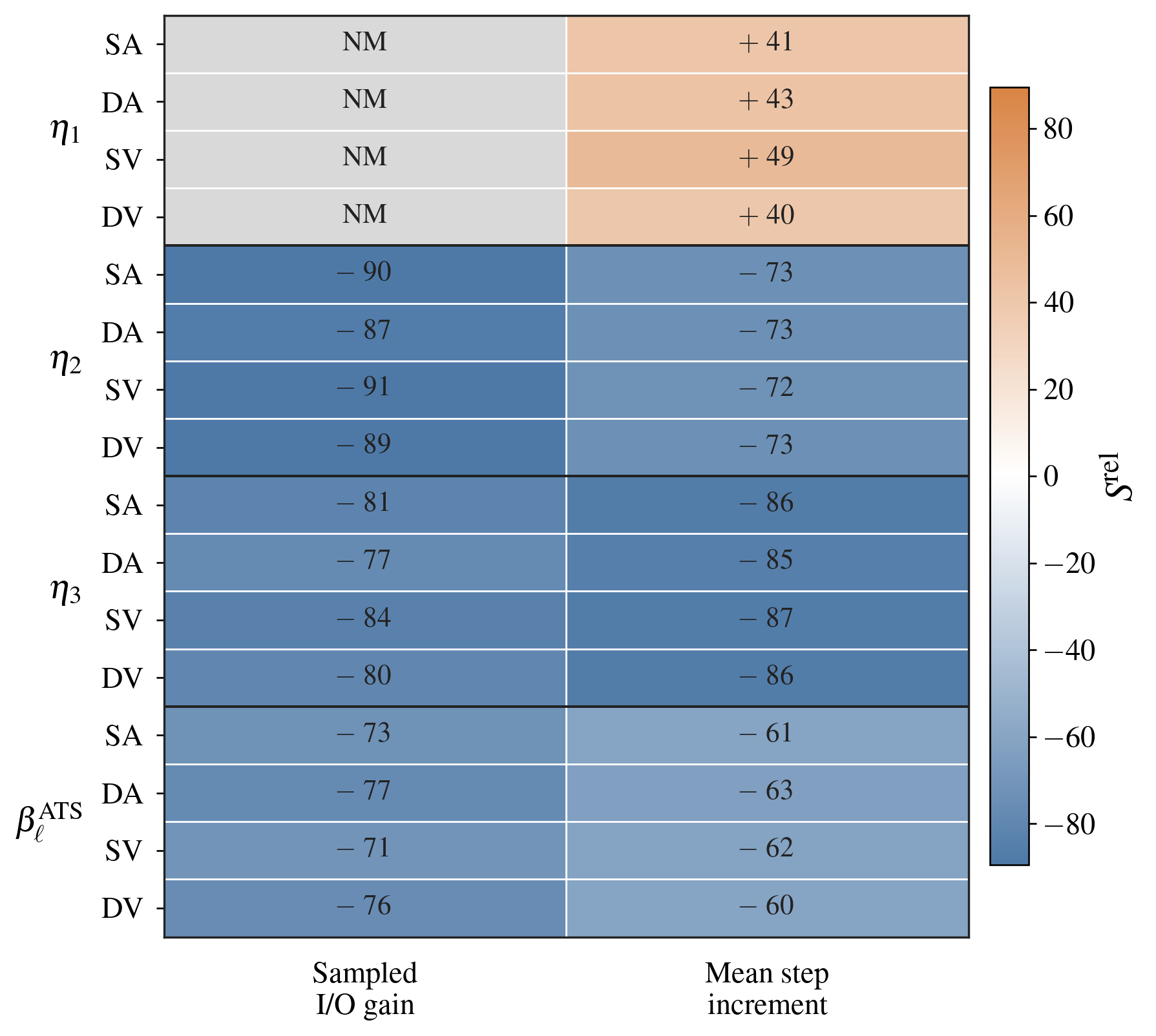}
        \caption{Heatmap of attentional selection parameter effects across input families.}
        \label{fig:montecarloattentionalselection_heatmap}
    \end{subfigure}\hfill
    \begin{subfigure}[t]{0.46\linewidth}
        \centering
        \raisebox{-0.08cm}{%
        \includegraphics[width=.8\linewidth]{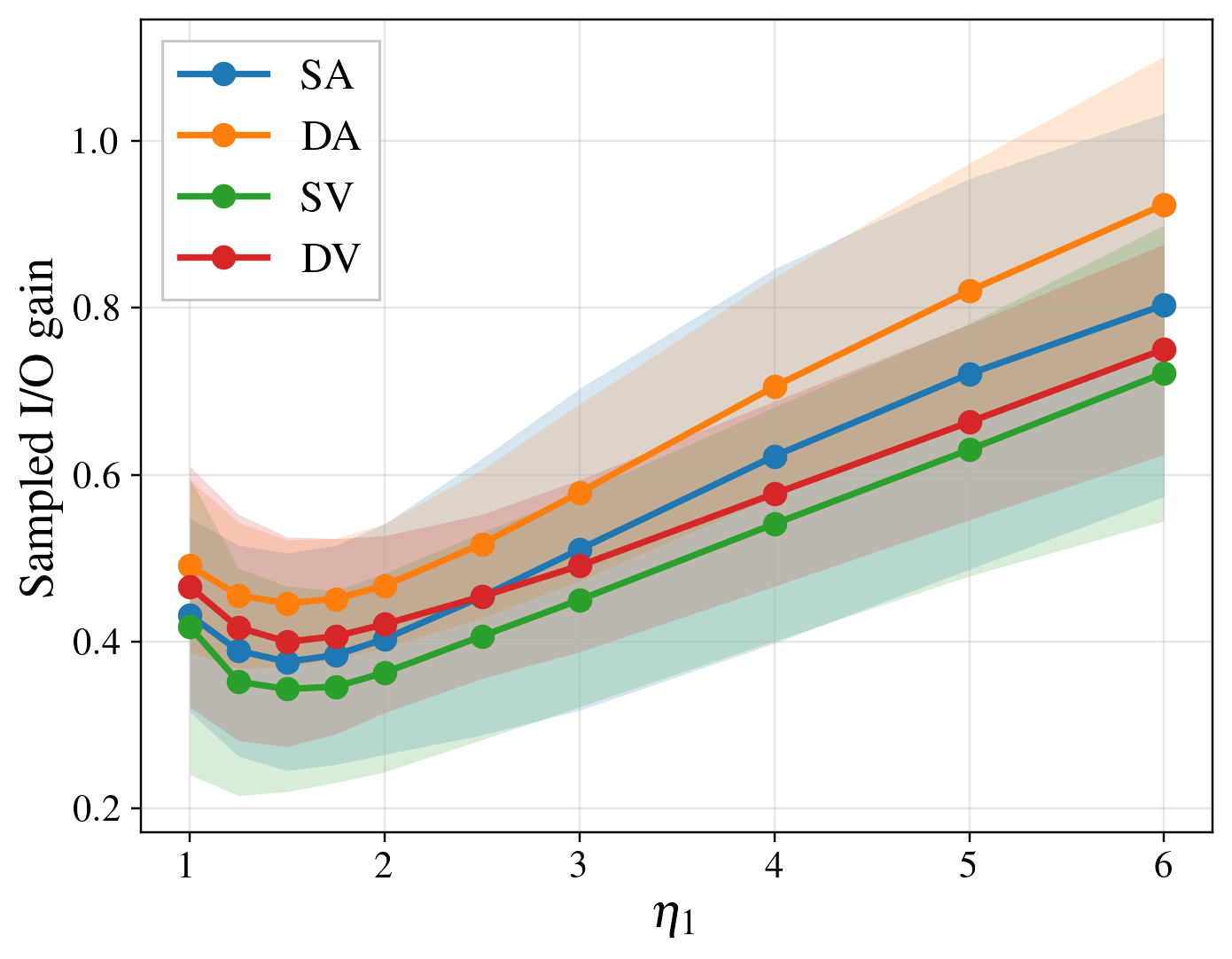}
        }
        \captionsetup{skip=0.4em}
        \caption{Sampled \gls{io} gain as a function of $\eta_1$.}
        \label{fig:montecarloattentionalselection_ngain}
    \end{subfigure}
    \caption{
    Sensitivity analysis of attentional selection.
    (\subref{fig:montecarloattentionalselection_heatmap}) Endpoint-change heatmap for one-at-a-time parameter sweeps across input families. Colors encode $\srel$ as defined in \eqref{eqn:srel}.
    (\subref{fig:montecarloattentionalselection_ngain}) Sampled \gls{io} gain as a function of $\eta_1$, showing a non-monotone response with a local minimum at low-to-intermediate response steepness.
    }
    \label{fig:montecarloattentionalselection}
\end{figure*}

\subsubsection{Predictive inference}
\noindent
We vary attended sensory precision $\eta_4$ and prior precision $\thetapil$ to evaluate Hypothesis~\ref{hypothesis2}. As shown in Figure~\ref{fig:montecarlopredictiveinference_heatmap}, increasing $\eta_4$ reduces mean and terminal tracking error, but increases mean step increment and sampled finite-horizon \gls{is} gain. Thus, stronger attended sensory precision improves tracking of the attentional selection target while transmitting more temporal fluctuation into the \gls{npe} state.

Increasing $\thetapil$ has the opposite effect: the \gls{npe} trajectory becomes smoother and less reactive, but deviates more from the instantaneous attentional selection target. Figure~\ref{fig:montecarlopredictiveinference_tradeoff} visualizes this tracking--reactivity trade-off. These results support Hypothesis~\ref{hypothesis2}: relative precision tunes the balance between perceptual adaptivity and smoothing.

\begin{figure*}[!t]
    \centering
    \begin{subfigure}[t]{0.50\linewidth}
        \centering
        \hspace*{-0.72cm}%
        \includegraphics[width=.78\linewidth]{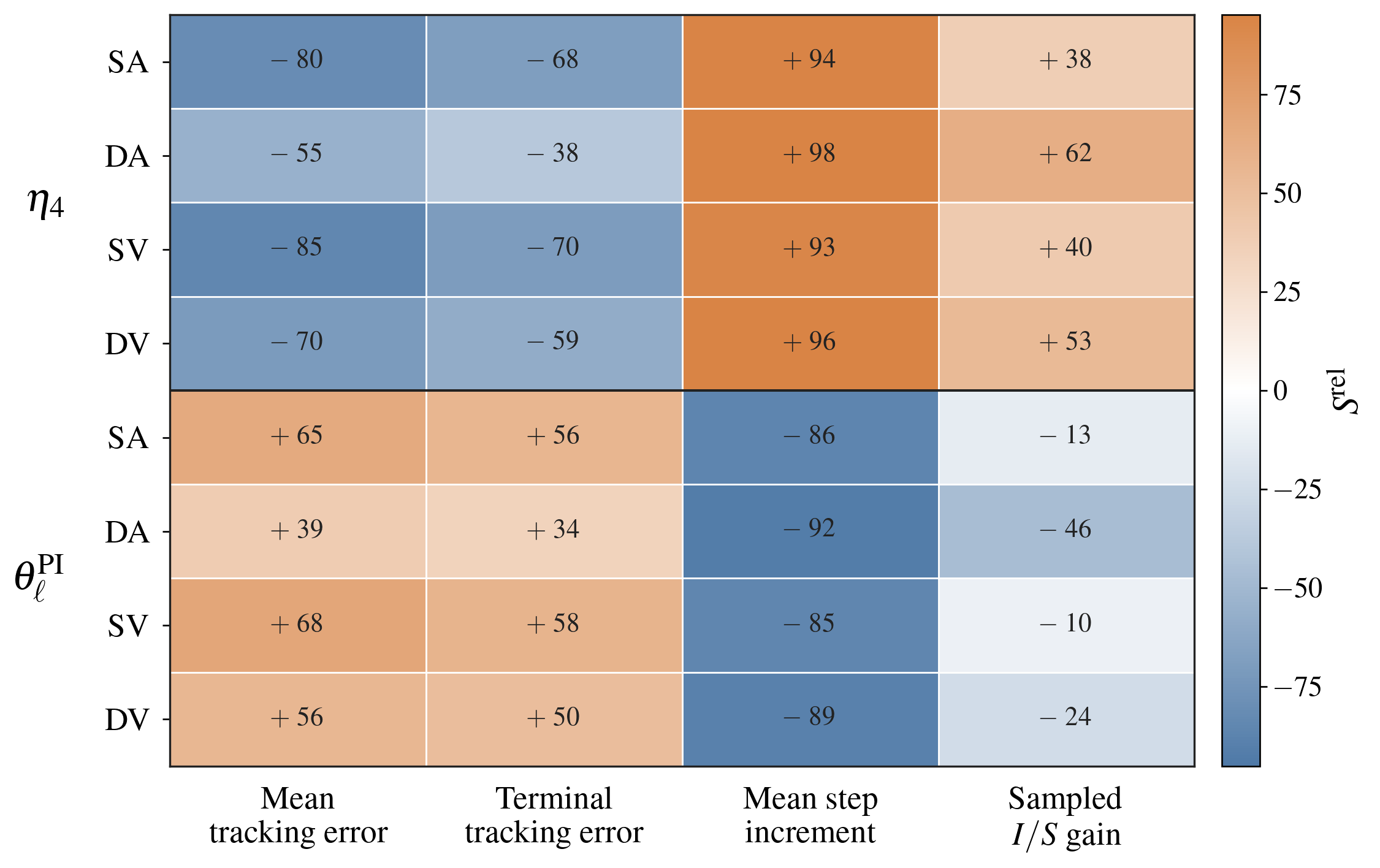}
        \captionsetup{skip=1.2em}
        \caption{Heatmap of predictive inference parameter effects across input families.}
        \label{fig:montecarlopredictiveinference_heatmap}
    \end{subfigure}\hfill
    \begin{subfigure}[t]{0.48\linewidth}
        \centering
        \raisebox{-0.22cm}{%
        \includegraphics[width=.8\linewidth]{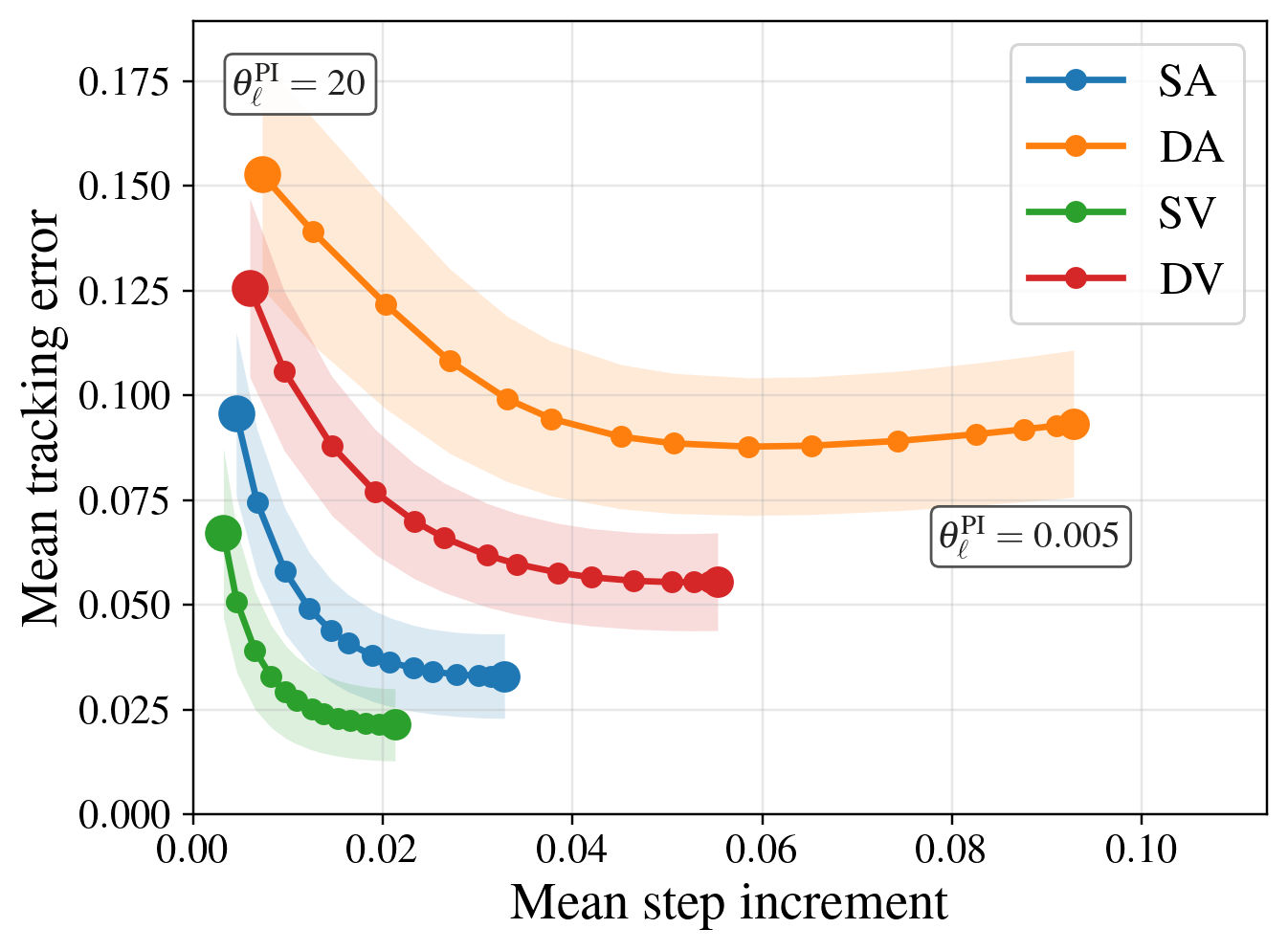}
        }
        \caption{Tracking-reactivity trade-off under a sweep of prior precision $\thetapil$.}
        \label{fig:montecarlopredictiveinference_tradeoff}
    \end{subfigure}
    \caption{
    Sensitivity analysis of predictive inference.
    (\subref{fig:montecarlopredictiveinference_heatmap}) Endpoint-change heatmap for one-at-a-time parameter sweeps across input families. Colors encode $\srel$ as defined in \eqref{eqn:srel}.
    (\subref{fig:montecarlopredictiveinference_tradeoff}) Representative $\thetapil$-sweep showing the trade-off between mean tracking error and mean step increment of the \gls{npe} trajectory.
    }
    \label{fig:montecarlopredictiveinference}
\end{figure*}

\subsubsection{Cognition module}
\noindent
We vary recurrent coupling $\eta_5$, input-gated coupling $\eta_6$, additive perceptual drive $\eta_7$, and self-inhibition $\kappa$ to evaluate Hypotheses~\ref{hypothesis3}--\ref{hypothesis4}. We use both trajectory sweeps under time-varying inputs and constant-input scans. In the constant-input scans, the terminal trajectory average, defined in Table~\ref{tab:sensitivity_metrics}, is used as an estimated steady response $\hat{\bx}^{\star}$. We then evaluate the spectral radius $\rho(J^\star)$ of the discrete-time Jacobian at this estimate as a local stability indicator. When $\hat{\bx}^{\star}$ approximates a fixed point, $\rho(J^\star)<1$ indicates local asymptotic stability.

Figure~\ref{fig:montecarlocognition_heatmapmain} shows that increasing $\eta_5$ or $\eta_7$ increases state excursions, temporal variation, and sampled \gls{is} gain, whereas increasing $\kappa$ damps these effects. The input-gated scale $\eta_6$ has weaker and more family-dependent effects in the baseline operating regime.

The constant-input scans in Figure~\ref{fig:montecarlocognition_eqscanmain} separate amplification from forcing. Increasing $\eta_5$ raises both steady-state magnitude and $\rho(J^\star)$, indicating stronger sustained activation and reduced local stability margin. Increasing $\eta_7$ mainly raises steady-state magnitude while leaving $\rho(J^\star)$ nearly unchanged, consistent with additive forcing. Increasing $\kappa$ reduces both steady-state magnitude and $\rho(J^\star)$, confirming the stabilizing role of self-inhibition.

The two-dimensional sweep in Figure~\ref{fig:montecarlocognition_stabilitymaps} further shows that stronger recurrent amplification remains locally stable only when accompanied by sufficient self-inhibition. Thus, Hypothesis~\ref{hypothesis3} is supported for recurrent coupling and self-inhibition, and Hypothesis~\ref{hypothesis4} is supported for additive drive. The weaker effect of $\eta_6$ suggests that input-gated coupling becomes distinguishable only in operating regimes that sufficiently activate the corresponding bilinear term.

\begin{figure*}
    \centering
    \begin{subfigure}[t]{0.48\linewidth}
        \centering
        \includegraphics[width=.85\linewidth]{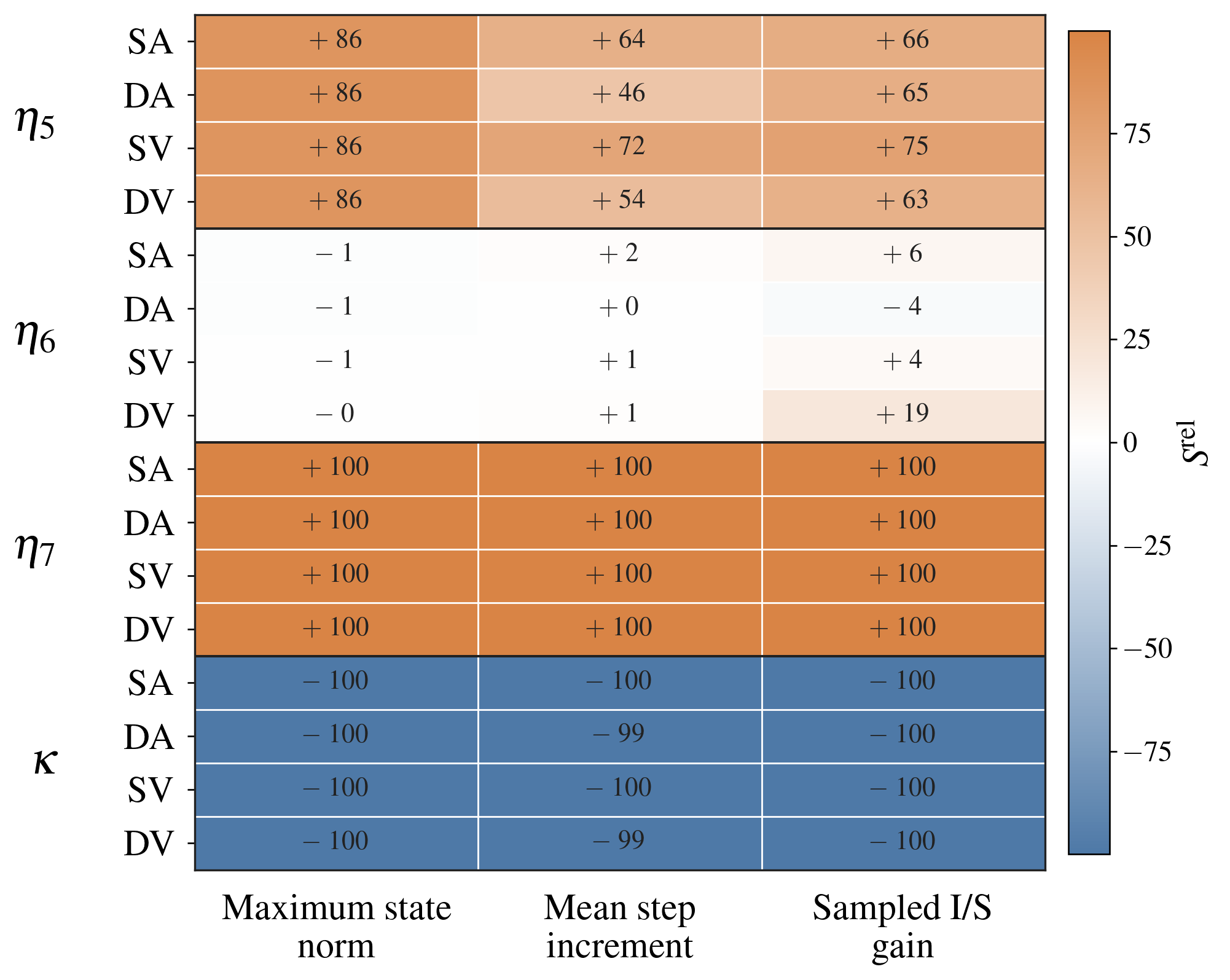}
        \caption{Heatmap of cognition parameter effects across input families.} \label{fig:montecarlocognition_heatmapmain}
    \end{subfigure}\hfill
    \begin{subfigure}[t]{0.50\linewidth}
        \centering
        \raisebox{0.04cm}{%
        \includegraphics[width=.95\linewidth]{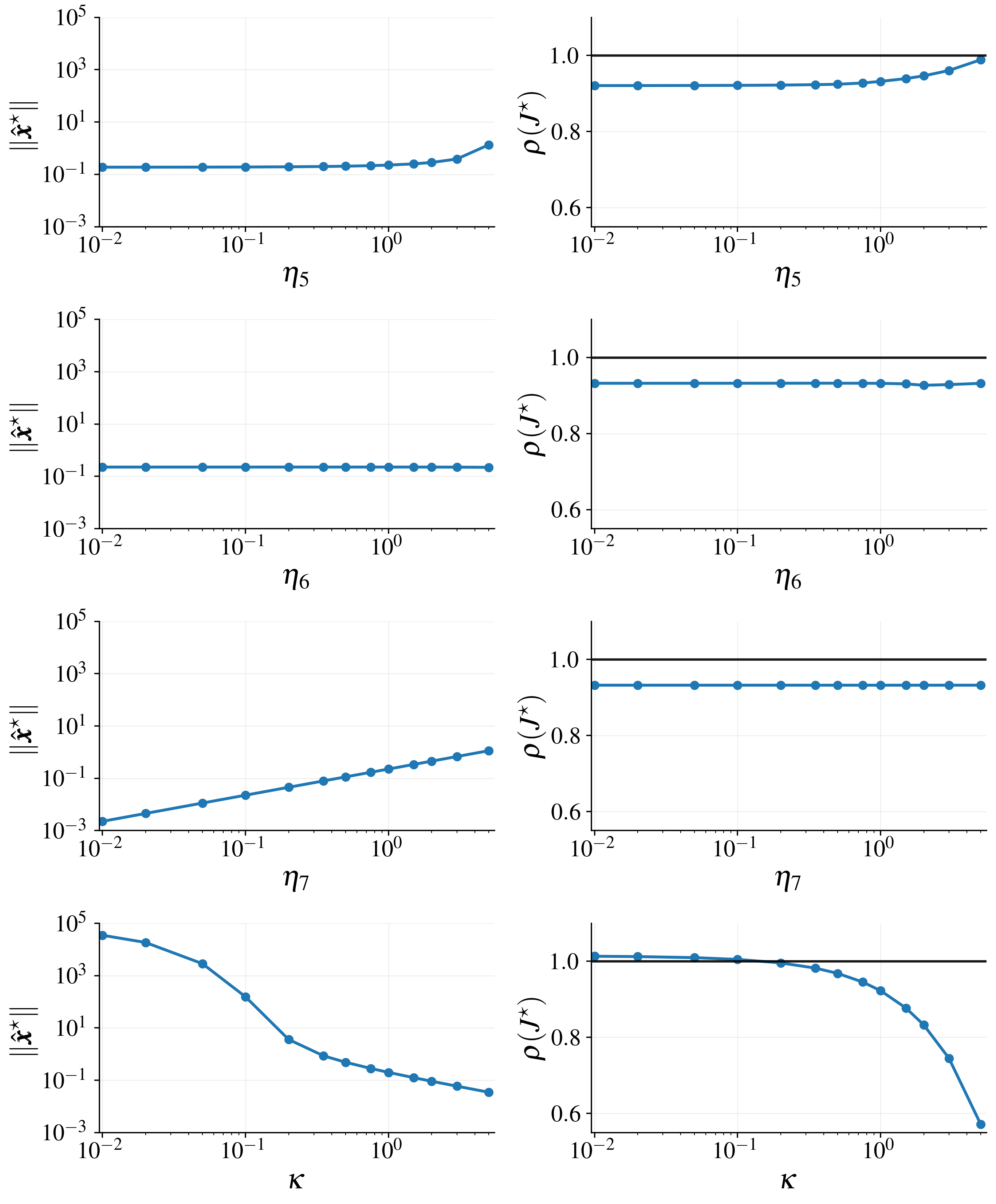}
        }
        \captionsetup{skip=0.6em}
        \caption{Constant-input steady-state and local stability scans.} \label{fig:montecarlocognition_eqscanmain}
    \end{subfigure}
    \caption{
    Sensitivity and constant-input local stability analysis of the cognition module.
    (\subref{fig:montecarlocognition_heatmapmain}) Endpoint-change heatmap for maximum state norm, mean step increment, and sampled finite-horizon \gls{is} gain across input families. Colors encode $\srel$ as defined in \eqref{eqn:srel}.
    (\subref{fig:montecarlocognition_eqscanmain}) One-at-a-time constant-input sweeps of $\eta_5$, $\eta_6$, $\eta_7$, and $\kappa$. Left panels show the terminal steady-state magnitude; right panels show the spectral radius $\rho(J^\star)$ of the constant-input Jacobian. The black line marks $\rho(J^\star)=1$.
    }
    \label{fig:montecarlocognition_main}
\end{figure*}

\begin{figure*}[!htbp]
    \centering
    \begin{subfigure}[t]{0.49\linewidth}
        \centering
        \includegraphics[width=.85\linewidth]{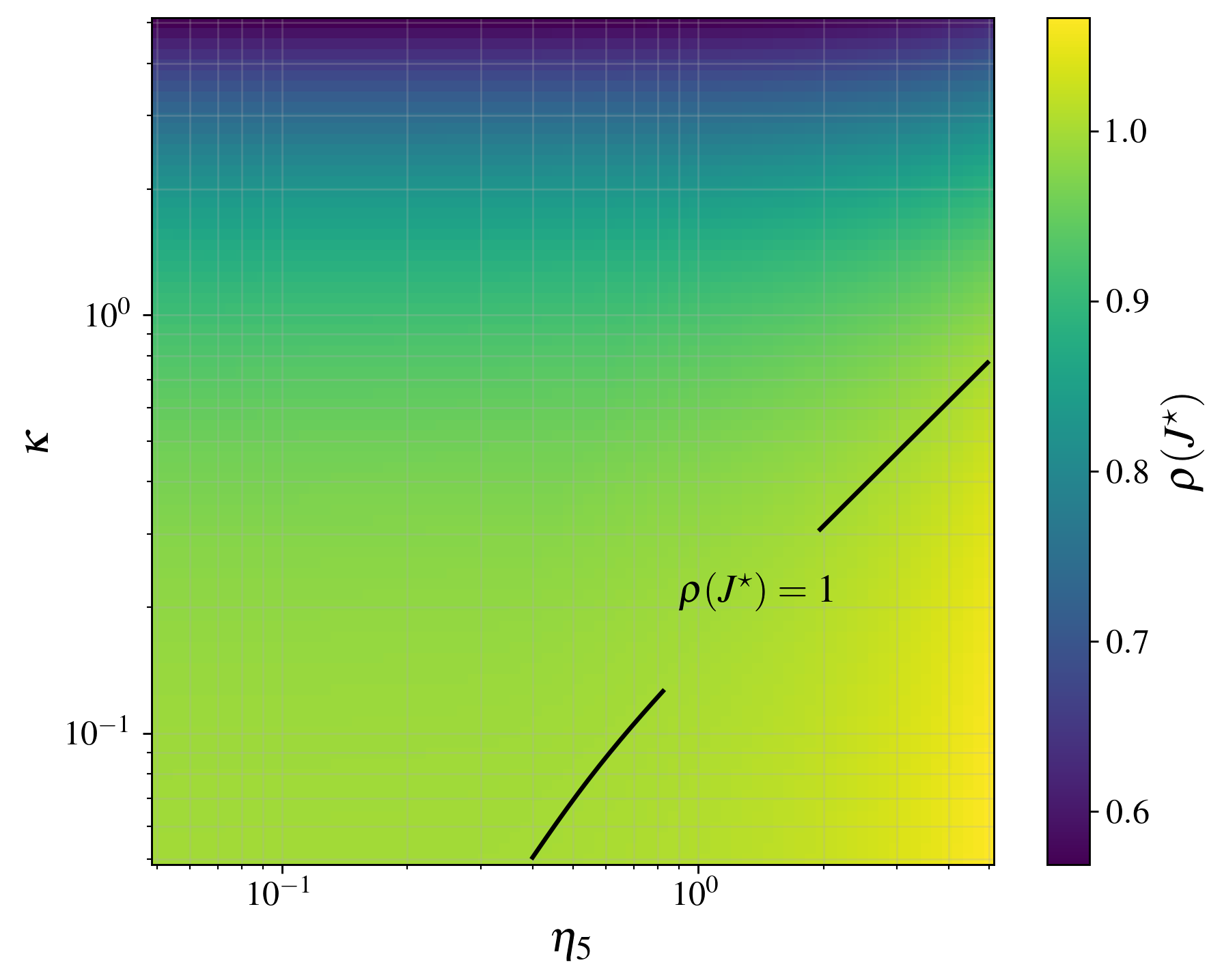}
        \caption{Spectral radius over $(\eta_5,\kappa)$.}
        \label{fig:montecarlocognition_spectralradiusmap}
    \end{subfigure}\hfill
    \begin{subfigure}[t]{0.49\linewidth}
        \centering
        \includegraphics[width=.85\linewidth]{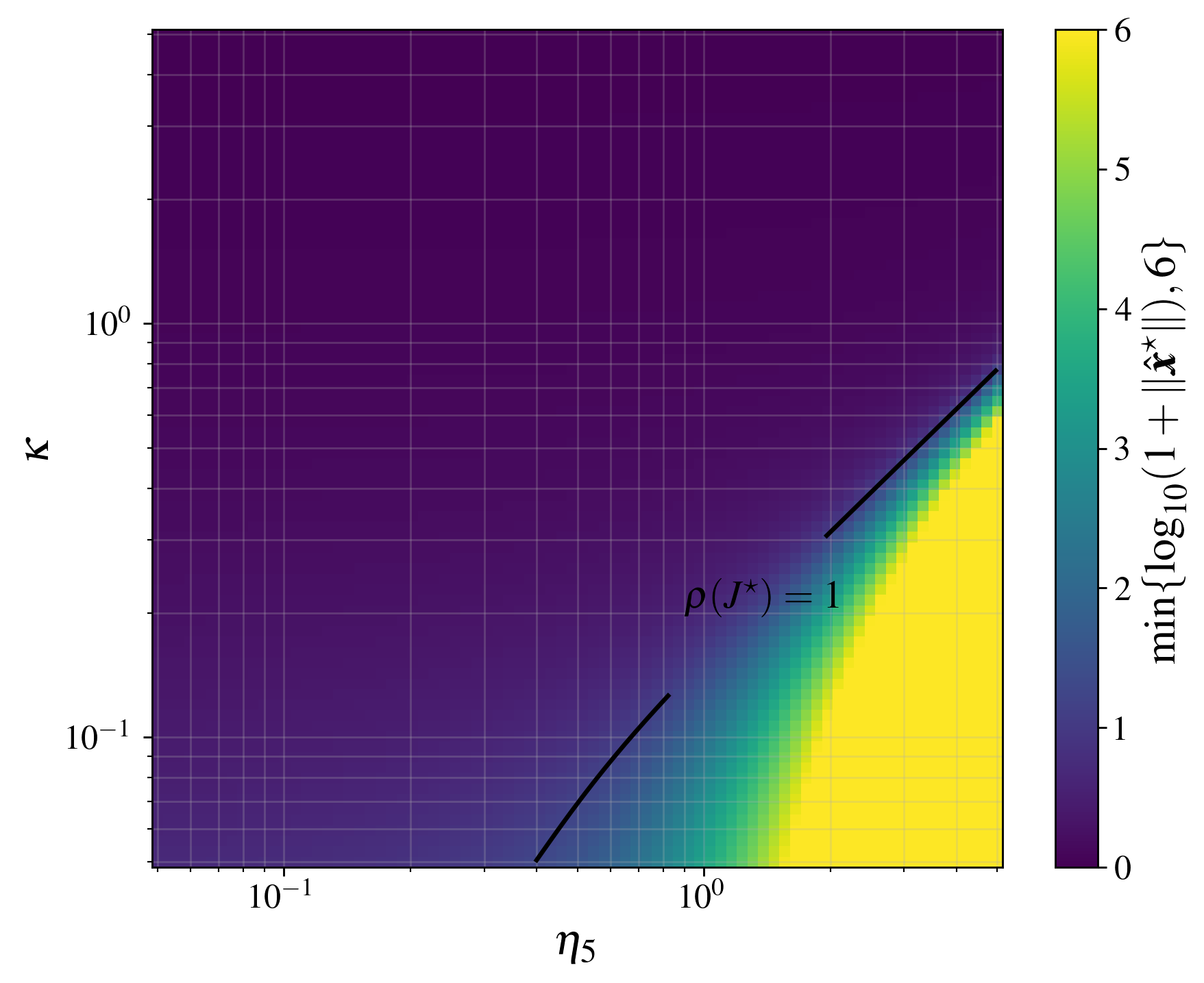}
        \caption{Clipped logarithmic terminal-response magnitude over $(\eta_5,\kappa)$.}
        \label{fig:montecarlocognition_eqnormmap}
    \end{subfigure}
    \caption{
    Two-dimensional constant-input sweep of the cognition module over recurrent coupling scale $\eta_5$ and self-inhibition gain $\kappa$. All other cognition parameters are held fixed at their baseline values in Table~\ref{tab:baseline_simulation_parameters} in \ref{appendix:baselinesensitivityparams}; the black contour denotes $\rho(J^\star)=1$.
    (\subref{fig:montecarlocognition_spectralradiusmap}) Increasing $\eta_5$ reduces the local stability margin, whereas increasing $\kappa$ restores damping.
    (\subref{fig:montecarlocognition_eqnormmap}) The color scale reports $\min \{ \log_{10}(1+\norm{\hat{\bx}^{\star}}),6 \}$. Saturated yellow regions indicate terminal response magnitudes at or above the plotting cutoff.
    }
\label{fig:montecarlocognition_stabilitymaps}
\end{figure*}

\subsubsection{Intention formation}
\noindent
We vary $\eta_8$, $\eta_9$, $\eta_{10}$, and $\eta_{11}$ to evaluate Hypothesis~\ref{hypothesis5}. Figure~\ref{fig:montecarlointention_heatmap} shows that belief-gate parameters and goal-salience parameters have distinct effects. Increasing the belief offset $\eta_{10}$ produces large positive endpoint changes in sampled \gls{io} gain and mean step increment, and increasing the belief-weight scale $\eta_{11}$ produces similar but weaker effects. Thus, shifting the belief gate toward activation increases the sensitivity and variability of intention formation.

In contrast, increasing the goal-salience parameters $\eta_8$ and $\eta_9$ reduces both metrics in the sampled operating regime, reflecting the operating range of the saturating goal-salience nonlinearity: larger $\eta_9$ shifts the response toward higher goal intensities, while larger $\eta_8$ sharpens the transition and reduces perturbation transmission. Figure~\ref{fig:intention_belief_surface} confirms the qualitative gating mechanism: supportive beliefs increase intention intensity, whereas suppressive beliefs reduce it. Together, these results support Hypothesis~\ref{hypothesis5}.

\begin{figure*}[!t]
    \centering
    \begin{subfigure}[t]{0.50\linewidth}
        \centering
        \includegraphics[width=.8\linewidth]
        {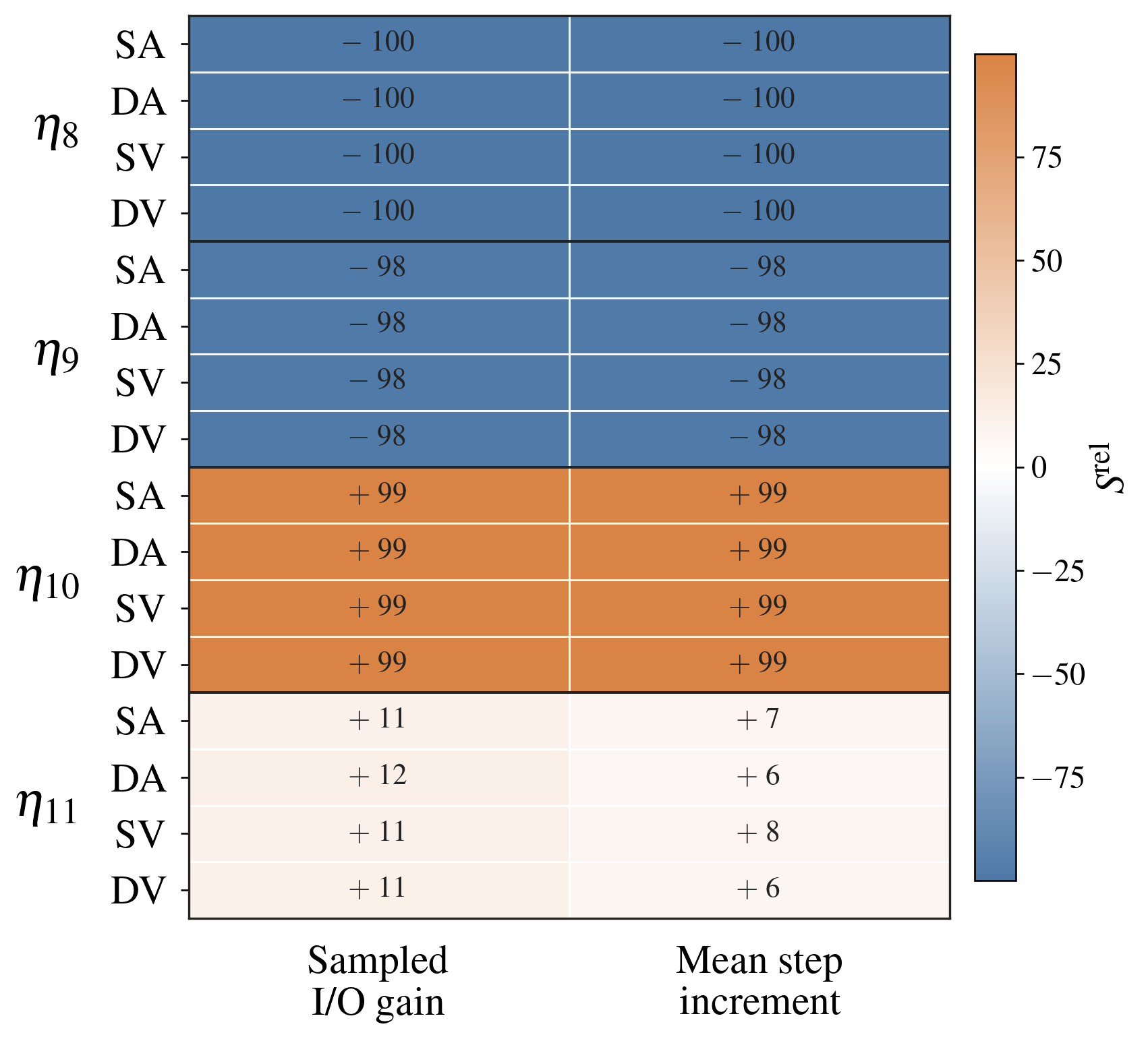}
        \captionsetup{skip=0.7em}
        \caption{Heatmap of intention formation parameter effects across input families.}
        \label{fig:montecarlointention_heatmap}
        \end{subfigure}\hfill
        \begin{subfigure}[t]{0.48\linewidth}
        \centering
        \raisebox{-0.06cm}{%
        \includegraphics[width=.8\linewidth]{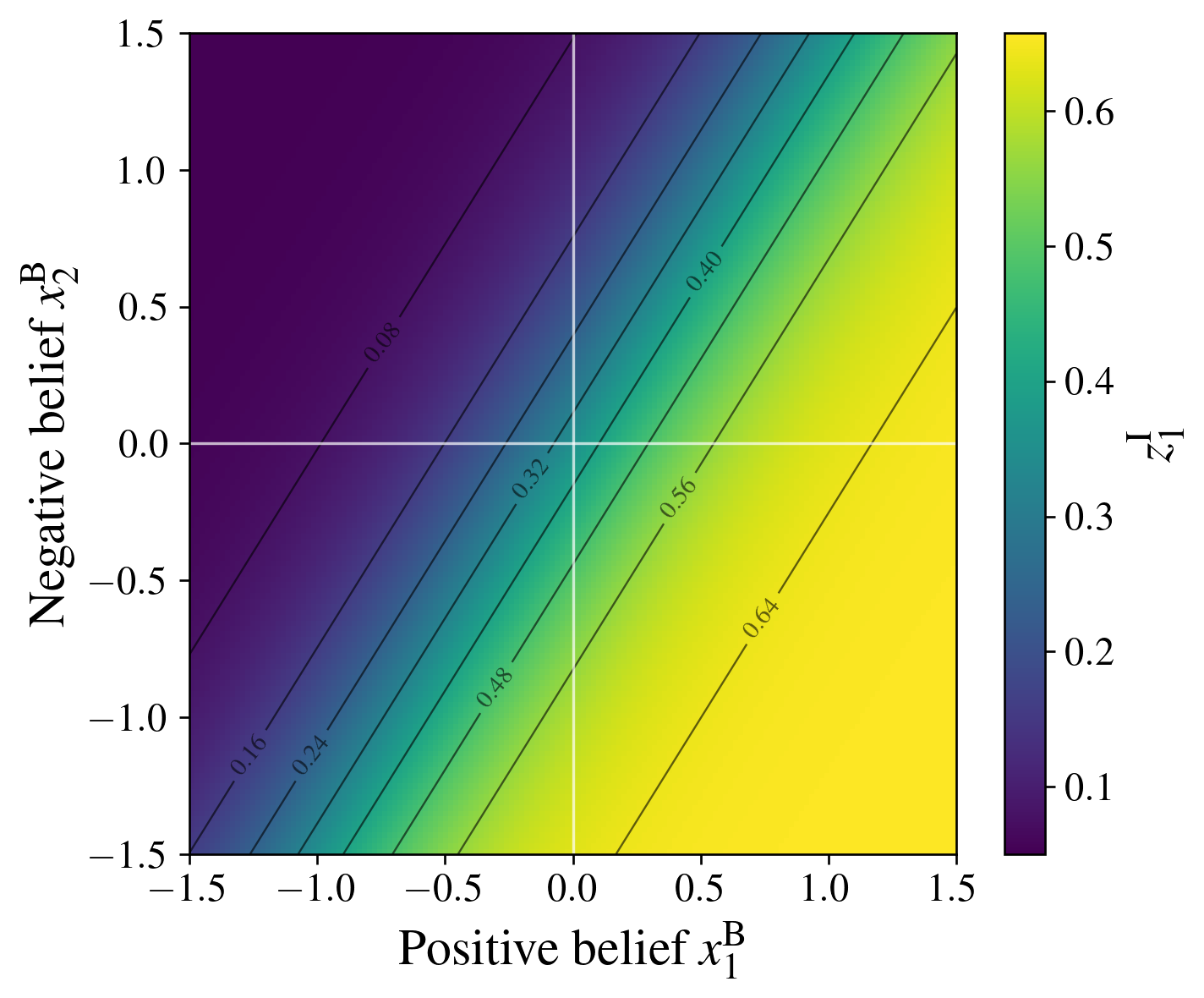}
        }
        \caption{Belief-gate visualization for intention formation.}
        \label{fig:intention_belief_surface}
    \end{subfigure}
        \caption{
        Sensitivity analysis and belief-gate visualization of intention formation.
        (\subref{fig:montecarlointention_heatmap}) Endpoint-change heatmap for sampled \gls{io} gain and mean step increment across input families. Colors encode $\srel$ as defined in \eqref{eqn:srel}.
        (\subref{fig:intention_belief_surface}) Belief-gate visualization with fixed goal input. A positively weighted belief $x_1^{\mathrm{B}} \in \mathcal{B}_1^{+}$ and a negatively weighted belief $x_2^{\mathrm{B}} \in \mathcal{B}_1^{-}$ are varied; color indicates $z_1^{\mathrm{I}}$, and black contour lines indicate level sets.
        }
    \label{fig:montecarlointention_main}
\end{figure*}

\subsubsection{Action selection}
\noindent
We vary the action threshold $\phiasi$ and competition-weight scale $\eta_{12}$ to evaluate Hypothesis~\ref{hypothesis6}. Figure~\ref{fig:montecarloaction_heatmap} shows that switch count remains approximately unchanged, indicating that these parameters mainly regulate activation strength and decisiveness rather than repeatedly changing the winning action.

Increasing $\phiasi$ reduces sampled \gls{io} gain, mean step increment, and winner margin across all input families. This is consistent with thresholding that makes action activation harder to elicit. Increasing $\eta_{12}$ also reduces sampled \gls{io} gain and mean step increment, indicating that stronger competition suppresses action-amplitude sensitivity and temporal reactivity. Its effect on winner margin is less uniform and therefore operating-regime dependent. Overall, these trends support Hypothesis~\ref{hypothesis6}: in the present operating regime, action selection parameters mainly regulate action strength and robustness, while upstream intention trajectories largely determine which action is selected.

\begin{figure*}[!t]
    \centering
    \begin{subfigure}[t]{0.50\linewidth}
        \centering
        \includegraphics[width=.8\linewidth]{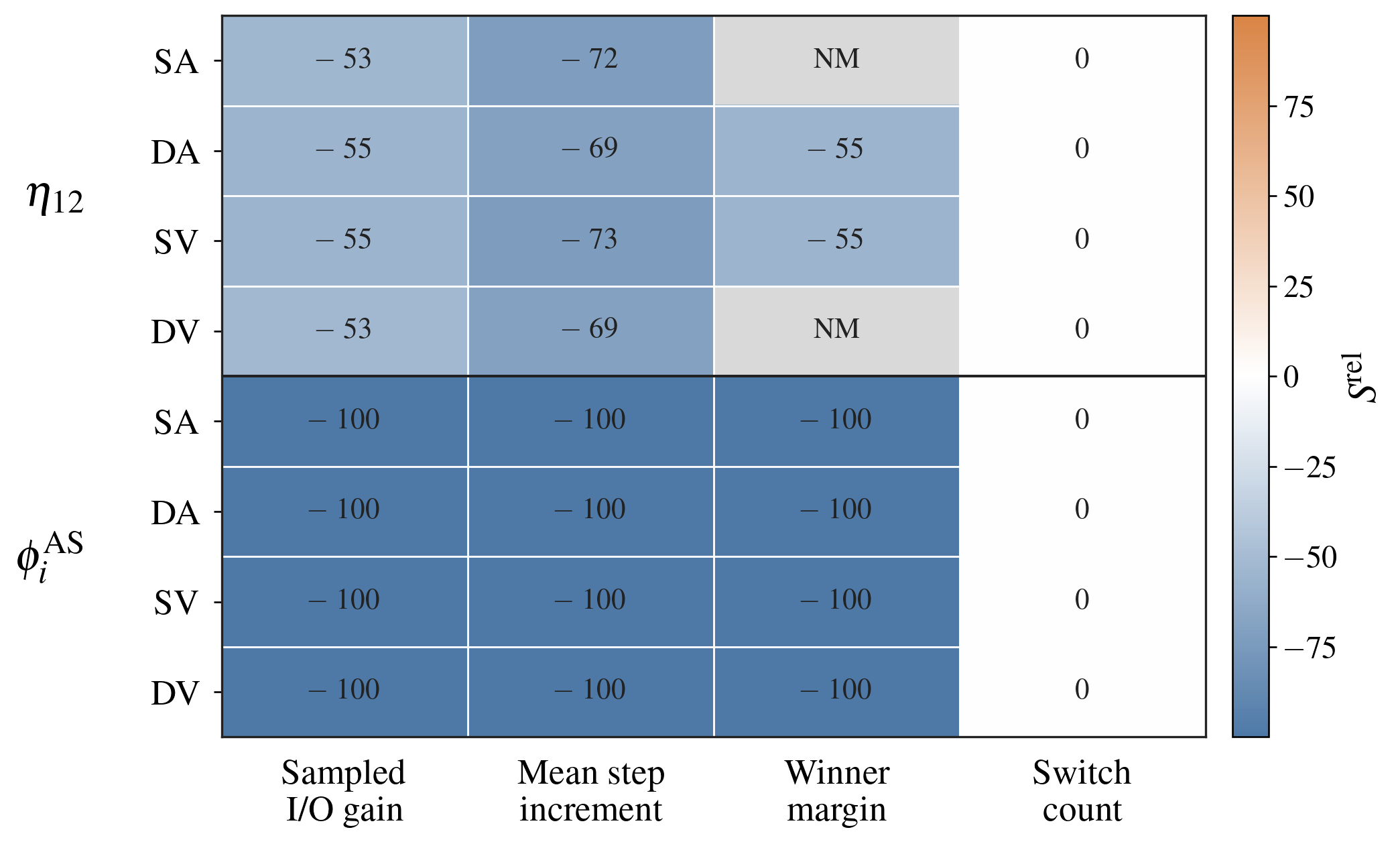}
        \captionsetup{skip=1.3em}
        \caption{Heatmap of action selection parameter effects across input families.}
        \label{fig:montecarloaction_heatmap}
        \end{subfigure}\hfill
        \begin{subfigure}[t]{0.48\linewidth}
        \centering
        \raisebox{-0.26cm}{%
        \includegraphics[width=.8\linewidth]{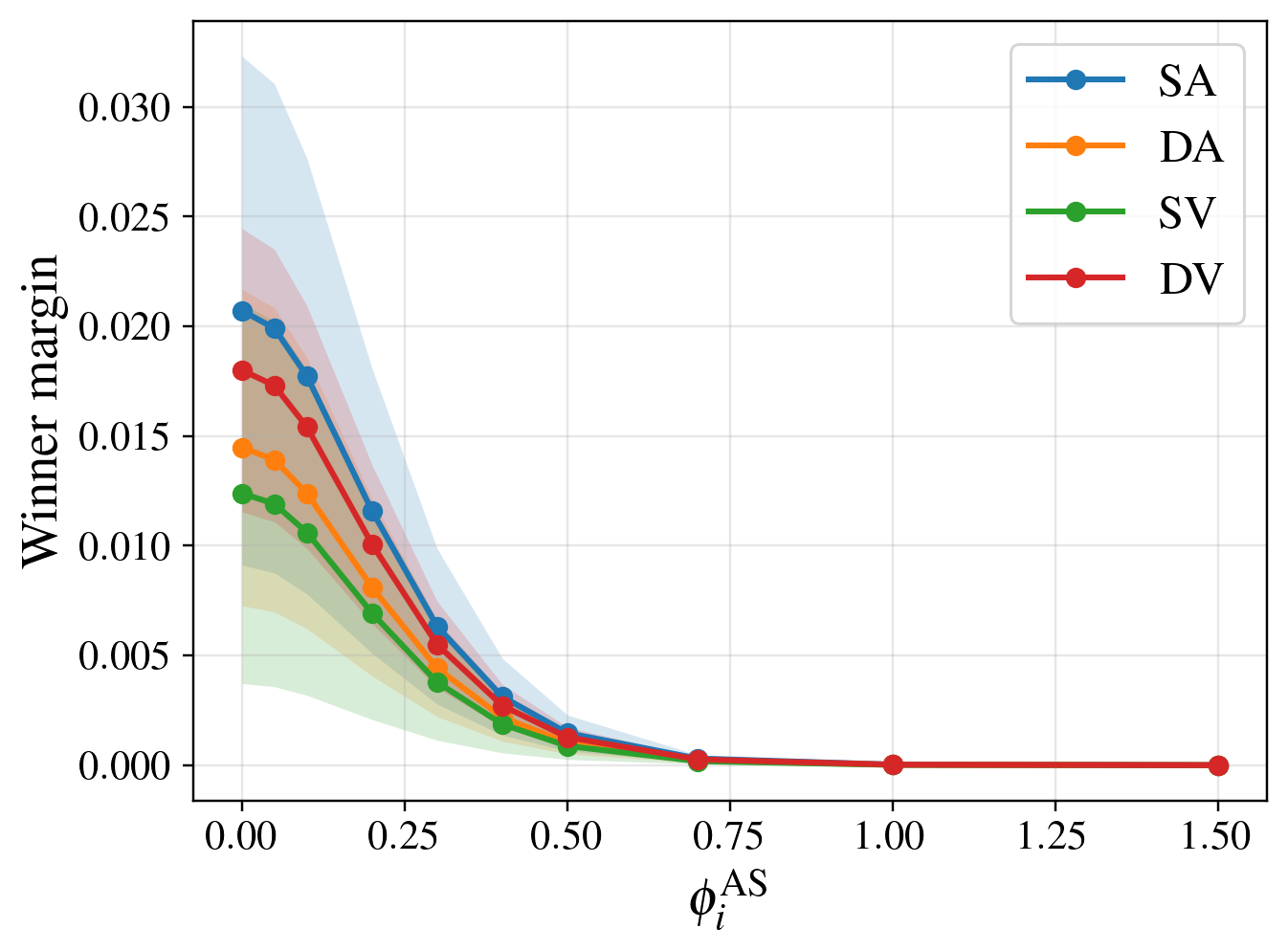}
        }
        \caption{Winner margin under a sweep of $\phi_i^{\mathrm{AS}}$.}
        \label{fig:montecarloaction_phi_margin}
    \end{subfigure}
    \caption{
    Sensitivity analysis of the action selection stage.
    (\subref{fig:montecarloaction_heatmap}) Endpoint-change heatmap for sampled \gls{io} gain, mean step increment, winner margin, and switch count across input families. Colors encode $\srel$ as defined in \eqref{eqn:srel}.
    (\subref{fig:montecarloaction_phi_margin}) Sweep of $\phi_i^{\mathrm{AS}}$, showing that stronger thresholding reduces the winner margin across all input families.
    }
    \label{fig:montecarloaction_main}
\end{figure*}

\subsection{Implications for practical parameter distinguishability}
\label{sec:parameter_identifiability}
\noindent
The sensitivity analyses provide a preliminary indication of practical parameter distinguishability within the chosen synthetic input regimes, but they are not a formal identifiability result. Parameters are easier to distinguish when they produce different signatures across metrics. For example, in predictive inference, $\eta_4$ and $\thetapil$ have opposite effects on tracking error and sampled \gls{is} gain. In cognition, combining trajectory metrics with constant-input scans helps separate recurrent amplification from additive drive: $\eta_5$ changes both state magnitude and local stability, whereas $\eta_7$ mainly changes steady-state magnitude in the baseline regime.

Conversely, parameters that affect the same metrics in the same direction may require richer input excitation or additional observations. Future empirical estimation should therefore combine behavioral and physiological features. A natural extension is to embed the model in an observer or state-estimation framework, such as extended Kalman filtering for unobserved mental states from intermittent measurements \cite{b102}.

\subsection{Closed-loop implementation: rehabilitation showcase}
\label{sec:rehabilitation_case_study}
\noindent
This section illustrates how the perception--cognition--decision model can be embedded in a receding-horizon controller. The case study is a model-class-matched proof of concept, not empirical validation or evidence of clinical effectiveness.

We consider a simulated rehabilitation scenario in which a robotic coach suggests upper-limb movements to a patient. At interaction step $k$, meaning one discrete robot--patient exchange in which the robot proposes a movement and observes the simulated patient response, the robot suggests a movement difficulty $\dsuggest(k) \in \{0,1, \ldots,\overline{d}\}$, where $\dsuggest(k)=0$ represents rest and larger values increasingly demanding rehabilitation movements. The patient may either perform or avoid the suggestions, thus the action $\bza(k)$ has two elements $\zai(k)$ (cf. \eqref{eqn:actionselection}), with $i = \left\{\text{perform, avoid}\right\}$. The binary response is obtained from the action-selection outputs as

\begin{align}
\label{eqn:ypatient}
    \ypatient(k)
    =
    \mathbb{I}
    \left[
    z^{\mathrm{A}}_{\mathrm{perform}}(k)
    \geq
    z^{\mathrm{A}}_{\mathrm{avoid}}(k)
    \right],
\end{align}
where $\ypatient(k)=1$ indicates performance, $\ypatient(k)=0$ indicates rejection, and $\mathbb{I}[\cdot]$ denotes the indicator function, equal to $1$ if its argument is true and $0$ otherwise. The realized performed difficulty is $d^{\mathrm{perform}}(k)=\ypatient(k)\cdot \dsuggest(k)$, with $\dsuggest(k)$ the current suggested difficulty.

The patient receives two normalized inputs: $\dsuggest(k)$ and the previously performed difficulty $d^{\mathrm{perform}}(k-1)$, representing recently realized effort. This separation distinguishes anticipatory threat due to the proposed movement from fatigue caused by previously performed movement. Thus,

\begin{align*}
    \boldsymbol{u}^{\mathrm{patient}}(k)
    = \frac{1}{\overline{d}}
    \begin{bmatrix}
        \dsuggest(k) &
        d^{\mathrm{perform}}(k - 1)
    \end{bmatrix}^\top.
\end{align*}

The simulated patient is generated by propagating one instance of the perception--cognition--decision model introduced in Section~\ref{sec:mathreprhumandecisionmaking}, with the rehabilitation-specific parameterization listed in Table~\ref{tab:rehabilitation_case_study_parameters} in \ref{appendix:rehabilitationparams}. We define the state vector $\bx$ including six elements, i.e., $i \in \left\{\text{perform}, \text{avoid}, \text{comfort}, \text{capability}, \text{threat}, \text{fatigue} \right\}$. Here, perform and avoid are goal states, whereas belief states include perceived comfort, perceived capability to perform the movement, perceived threat of the movement (e.g. pain, risk, difficulty), and perceived physical or mental fatigue. Comfort and capability support the perform intention and suppress the avoid intention, whereas threat and fatigue support the avoid intention and suppress the perform intention.

The robot propagates a perturbed copy of the patient model, represented by perturbed module mappings $\zhatperception(\cdot)$, $\zhatcognition(\cdot)$, and $\zhatdecision(\cdot)$ for perception, cognition, and decision-making, respectively. The perturbation and process-noise settings are reported in Table~\ref{tab:rehabilitation_noise_perturbation_settings} in \ref{appendix:rehabilitationparams}. Throughout the receding-horizon rollout, a quantity written as $q(k^{\prime}\mid k)$ denotes the value of $q$ at prediction step $k^{\prime}\geq k$, evaluated using information available at interaction step $k$. Hatted variables denote robot-predicted internal quantities.

During prediction, the robot uses the same normalized input structure as the simulated patient:
\begin{align}
\label{eqn:urobot}
    \urobot(k' \mid k)
    =
    \frac{1}{\overline d}
    \begin{bmatrix}
        \dsuggest(k' \mid k) &
        \dhatperform(k'-1 \mid k)
    \end{bmatrix}^{\top}.
\end{align}
The rollout is initialized using the realized previous values, $\dhatperform(k-1\mid k)=\dperform(k-1)$ and $\dsuggest(k-1\mid k)=\dsuggest(k-1)$.

In addition to task-completion feedback, the robot receives a scalar perceived comfort feedback signal, interpreted as a noisy observation of the patient comfort state:
\begin{align*}
    \xi^{\mathrm B}_{\mathrm{comfort}}(k)
    \sim
    \mathcal N
    \left(
    x^{\mathrm B}_{\mathrm{comfort}}(k),
    \sigma_{\mathrm{comfort}}^2
    \right).
\end{align*}
Before selecting the next movement, the robot applies a heuristic correction to the comfort component of its internal prediction:
\begin{align*}
\begin{split}
    \xhatb{comfort}(k \mid k)
    &\leftarrow
    (1-c)
    \xhatb{comfort}(k \mid k-1)
    \\
    &+
    c
    \operatorname{clip}
    \big(
    \xi^{\mathrm{B}}_{\mathrm{comfort}}(k),-1,1
    \big),
\end{split}
\end{align*}
where $c\in[0,1]$ is a correction gain that blends the model-predicted comfort state with the noisy comfort feedback. The remaining belief and goal states are not directly observed and are propagated using the internal model of the robot. Figure~\ref{fig:rehabilitation_block_diagram} summarizes the resulting closed-loop information flow.

\begin{figure}[!t]
    \centering
    \resizebox{\columnwidth}{!}{%
    \begin{tikzpicture}[
    font=\small,
    node distance=1.10cm and 1.35cm,
    block/.style={
    draw,
    rounded corners,
    align=center,
    minimum width=3.25cm,
    minimum height=0.85cm
    },
    arrow/.style={->, thick}
    ]
    
    \node[] (target) {$\dtarget(k)$};
    
    \node[block, below=of target] (mpc)
    {Receding-horizon\\MPC controller};
    
    \node[block, right=2.2cm of mpc] (patient)
    {Patient model};
    
    \node[
    block,
    below=2.5cm of patient,
    xshift=-5.50cm,
    minimum width=4.40cm,
    minimum height=2.20cm,
    label={[yshift=-0.5cm]above:Robot prediction model}
    ] (internal) {};
    
    \node[
    block,
    minimum width=3.25cm,
    minimum height=0.85cm
    ] (comfort) at ([xshift=0cm,yshift=-0.25cm]internal.center)
    {$\hat{x}^{\textup{\textsc{b}}}_{\mathrm{comfort}}(k\mid k)$\\correction};
    
    \draw[arrow] (target) -- (mpc);
    
    \draw[arrow] (mpc) -- node[above, align=center]
    {$\dsuggest(k)$}
    (patient);
    
    \draw[arrow]
    ([xshift=-0.20cm]patient.south)
    -- ++(0,-0.75)
    coordinate (dperformbottom)
    -- ++(-2.5,0)
    coordinate (dperformleft)
    -- ++(0,0.88)
    -- ([yshift=-0.30cm]patient.west);
    
    \node[below=0.08cm, inner sep=1pt] at
    ($(dperformbottom)!0.50!(dperformleft)$)
    {$d^{\mathrm{perform}}(k-1)$};
    
    \draw[arrow]
    ([xshift=1.2cm]patient.south)
    -- ++(0,-3.80)
    -- node[above, fill=white, inner sep=1pt]
    {$\xi^{\mathrm B}_{\mathrm{comfort}}(k)$}
    (comfort.east);
    
    \draw[arrow]
    ([xshift=0.5cm]patient.south)
    -- ++(0,-3.1555)
    -- ++(-3.25,0)
    node[midway, above, fill=white, inner sep=1pt, align=center]
    {$\ypatient(k)$\\$d^{\mathrm{perform}}(k)$}
    |- ([yshift=0.45cm]internal.east);
    
    \draw[arrow] (mpc.south) -- 
    node[midway, right, fill=white, inner sep=1pt, align=center]
    {$\urobot(k'\mid k)$}
    (internal.north);
    
    \draw[arrow] (internal.west) -- ++(-1.0,0) |- 
    node[pos=0.26, right, fill=white, inner sep=1pt, align=center]
    {$\hat{\bx}(k'\mid k)$\\
    $\hat{y}^{\mathrm{patient}}(k'\mid k)$\\
    $\hat{d}^{\mathrm{perform}}(k'\mid k)$}
    (mpc.west);
    
    \end{tikzpicture}%
    }
    \caption{
    Block diagram of the closed-loop rehabilitation showcase. 
    }
    \label{fig:rehabilitation_block_diagram}
\end{figure}
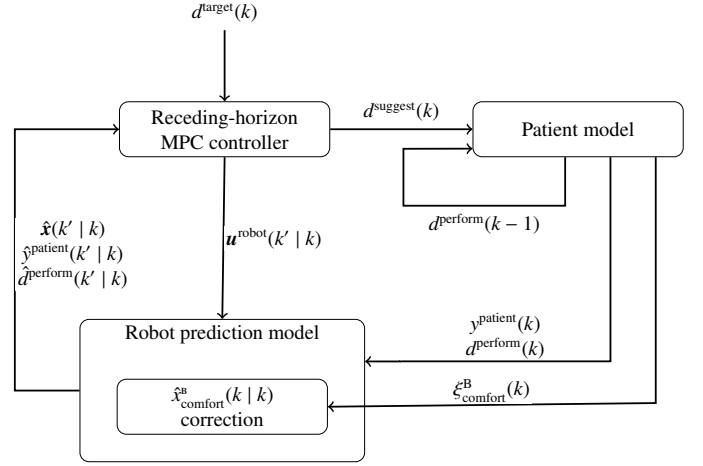

The robot uses receding-horizon \gls{mpc} with horizon $n^{\mathrm P} \in \mathbb{N}$. For each candidate vector 
of suggested difficulty levels 

\begin{align*}
\boldsymbol{d}^{\textup{suggest}}_{k:k+n^{\mathrm{P}}-1\mid k}
    =
    \left[
    \dsuggest(k\mid k),
    \cdots,
    \dsuggest(k+n^{\mathrm P}-1\mid k)
    \right]^\top,
\end{align*}
the robot propagates its internal model and computes a predicted cumulative cost, defined as 

\begin{align*}
    \widehat{\mathcal{J}}_k
    =
    \sum_{j=0}^{n^{\mathrm{P}}-1}
    J
    \Bigg(
    &\hat{\bx}(k+j+1 \mid k),
    \yhatpatient(k+j \mid k), \nonumber\\
    &\dsuggest(k+j\mid k),
    \dsuggest(k+j-1\mid k)
    \Bigg).
\end{align*}
The stage cost is evaluated as

\begin{align}
\label{eqn:costfunction}
    \begin{split}
    J&
    \Biggl(
    \hat{\bx}(k^{\prime}+1 \mid k),
    \yhatpatient(k^{\prime} \mid k), \\
    &\dsuggest(k^{\prime}\mid k),
    \dsuggest(k^{\prime}-1\mid k)
    \Biggr) \\
    &\coloneqq
    \lambda_1
    \left[
    \xhatb{threat}(k^{\prime}+1 \mid k)
    \right]_+
    +
    \lambda_2
    \left[
    \xhatb{fatigue}(k^{\prime}+1 \mid k)
    \right]_+ \\
    &+
    \lambda_3
    \left[
    \xbarbcomfort
    -
    \xhatb{comfort}(k^{\prime}+1 \mid k)
    \right]_+^2 \\
    &+
    \lambda_4
    \left(
    1-\yhatpatient(k^{\prime} \mid k)
    \right) \\
    &+
    \frac{\lambda_5}{\overline{d}^2}
    \left(
    \dhatperform(k^{\prime} \mid k)
    -
    \dtarget(k^{\prime})
    \right)^2 \\
    &+
    \frac{\lambda_6}{\overline{d}^2}
    \left(
    \dsuggest(k^{\prime}\mid k)
    -
    \dtarget(k^{\prime})
    \right)^2 \\
    &+
    \frac{\lambda_7}{\overline{d}^2}
    \left(
    \dsuggest(k^{\prime} \mid k)
    -
    \dsuggest(k^{\prime}-1 \mid k)
    \right)^2,
    \end{split}
\end{align}
where 
\begin{align}
\label{eqn:dhatperform}
    \dhatperform(k^{\prime} \mid k)
    \coloneqq
    \yhatpatient(k^{\prime} \mid k)
    \dsuggest(k^{\prime} \mid k),
\end{align}
and $\yhatpatient(k' \mid k)$ is the predicted binary performance decision obtained by applying \eqref{eqn:ypatient} to the predicted action selection outputs generated during the MPC rollout. The stage cost \eqref{eqn:costfunction} penalizes predicted threat, fatigue, comfort below an acceptable level $\xbarbcomfort$, rejection, deviation from a rehabilitation target $\dtarget(k')$, and abrupt changes in suggested difficulty. The weights satisfy $\lambda_1,\ldots,\lambda_7 \geq 0$. 

The finite-horizon objective is not assumed to be differentiable or convex. In this discrete implementation, candidate difficulty sequences are evaluated by exhaustive enumeration using the rollout equations of the robot prediction model; the full constrained rollout formulation is given in \ref{appendix:closedlooprollout}. Only the first element of the minimizing sequence is applied, $\dsuggest(k)={\dsuggest}^{\star}(k\mid k)$, and the procedure is repeated after new feedback is received. The realized cost is an internal simulation metric, not an externally validated clinical outcome.

As baselines, we use a target-following controller and a random controller. Neither baseline uses feedback or propagates the patient model. The target-following controller applies the prescribed rehabilitation target directly, $\dsuggest(k)=\dtarget(k)$, whereas the random controller samples uniformly from the admissible difficulty set. Thus, the target-following controller provides a structured open-loop baseline, while the random controller provides an unstructured open-loop reference. 

Figure~\ref{fig:rehabilitation_case_study_results} summarizes the closed-loop behavior. In the representative run, the model-based controller keeps the action selection margin $z^{\mathrm{A}}_{\mathrm{perform}}-z^{\mathrm{A}}_{\mathrm{avoid}}$ positive throughout the simulation, indicating that the patient remains in the perform regime. The target-following baseline initially remains in the perform regime, but its margin decreases over time and becomes negative after approximately the middle of the horizon, leading to repeated refusals. The random baseline frequently proposes high-intensity movements without accounting for patient state or rehabilitation target, and quickly drives the patient into the avoid regime. Across repeated simulations, the model-based controller obtains the lowest realized cumulative cost within this simulated model-class-matched proof-of-concept setting.

Figure~\ref{fig:rehabilitation_cognitive_state_dynamics} shows that these behavioral differences are reflected in the latent patient states. Under model-based control, the perform goal remains higher, the avoid goal remains lower, threat is reduced, and capability is ultimately higher than under target following. Comfort decreases gradually because movements continue to be performed, but the controller avoids an abrupt collapse. Fatigue increases under model-based control because the patient continues performing movements. This should not be interpreted as fatigue being desirable. Rather, under the chosen cost weights and target schedule, moderate fatigue is tolerated when it accompanies continued task performance, low threat, sustained comfort, and target tracking. Different clinical priorities could be represented by increasing $\lambda_2$ or modifying the target schedule. In contrast, the later reduction in fatigue under target following reflects non-performance after repeated refusal rather than improved rehabilitation behavior.

Because both comparison controllers are open loop, the results demonstrate an advantage over these particular baselines but do not isolate the contribution of latent-state prediction from the broader benefit of feedback-based adaptation.

\begin{figure*}[!t]
    \centering
    \begin{subfigure}[t]{0.49\linewidth}
    \centering
    \includegraphics[width=\linewidth]{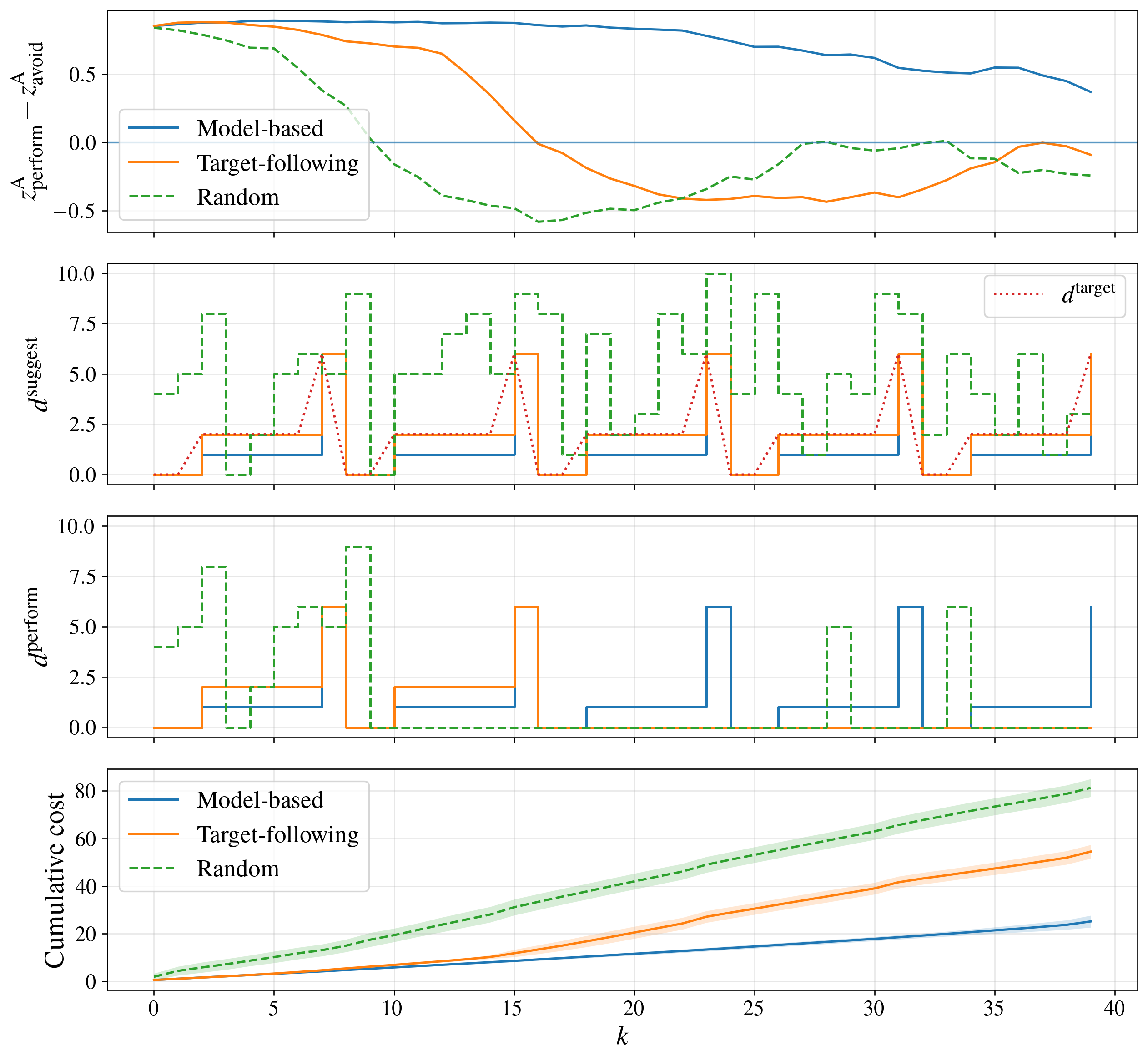}
    \caption{
    Representative closed-loop behavior and realized cumulative cost across repeated simulations.
    }
    \label{fig:rehabilitation_case_study_results}
    \end{subfigure}\hfill
    \begin{subfigure}[t]{0.49\linewidth}
    \centering
    \raisebox{0.05cm}{
    \includegraphics[width=\linewidth]{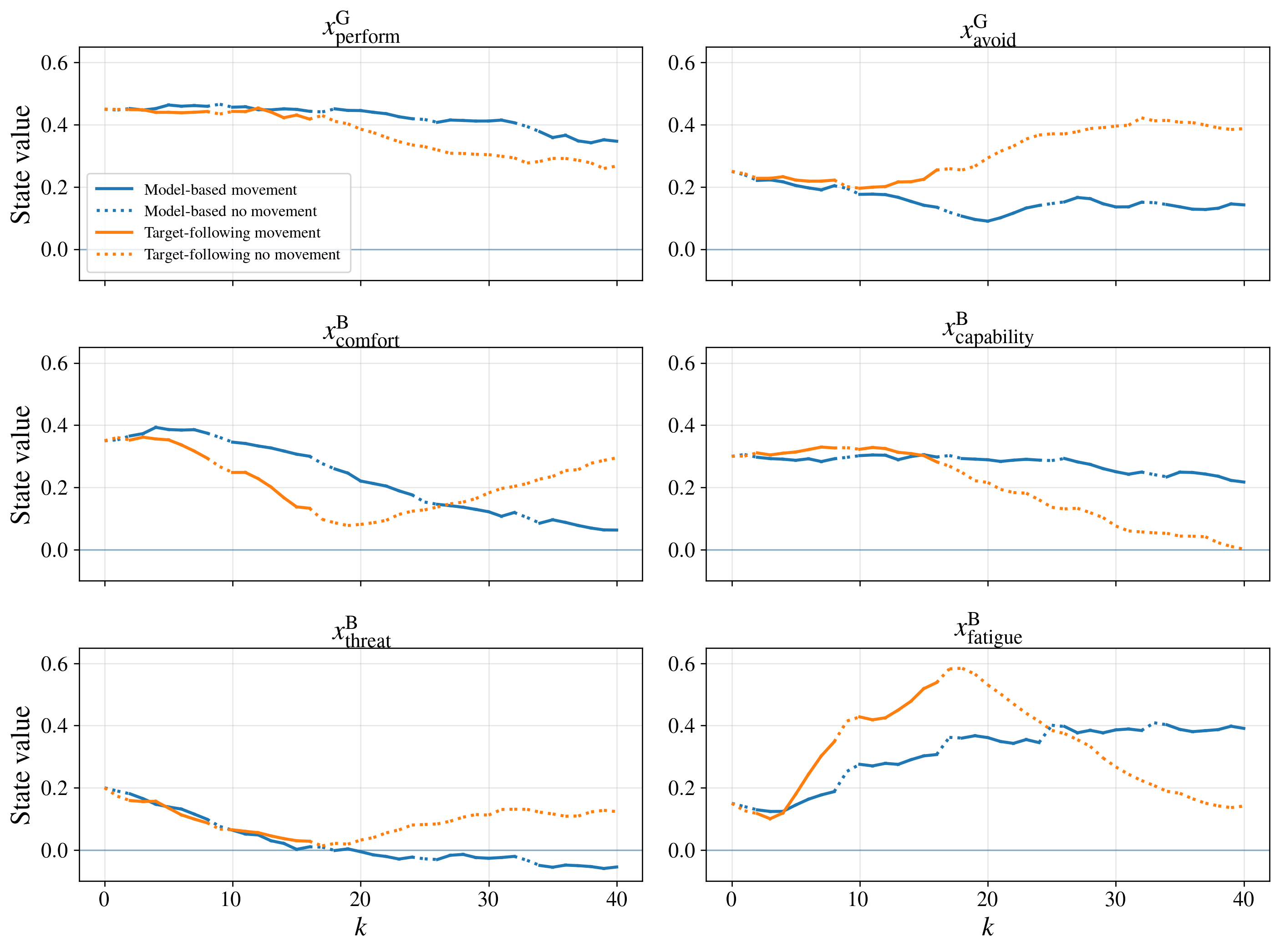}
    }
    \caption{
    Latent cognitive state trajectories under model-based control and target following.
    }
    \label{fig:rehabilitation_cognitive_state_dynamics}
    \end{subfigure}
    \caption{
    Closed-loop rehabilitation showcase.
    (\subref{fig:rehabilitation_case_study_results}) Closed-loop rehabilitation case-study results. The top three panels show one representative closed-loop run: the action selection margin $z^{\mathrm{A}}_{\mathrm{perform}}-z^{\mathrm{A}}_{\mathrm{avoid}}$, the suggested difficulty, and the performed difficulty. Positive margin values indicate that the patient performs the suggested movement. The final panel reports the mean realized cumulative cost over 50 repeated simulations, with shaded bands indicating one standard deviation.
    (\subref{fig:rehabilitation_cognitive_state_dynamics}) Latent cognitive state trajectories under the model-based controller and the target-following baseline. The target-following baseline follows the prescribed rehabilitation progression without patient-state adaptation. Dotted trajectory segments indicate time intervals in which no movement is performed, including both suggested rest and refused nonzero suggestions.
    }
    \label{fig:rehabilitation_showcase}
\end{figure*}

\section{Conclusion}
\label{sec:conclusion}
\noindent
This paper introduced a modular state-space framework linking sensory input to behavioral output through explicit perception, cognition, and decision-making dynamics. The framework combines interpretable latent-state structure with sufficient conditions for boundedness, regularity, forward invariance, perceptual contraction under constant input, and regional \gls{iss} of the cognition dynamics.

The numerical analyses illustrated how model parameters regulate perceptual tracking, cognitive amplification and damping, intention expression, and action decisiveness. A model-class-matched rehabilitation showcase further demonstrated how the framework can be embedded in receding-horizon control using partial feedback. Together, these results illustrate how structured cognitive models can support both dynamical analysis and model-based feedback control.

Several limitations remain. First, the results are simulation-based and do not establish empirical validity with human participants. In the rehabilitation showcase, the simulated patient is generated by the same model class used by the controller, so the experiment does not test robustness to structural mismatch with real human behavior. Second, the model parameters, rehabilitation target, and objective weights are manually specified rather than estimated from behavioral or clinical data. Third, the rehabilitation example uses a simplified binary response and scalar comfort-feedback signal, whereas real rehabilitation involves graded performance, movement quality, pain, recovery, and multimodal feedback. Fourth, exhaustive enumeration is transparent for the small discrete action space considered here, but does not scale to larger horizons, richer action spaces, or continuous robot actions. Finally, the stability analysis does not yet provide recursive feasibility, robust constraint satisfaction, or closed-loop stability guarantees for the full control loop.

These limitations suggest several directions for future work. A first priority is empirical validation using open-source behavioral or \gls{hri} datasets, followed by targeted human-participant studies. A second direction is observer-based estimation, building on the practical distinguishability analysis in Section~\ref{sec:parameter_identifiability}, to infer latent states and constrain uncertain parameters from partial observations. A third direction is robust or stochastic cognitive \gls{mpc}, with explicit treatment of model mismatch, measurement noise, and human response uncertainty. More scalable optimization methods, such as branch-and-bound, mixed-integer \gls{mpc}, or approximate \gls{mpc}, might also be needed for larger state–action spaces and longer horizons. Finally, integrated \gls{mpc}–learning architectures could use residual models or adaptive cost shaping to personalize control online while retaining the interpretability of the model-based framework.

Overall, the proposed framework provides a mathematically explicit foundation for modeling human decision dynamics in adaptive human-centered systems. Its main contribution is to offer a white-box dynamical structure through which latent cognitive states, stability properties, and predictive control objectives can be jointly analyzed. This makes the framework a step toward human–autonomy interaction that adapts to users while preserving interpretability, analyzability, and principled closed-loop behavior.

\section*{Data availability}
\noindent
The data used in this study are synthetically generated within the numerical simulations. The simulation code is publicly available in the GitHub repository \url{https://github.com/svenschoonebeek/State-space-model-of-HDM}.

\bibliographystyle{elsarticle-num-names} 
\bibliography{cas-refs}

\appendix
\section{Proof of analytical results}
\phantomsection
\addcontentsline{toc}{section}{Proof of analytical results}
\label{appendix}
\noindent
This appendix contains the proofs of the lemmas and propositions stated in Section~\ref{sec:stabilityanalysis}.
\medskip

\setcounter{equation}{0}
\renewcommand{\theequation}{A\arabic{equation}}

\setcounter{table}{0}
\renewcommand{\thetable}{A\arabic{table}}

\begin{proof}[Proof of Lemma~\ref{lemma:attentionalselectionsinglestep}]
Under Assumption~\ref{ass:attentionalselection}, $\fasl(\cdot)$ and $\gasl(\cdot)$ are continuously differentiable on open neighborhoods of their admissible domains, and $\gasl(\bu)\geq \gaslmin>0$ holds on $\uset$. Hence, by the quotient rule, the mapping $\bu \mapsto \zasl(\bu)$ defined by \eqref{eqn:perceptualaccessgeneral} is continuously differentiable on an open neighborhood of $\uset$. 
Accordingly, for any $\bu_1,\bu_2\in\uset$, the convexity of $\uset$ and the mean value theorem imply
\begin{align*}
\begin{split}
    \abs{\zasl(\bu_1)-\zasl(\bu_2)}
    &\leq
    \sup_{\xi\in[\bu_1,\bu_2]}
    \norm{\nabla_{\bu}\zasl(\xi)}
    \norm{\bu_1-\bu_2} 
\end{split}
\end{align*}
Since $\uset$ is compact and $\nabla_{\bu}\zasl(\bu)$ is continuous on $\uset$, $\nabla_{\bu}\zasl(\bu)$ attains a finite supremum on $\uset$. 
Therefore, Lemma~\ref{lemma:attentionalselectionsinglestep} holds for $\latsl \coloneqq \sup_{\bu\in\uset}\norm{\nabla_{\bu}\zasl(\bu)}$. 
\end{proof}
\medskip

\begin{proof}[Proof of Proposition~\ref{prop:attentionalselectiontrajectory}]
Since $\bu(i)\in\uset$ for all $i$ along any admissible input trajectory, 
Lemma~\ref{lemma:attentionalselectionsinglestep} applies to each consecutive pair $\bu(i)$ and $\bu(i + 1)$. 
Summing the resulting inequalities over $i = k_1, \ldots, k_2 - 1$ yields the claim. 
Since the sum is taken over finitely many admissible inputs, 
the right-hand side of \eqref{eqn:propattnselection1} is finite. 
\end{proof}
\medskip

\begin{proof}[Proof of Lemma~\ref{lem:npe_forward_invariance}]
Both $\xnpel(k-1)$ and $\zasl(\bu(k))$ belong to $\left[0,\zaslmax\right]$. Since $\xnpel(k)$ is a convex combination of these two quantities (cf.~\eqref{eqn:nperecursionconvex}), it also belongs to $\left[0,\zaslmax\right]$. Hence,
$\xnpel(k)\in\xnpelset$.
By induction, if $\xnpel(0)\in\xnpelset$, then $\xnpel(k)\in\xnpelset$ for all $k\geq0$.
\end{proof}
\medskip

\begin{proof}[Proof of Lemma~\ref{lem:npe_contraction_constant}]
For the constant input $\bar{\bu}$, the attentional selection weight $\zasl(\bar{\bu})$ is constant for all $k$. Hence, using \eqref{eqn:nperecursionconvex} and \eqref{eqn:npetrackingerror} we have
\begin{align*}
    \enpe_\ell(k)
    =
    \left(1-\alphapil(k)\right)\enpe_\ell(k-1).
\end{align*}
Iterating this recursion and using $\alphapil(i)\geq\underline{\alphapil}>0$ yield
\begin{align*}
    \left|\enpe_\ell(k)\right|
    \le
    \prod_{i=0}^{k-1}
    \left(1-\alphapil(i)\right)
    \left|\enpe_\ell(0)\right| \le 
    \left(1-\underline{\alphapil}\right)^k
    \left|\enpe_\ell(0)\right|.
\end{align*}
Because $0<\underline{\alphapil} < 1$, it follows that
$\left(1-\underline{\alphapil}\right)^k\to0$ as $k\to\infty$. Hence $\enpe_\ell(k)$ converges exponentially to zero.
\end{proof}

\begin{proof}[Proof of Lemma~\ref{lem:npe_update_lipschitz}]
By definition of $T_\ell$ in \eqref{eq:def_Tlk_npe},
\begin{align*}
    T_{\ell}&(\bu_1,\xi_1)-T_{\ell}\left(\bu_2,\xi_2\right)
    = \\
    &\left(1-\alphapil(k)\right)\left(\xi_1-\xi_2\right) 
    +
    \alphapil(k)\bigl(\zasl \left(\bu_1\right)-\zasl\left(\bu_2\right)\bigr).
\end{align*}
Hence, by triangle inequality and Lemma~\ref{lemma:attentionalselectionsinglestep},

\begin{align*}
    &\abs{T_{\ell} \left(\bu_1,\xi_1 \right)-T_{\ell} \left(\bu_2,\xi_2\right)}
    \le
    \\
    &\alphapil(k)\latsl\norm{\bu_1-\bu_2} +
    \left(1-\alphapil(k) \right)\abs{\xi_1-\xi_2}.
\end{align*}
Therefore, $T_\ell(\cdot,\cdot)$ is Lipschitz on $\uset\times\xnpelset$ with respect to the product norm.
\end{proof}
\medskip

\begin{proof}[Proof of Lemma~\ref{lem:uniformselfinhibition}]
For each $i\in\mathcal I$, the self-dynamics are given by
$
    \fiic\left(\bxnpe(k)\right)
    =
    -\kappa
    \exp\left(\gamma_i+ \left\langle\Lambda_i,\bxnpe(k) \right\rangle\right).
$
Since 
$\xnpeset$ is compact, we have $\norm{\bxnpe(k)} \leq \xnpebar$. 
By H\"{o}lder's inequality,
\begin{align*}
    \abs{ \left \langle\Lambda_i,\bxnpe(k) \right \rangle}
    \leq
    \norm{\Lambda_i}\norm{\bxnpe(k)}
    \leq
    \norm{\Lambda_i}\xnpebar.
\end{align*}
Hence,
$
    \gamma_i-\norm{\Lambda_i}\xnpebar
    \leq
    \gamma_i+ \left\langle\Lambda_i,\bxnpe(k) \right\rangle
    \leq
    \gamma_i+\norm{\Lambda_i}\xnpebar.
$
Exponentiating and multiplying by $-\kappa<0$ gives
\begin{align*}
    -\kappa
    \exp\left(
        \gamma_i+\norm{\Lambda_i}\xnpebar
    \right)
    \leq
    \fiic \left(\bxnpe(k) \right)
    \leq
    -\kappa
    \exp\left(
        \gamma_i-\norm{\Lambda_i}\xnpebar
    \right).
\end{align*}
Using the definitions of $\underline{\Lambda}$ and $\overline{\Lambda}$, this implies
\begin{align*}
    -\overline{\Lambda}
    \leq
    \fiic \left(\bxnpe(k) \right)
    \leq
    -\underline{\Lambda},
\end{align*}
for all $k$ and all $i\in\mathcal I$.
\end{proof}
\medskip

\begin{proof}[Proof of Proposition~\ref{prop:cognition_forward_invariance}]
Write the cognition update \eqref{eqn:dcmnonlinear} in vector form as
\begin{align}
\begin{split}
\label{eqn:cognitionmoduleupdatevectorform}
    \bx(k+1)
    &=
    \left(
    I+\Delta t\,\fc \left(\bxnpe(k)\right)
    +\Delta t\,G^{\textup{\textsc{c}}} \left(\bx(k),\bxnpe(k) \right)
    \right)\bx(k) \\
    &+
    \Delta t\,\Theta \; \bxnpe(k),
\end{split}
\end{align}
where $\fc \left(\bxnpe(k) \right)\coloneqq\mathrm{diag} \left(\fiic \left(\bxnpe(k) \right) \right)$. 
By Lemma~\ref{lem:uniformselfinhibition}, for every $i\in\mathcal I$ we have
$-\overline{\Lambda}
\le
\fiic \left(\bxnpe(k) \right)
\le
-\underline{\Lambda}$. 

\noindent
Since $\Delta t\leq 1/\overline{\Lambda}$, it follows that each diagonal entry of $I+\Delta t\,\fc\left(\bxnpe(k)\right)$ lies in 
$\left[0,\;1-\Delta t\,\underline{\Lambda}\right]$. 
Hence, for the induced matrix norm used throughout we have 
\[
\norm{I+\Delta t\,\fc \left(\bxnpe(k) \right)}
\le
1-\Delta t\,\underline{\Lambda}.
\]

\noindent
From \eqref{eqn:gimc}, we have $
    G^{\textup{\textsc{c}}} \left(\bx(k),\bxnpe(k)\right)
    =
    \Phi
    +
    \sum_{\ell\in\mathcal L}\xnpel(k)\Psi_\ell
    +
    \sum_{q\in\mathcal I}x_q(k)\Xi_q$. 
Using sub-additivity of induced norms and the definitions of $\overline{\Phi}$, $\overline{\Psi}$, and $\overline{\Xi}$, we obtain 
\[
\norm{G^{\textup{\textsc{c}}}\left(\bx(k),\bxnpe(k)\right)}
\le
\overline{\Phi}
+
\overline{\Psi}
+
\sum_{q\in\mathcal{I}} \abs{x_q(k)}  \norm{\Xi_q}.
\]
Moreover, $\sum_{q\in\mathcal I} \abs{x_q(k)} \norm{\Xi_q} 
\leq \overline{\Xi} \norm{\bx(k)}$.

\noindent
Therefore, $
    \norm{G^{\textup{\textsc{c}}} \left(\bx(k),\bxnpe(k) \right)}
    \leq
    \overline{\Phi}
    +
    \overline{\Psi}
    +
    \overline{\Xi}
    \norm{\bx(k)} 
$.
\smallskip

\noindent
Taking norms in \eqref{eqn:cognitionmoduleupdatevectorform} and using triangle inequality and sub-multiplicativity, we obtain an upper bound on the next state norm in terms of the current state norm: 
\begin{align}
\begin{split}
\label{eqn:cognitionmoduleforwardinvarianceproof}
    &\norm{\bx(k+1)}
    \leq
    \\
    &\left(
    1-\Delta t\,\underline{\Lambda}
    +
    \Delta t
    \left[
    \overline{\Phi}
    +
    \overline{\Psi}
    +
    \overline{\Xi}
    \norm{\bx(k)}
    \right]
    \right)
    \norm{\bx(k)} +
    \Delta t\,\overline{\Theta} \xnpebar.
\end{split}
\end{align}
This boils down to 
\begin{align*}
    &\norm{\bx(k+1)}
    \leq 
    \\
    &
    \norm{\bx(k)}+\Delta t
    \left(
    \overline{\Xi} \norm{\bx(k)}^2
    +
    \left(
    \overline{\Phi}+\overline{\Psi}-\underline{\Lambda}
    \right) \norm{\bx(k)}
    +
    \overline{\Theta} \xnpebar
    \right).
\end{align*}
By assumption, \eqref{eqn:cognitionmodulelemmaquadraticinequality} holds, which implies that $R$ is an upper bound for $\norm{\bx(k+1)}$ based on the above inequality.  
Thus, whenever $ \norm{\bx(k)} \le R$, we also have $\norm{\bx(k+1)}\le R$. Therefore the set
$\mathcal X_R=\{\bx:\|\bx\|\le R\}$ is forward invariant.
\end{proof}
\medskip

\begin{proof}[Proof of Corollary~\ref{cor:cognition_forward_invariance_explicit}]
Let
\begin{align*}
    q(r)\coloneqq \overline{\Xi}r^2+\sigma r+\tau.
\end{align*}
Since $r+\Delta tq(r)$ is convex on $[0,R]$, its maximum is attained at an endpoint. Hence, \eqref{eqn:cognitionmodulelemmaquadraticinequality} is equivalent to
\[
\Delta t\tau\le R,
\qquad
q(R)\le0.
\]
If $\overline{\Xi}=0$, then $q(R)=\sigma R+\tau\le0$, which for $\sigma<0$ is equivalent to $R\ge\tau/-\sigma$. If $\overline{\Xi}>0$, then $q(R)\le0$ holds precisely for $R\in[R_-,R_+]$, provided $\sigma<0$ and $\sigma^2-4\overline{\Xi}\tau\ge0$. Combining this interval with $\Delta t\tau\le R$ gives the stated conditions.
\end{proof}

\begin{proof}[Proof of Proposition~\ref{prop:cognition_iss_bound}]
Since $\mathcal X_R$ is forward invariant and $\bx(0)\in\mathcal X_R$, we have
$\norm{\bx(k)}\le R$ for all $k\ge0$. Using the estimate derived in
\eqref{eqn:cognitionmoduleforwardinvarianceproof}, and bounding the quadratic state-dependent term by
$\norm{\bx(k)}^2\le R\norm{\bx(k)}$, gives
\begin{align}
\label{eqn:iss_recursion}
    \norm{\bx(k+1)}
    \leq
    \alpha_R \norm{\bx(k)}
    +
    \beta \norm{\bxnpe(k)},
\end{align}
where
\begin{align*}
    \alpha_R
    =
    1+\Delta t
    \left(
    \overline{\Xi} R
    +
    \overline{\Phi}
    +
    \overline{\Psi}
    -
    \underline{\Lambda}
    \right),
    \qquad
    \beta
    =
    \Delta t\,\overline{\Theta}.
\end{align*}
Iterating recursion \eqref{eqn:iss_recursion} yields
\begin{align*}
    \norm{\bx(k)}
    &\leq
    \alpha_R^k\norm{\bx(0)}
    +
    \sum_{i=0}^{k-1}
    \alpha_R^{k-1-i}
    \beta\norm{\bxnpe(i)} \notag 
\end{align*}
Since $0\leq\alpha_R<1$, the geometric sum is bounded by $(1-\alpha_R)^{-1}$, and therefore,
\begin{align*}
    \norm{\bx(k)}
    \leq
    \alpha_R^k\norm{\bx(0)}
    +
    \frac{\beta}{1-\alpha_R}
    \xnpebartraj,
\end{align*}
which proves \eqref{eqn:iss_bound_optionA}.
\end{proof}
\medskip

\begin{proof}[Proof of Corollary~\ref{cor:cognition_iss}]
The claim follows directly from \eqref{eqn:iss_bound_optionA}. Since $0\leq \alpha_R<1$ under Proposition~\ref{prop:cognition_iss_bound}, the function $\beta_{\mathrm{ISS}}(s,k)=\alpha_R ^k s$ is of class $\mathcal{KL}$ and the function $\gammaiss(s)=\dfrac{\beta}{1-\alpha_R}s$ is of class $\mathcal{K}$. 
Hence the cognition dynamics are \gls{iss} on $\mathcal{X}_R$. 
\end{proof}

\begin{proof}[Proof of Lemma~\ref{lem:intention_bibo_lip}]
Since $\mathcal{X}_R$ is compact, its projected goal and belief domains are compact as well. 
By Assumption~\ref{ass:intentionformation}, the functions $\fisi(\cdot)$ and $\gisi(\cdot)$ are locally Lipschitz, hence Lipschitz on these compact sets. 
Because the intention mapping $\zii(\cdot, \cdot)$ is formed by combining these bounded Lipschitz components 
through addition and multiplication (cf.~\eqref{eqn:intentionselection}), it is Lipschitz continuous. 
\end{proof}

\begin{proof}[Proof of Lemma~\ref{lem:action_bibo_lip}]
The admissible intention domain is compact, because it is inherited from the forward-invariant cognition set. 
By Assumption~\ref{ass:actionselection}, the component mappings entering the action selection rule are locally Lipschitz, 
hence Lipschitz on this compact domain. 
Since the action selection mapping is obtained by composing these bounded Lipschitz components with a smooth bounded gate, 
it is itself Lipschitz continuous on the admissible domain.
\end{proof}

\section{Simulation parameters}
\phantomsection
\addcontentsline{toc}{section}{Simulation parameters}

\setcounter{equation}{0}
\renewcommand{\theequation}{B\arabic{equation}}

\setcounter{table}{0}
\renewcommand{\thetable}{B\arabic{table}}

This appendix reports the diagnostic metrics, simulation parameters, and implementation details used in the numerical sensitivity analyses of Section~\ref{sec:parametersensitivityanalyses} and the closed-loop rehabilitation showcase of Section~\ref{sec:rehabilitation_case_study}.

\subsection{Closed-loop rollout formulation}
\label{appendix:closedlooprollout}

The optimal candidate sequence is obtained by exhaustive enumeration over the admissible difficulty sequences, 
subject to the model rollout constraints used to evaluate $\widehat{\mathcal J}_k$. The prediction constraints use the shifted form of the cognition update, where $\bxhat(k'\mid k)$ is obtained from $\bxhat(k'-1\mid k)$ and $\bxnpehat(k'\mid k)$:

\begin{align*}
    &{\boldsymbol{d}_{k:k+n^{\mathrm P}-1\mid k}^{\textup{suggest}}}^{\star}
    =
    \operatorname*{arg\,min}_{
    \boldsymbol{d}_{k:k+n^{\mathrm P}-1\mid k}^{\textup{suggest}}
    \in \{0,\ldots,\overline d\}^{n^{\mathrm P}}
    }
    \widehat{\mathcal J}_k, \nonumber\\
    & \text{subject to:} \nonumber \\
    & \bxnpehat(k^{\prime} \mid k) =
    \zhatperception\left(
    \urobot(k^{\prime} \mid k),
    \bxnpehat(k^{\prime}-1 \mid k)
    \right), \nonumber \\
    & \bxhat(k^{\prime} \mid k) =
    \zhatcognition \left(
    \bxnpehat(k^{\prime} \mid k),
    \bxhat(k^{\prime}-1 \mid k)
    \right), \nonumber \\
    & \bzahat(k^{\prime} \mid k) =
    \zhatdecision \left(
    \bxhat(k^{\prime} \mid k)
    \right), \nonumber \\
    & \eqref{eqn:urobot},\ \eqref{eqn:dhatperform}. \nonumber 
\end{align*}

The outcome over $N$ interactions is reported as the realized cumulative cost

\begin{align*}
    \mathcal J^{\mathrm{realized}}
    =
    \sum_{k=0}^{N-1}
    J^{\mathrm{realized}}(k),
\end{align*}
where $J^{\mathrm{realized}}(k)$ is obtained from the same stage cost in \eqref{eqn:costfunction} by replacing predicted quantities with realized quantities:

\begin{align*}
    J^{\mathrm{realized}}(k)
    \coloneqq
    J
    \left(
    \bx(k+1),
    \ypatient(k),
    \dsuggest(k),
    \dsuggest(k-1)
    \right).
\end{align*}

\subsection{Closed-loop rehabilitation parameters}
\label{appendix:rehabilitationparams}

\begin{table}[!htbp]
    \centering
    \caption{Noise and perturbation settings used to generate the simulated patient trajectory, partial comfort feedback, and robot-model mismatch in the closed-loop rehabilitation case study.}
    \label{tab:rehabilitation_noise_perturbation_settings}
    \resizebox{\linewidth}{!}{%
    \begin{tabular}{lll}
        \toprule
        Quantity & Equation & Value \\
        \midrule
        Robot-model parameter perturbation
        &
        $\displaystyle
        \delta^{\mathrm{robot}}
        \sim
        \mathcal N
        \left(
        \delta^{\mathrm{true}},
        \sigma_{\mathrm{par}}^2
        \left|\delta^{\mathrm{true}}\right|^2
        \right)$
        &
        $\sigma_{\mathrm{par}}=0.10$ \\
        \addlinespace[0.35em]

        True patient process noise
        &
        $\displaystyle
        \bx(k+1)
        \leftarrow
        \bx(k+1)+\boldsymbol{\omega}(k),
        \quad
        \boldsymbol{\omega}(k)
        \sim
        \mathcal N(0,\sigma_{\omega}^{2}I)$
        &
        $\sigma_{\omega}=0.008$ \\
        \addlinespace[0.35em]

        Comfort feedback
        &
        $\displaystyle
        \xi^{\mathrm B}_{\mathrm{comfort}}(k)
        \sim
        \mathcal N
        \left(
        x^{\mathrm B}_{\mathrm{comfort}}(k),
        \sigma_{\mathrm{comfort}}^2
        \right)$
        &
        $\sigma_{\mathrm{comfort}}=0.03$ \\
        \bottomrule
    \end{tabular}%
    }
\end{table}

\begin{table}[!htbp]
    \centering
    \caption{Parameter values used in the closed-loop rehabilitation case study. The state ordering used for all cognition matrices is $(\mathrm{perform},\mathrm{avoid},\mathrm{comfort},\mathrm{capability},\mathrm{threat},\mathrm{fatigue})$. For $\Psi_\ell$, an entry listed as $(i,j)$ denotes the coefficient multiplying state component $x_j$ in the update of state component $x_i$. For $\Xi$, an entry listed as $\Xi_q(i,j)$ denotes the coefficient in the slice associated with state component $x_q$; this coefficient multiplies $x_j$ in the contribution to the update of $x_i$ (cf. \eqref{eqn:gimc}). All unlisted entries of $\Psi_\ell$ and $\Xi$ are zero.}
    \label{tab:rehabilitation_case_study_parameters}
    \resizebox{\linewidth}{!}{%
    \begin{tabular}{lll}
        \toprule
        Module/layer & Parameter & Value \\
        \midrule
        \multirow{6}{*}{Setup}
        & $N$, $n^{\mathrm{P}}$ & $40,\ 3$ \\
        & seed & $42$ \\
        & $\overline{d}$ & $10.0$ \\
        & $\dtarget(k)$ & $(0,0,2,2,2,2,2,6)$, repeated every $8$ steps \\
        & MC runs & $50$ \\
        & $c$ & $0.40$ \\
        \midrule
        \multirow{6}{*}{Attentional selection}
        & $\card{\mathcal L}$ & $2$ \\
        & $\faslmax$ & $0.75$ for all $\ell\in\mathcal L$ \\
        & $\eta_1,\eta_2$ & $1.80,\ 0.25$ for all $\ell\in\mathcal L$ \\
        & $\betaasl$ & $0.65$ \\
        & $p_m$ & $2$ for all $m\in\mathcal L$ \\
        & $\gamma_{\ell m}^{\mathrm{ATS}}$ & $0.70$ for all $\ell,m\in\mathcal L$ \\
        \midrule
        \multirow{1}{*}{Predictive inference}
        & $\phipil,\thetapil,\chipil$ & $1.0,\ 1.0,\ 1.0$ \\
                \midrule
        \multirow{18}{*}{Cognition}
        & $\card{\mathcal I}$ & $6$ \\
        & state & $\bx = \left[x_{\mathrm{perform}}^{\mathrm G}, x_{\mathrm{avoid}}^{\mathrm G}, x_{\mathrm{comfort}}^{\mathrm B}, x_{\mathrm{capability}}^{\mathrm B}, x_{\mathrm{threat}}^{\mathrm B}, x_{\mathrm{fatigue}}^{\mathrm B}\right]^\top$ \\
        & $\Delta t,\kappa$ & $0.08,\ 0.55$ \\
        & $\boldsymbol{\gamma}$ &
        $\begin{bmatrix}-0.50 & 0.25 & -1.20 & -1.20 & 0.35 & 0.50\end{bmatrix}^{\top}$ \\
        & $\Lambda$ &
        $\begin{bmatrix}
        0.05 & 0.10\\
        0.16 & 0.12\\
        0.10 & 0.25\\
        0.00 & 0.15\\
        0.25 & 0.35\\
        0.00 & 0.35
        \end{bmatrix}$ \\
        & $\Phi$ &
        $\begin{bmatrix}
        0 & 0 & 0.30 & 0.30 & -0.30 & -0.25\\
        0 & 0 & -0.20 & -0.20 & 0.40 & 0.45\\
        0.10 & 0.45 & 0 & 0.25 & -0.08 & -0.25\\
        0.03 & -0.08 & 0.25 & 0 & -0.06 & -0.25\\
        0 & 0 & -0.15 & 0 & 0 & 0\\
        0 & 0 & 0 & 0 & 0 & 0
        \end{bmatrix}$ \\
        & $\Theta$ &
        $\begin{bmatrix}
        0.04 & 0.02\\
        0.28 & -0.30\\
        -0.03 & -0.70\\
        -0.06 & 0.35\\
        0.45 & -0.45\\
        0.00 & 2.20
        \end{bmatrix}$ \\
        & $\Psi_1(\mathrm{avoid},\mathrm{threat})$ & $0.60$ \\
        & $\Psi_1(\mathrm{avoid},\mathrm{fatigue})$ & $0.65$ \\
        & $\Psi_1(\mathrm{perform},\mathrm{comfort})$ & $-0.15$ \\
        & $\Psi_1(\mathrm{perform},\mathrm{capability})$ & $-0.15$ \\
        & $\Psi_2(\mathrm{comfort},\mathrm{fatigue})$ & $-0.45$ \\
        & $\Psi_2(\mathrm{capability},\mathrm{fatigue})$ & $-0.35$ \\
        & $\Psi_2(\mathrm{perform},\mathrm{capability})$ & $0.10$ \\
        & $\Xi_{\mathrm{fatigue}}(\mathrm{avoid},\mathrm{threat})$ & $0.45$ \\
        & $\Xi_{\mathrm{threat}}(\mathrm{avoid},\mathrm{fatigue})$ & $0.35$ \\
        & $\Xi_{\mathrm{comfort}}(\mathrm{perform},\mathrm{capability})$ & $0.25$ \\
        & $\Xi_{\mathrm{capability}}(\mathrm{perform},\mathrm{comfort})$ & $0.25$ \\
        \midrule
        \multirow{3}{*}{Initial conditions}
        & $\bx(0)$ &
        $\begin{bmatrix}0.45 & 0.25 & 0.35 & 0.30 & 0.20 & 0.15\end{bmatrix}^{\top}$ \\
        & $\hat{\bx}(0)$ & $\bx(0)+\mathcal N(0,0.08^2 I)$, $\xhatb{comfort}(0) = x_{\mathrm{comfort}}^{\mathrm B}(0)$ \\
        & $\bxnpe(0), \hat{\bx}^{\mathrm{LPE}}(0)$ & $\mathbf{0}, \mathbf{0}$ \\
        \midrule
        \multirow{9}{*}{Intention formation}
        & $\phi^{\textup{\textsc{if}}}_{\mathrm{perform}}, \phi^{\textup{\textsc{if}}}_{\mathrm{avoid}}$ & $0.045,\ 0.035$ \\
        & $\eta_8,\eta_9,\eta_{10}, \eta_{11}$ 
        & $1.8,\ 0.35,\ 0.0,\ 0.0$ \\ & $\fisimax,\gisimax,\lambda_i^{\mathrm{IF}}$ 
        & $1.0,\ 1.0,\ 5.0$ for all $i \in \mathcal{G}$\\
        & $\mathcal B_{\mathrm{perform}}^+,\mathcal B_{\mathrm{perform}}^-$
        & $\{\mathrm{comfort},\mathrm{capability}\},\ 
           \{\mathrm{threat},\mathrm{fatigue}\}$ \\
        & $\mathcal B_{\mathrm{avoid}}^+,\mathcal B_{\mathrm{avoid}}^-$
        & $\{\mathrm{threat},\mathrm{fatigue}\},\ 
           \{\mathrm{comfort},\mathrm{capability}\}$ \\
        & $w_{\mathrm{perform},j}^{+,\mathrm{base}}$
        & $w_{\mathrm{perform},\mathrm{comfort}}^{+}=1.2,\ 
           w_{\mathrm{perform},\mathrm{capability}}^{+}=1.2$ \\
        & $w_{\mathrm{perform},j}^{-,\mathrm{base}}$
        & $w_{\mathrm{perform},\mathrm{threat}}^{-}=1.1,\ 
           w_{\mathrm{perform},\mathrm{fatigue}}^{-}=0.8$ \\
        & $w_{\mathrm{avoid},j}^{+,\mathrm{base}}$
        & $w_{\mathrm{avoid},\mathrm{threat}}^{+}=1.3,\ 
           w_{\mathrm{avoid},\mathrm{fatigue}}^{+}=1.0$ \\
        & $w_{\mathrm{avoid},j}^{-,\mathrm{base}}$
        & $w_{\mathrm{avoid},\mathrm{comfort}}^{-}=0.9,\ 
           w_{\mathrm{avoid},\mathrm{capability}}^{-}=0.9$ \\
        \midrule
        \multirow{4}{*}{Action selection}
        & $\phi^{\textup{\textsc{as}}}_{\mathrm{perform}}, \phi^{\textup{\textsc{as}}}_{\mathrm{avoid}}$ & $0.25,\ 0.27$ \\
        & $\lambda_i^{\mathrm{AS}}$ & $7.0$ for all $i\in\mathcal{G}$ \\
        & $\gamma_{ij}^{\mathrm{AS}}$ & $0.65$ for $i\neq j$, $0$ for $i=j$ \\
        & $\eta_{12}$ & $1.0$ \\
        \midrule
        \multirow{8}{*}{Controller cost}
        & $\lambda_1$ & $2.0$ \\
        & $\lambda_2$ & $1.5$ \\
        & $\lambda_3$ & $20.0$ \\
        & $\lambda_4$ & $1.0$ \\
        & $\lambda_5$ & $3.0$ \\
        & $\lambda_6$ & $1.0$ \\
        & $\lambda_7$ & $0.10$ \\
        & $\xbarbcomfort$ & $0.0$ \\
        \bottomrule
    \end{tabular}%
    }
\end{table}

\newpage

\subsection{Baseline sensitivity-analysis parameters}
\label{appendix:baselinesensitivityparams}

\begin{table}[H]
    \centering
    \caption{Baseline parameter values used in the numerical simulations. One-at-a-time parameter sweeps vary only the parameter under study and hold all other parameters fixed at the values shown here. Quantities written as $\operatorname{Unif}(a,b)$ are sampled independently from a uniform distribution on $[a,b]$. The parameters $\sigma^{\mathrm{SA}}$, $\sigma^{\mathrm{DA}}$, and $\sigma^{\mathrm{SV}}$ denote the standard deviations of the Gaussian perturbations in the corresponding input families. For the \gls{sv} family, $\mathrm{corr}^{\mathrm{SV}}$ denotes the channel-correlation coefficient used to construct $\Sigma^{\mathrm{SV}}$, with covariance entries proportional to $\left(\mathrm{corr}^{\mathrm{SV}}\right)^{|\ell-m|}$. For the \gls{dv} family, $\sigma^{\mathrm{DV}}$ denotes the Gaussian component of the composite perturbation $\epsilon_{\ell}^{\mathrm{DV}}(k)$. The parameters $\lambda_i^{\mathrm{IF}}$, $\lambda_{i}^{\mathrm{AS,G}}$, $\lambda_{i}^{\mathrm{AS,H}}$, and $\lambda_{i}^{\mathrm{AS,S}}$ denote the steepness parameters of the corresponding sigmoid nonlinearities used in the intention formation and action selection modules.}
    \label{tab:baseline_simulation_parameters}
    \resizebox{\linewidth}{!}{%
    \begin{tabular}{lll}
        \toprule
        Module/layer & Parameter & Baseline value \\
        \midrule

        \multirow{17}{*}{Input generation}
        & $T,\card{\mathcal L}$ & $300,\ 5$ \\
        & input seed & $42$ \\
        & $\overline u_\ell, a$ & $1.0$, $4.0$ \\
        & $\mu_\ell^{\mathrm{SA}}$ & $\operatorname{Unif}(0.10,0.70)$ \\
        & $\mu_{1,\ell}^{\mathrm{DA}}$ & $\operatorname{Unif}(0.10,0.50)$ \\
        & $\mu_\ell^{\mathrm{SV}}$ & $\operatorname{Unif}(0.10,0.75)$ \\
        & $\mu_\ell^{\mathrm{DV}}$ & $\operatorname{Unif}(0.10,0.65)$ \\
        & $A_\ell^{\mathrm{DA}}, A_\ell^{\mathrm{DV}}$ 
        & $\operatorname{Unif}(0.05,0.40)$, $\operatorname{Unif}(0.05,0.35)$ \\
        & $f_{1,\ell}^{\mathrm{DA}}, f_2^{\mathrm{DA}}$ 
        & $\operatorname{Unif}(1/120,1/20)$ \\
        & $f_\ell^{\mathrm{DV}}$ 
        & $\operatorname{Unif}(1/180,1/45)$ \\
        & $\varphi_{1,\ell}^{\mathrm{DA}}, \varphi_2^{\mathrm{DA}}, \varphi_\ell^{\mathrm{DV}}$ 
        & $\operatorname{Unif}(0,2\pi)$ \\
        & $\rho^{\mathrm{SA}}, \rho^{\mathrm{SV}}, \rho^{\mathrm{DV}}$ 
        & $0.98,\ 0.97,\ 0.94$ \\
        & $\sigma^{\mathrm{SA}}, \sigma^{\mathrm{DA}}, \sigma^{\mathrm{SV}}, \sigma^{\mathrm{DV}}$ 
        & $0.06,\ 0.08,\ 0.04,\ 0.10$ \\
        & $\mu_2^{\mathrm{DA}}$ & $0.25$ \\
        & $\mathrm{corr}^{\mathrm{SV}}$ & $0.50$ \\
        & DV occlusion probability & $0.025$ \\
        & DV occlusion amplitude & $\operatorname{Unif}(-0.35,0.35)$ \\

        \midrule

        \multirow{5}{*}{Attentional selection}
        & $\faslmax$ & $0.75$ \\
        & $\eta_1, \eta_2, \eta_3$ & $1.80, 0.60, 1.0$ \\
        & $\betaasl$ & $0.65$ \\
        & $p_m$ & $2$ \\
        & $\gamma_{\ell m}^{\mathrm{ATS,base}}$ & $0.70$ \\

        \midrule

        \multirow{3}{*}{Predictive inference}
        & $\eta_4, \thetapil$ & $1.00$ \\
        & $T^{\mathrm f}$ & $25$ \\
        & $\bxnpe(0)$ & $\mathbf 0$ \\

        \midrule

        \multirow{10}{*}{Cognition}
        & $\card{\mathcal I}$ & $4$ \\
        & $\Delta t$ & $0.08$ \\
        & $\kappa$ & $0.90$ \\
        & $\eta_5, \eta_6, \eta_7$ & $1.0$ \\
        & $\boldsymbol{\gamma}$ &
        $\begin{bmatrix}0.15 & 0.05 & 0.10 & 0.00\end{bmatrix}^{\top}$ \\
        & $\Lambda$ &
        $\begin{bmatrix}
        0.20 & -0.10 & 0.00 & 0.05 & 0.10 \\
        0.00 &  0.15 & 0.10 & 0.00 & 0.05 \\
       -0.10 &  0.00 & 0.20 & 0.10 & 0.00 \\
        0.05 &  0.10 & 0.00 & 0.20 & -0.05
        \end{bmatrix}$ \\
        & $\Phi^{\mathrm{base}}$ &
        $\begin{bmatrix}
        0.00 &  0.18 & -0.08 &  0.05 \\
        0.10 &  0.00 &  0.12 & -0.04 \\
       -0.06 &  0.14 &  0.00 &  0.10 \\
        0.08 & -0.05 &  0.16 &  0.00
        \end{bmatrix}$ \\
        & $\Theta^{\mathrm{base}}$ &
        $\begin{bmatrix}
        0.25 & 0.05 & 0.00 & 0.10 & 0.00 \\
        0.00 & 0.20 & 0.10 & 0.00 & 0.05 \\
        0.05 & 0.00 & 0.25 & 0.10 & 0.00 \\
        0.10 & 0.05 & 0.00 & 0.20 & 0.05
        \end{bmatrix}$ \\
        & $\{\Psi_\ell^{\mathrm{base}}\}_{\ell=1}^{5}$ &
        zero diagonal, $\mathcal{N}(0,0.05^2)$, seed $789$ \\
        & $T^{\mathrm f}$ & $25$ \\

        \midrule

        \multirow{5}{*}{Intention formation}
        & $\card{\mathcal G}, \card{\mathcal{B}}$, goal indices, belief indices 
        & $2, 2$, $\{3,4\}$, $\{1,2\}$ in one-indexed notation \\
        & $\phi_i^{\mathrm{IF}}, \fisimax, \gisimax$ 
        & $0.05,\ 0.85,\ 1.0$ \\
        & $\eta_8,\eta_9,\eta_{10},\eta_{11}$ 
        & $1.80,\ 0.60,\ 0.0,\ 1.0$ \\
        & $\lambda_i^{\mathrm{IF}}$ 
        & $2.5$ \\
        & $w^{+,\mathrm{base}}, w^{-,\mathrm{base}}$ &
        $\begin{bmatrix}1.2 & 0.0 \\ 0.0 & 1.2\end{bmatrix}, \begin{bmatrix}0.0 & 0.8 \\ 0.8 & 0.0\end{bmatrix}$ \\

        \midrule

        \multirow{5}{*}{Action selection}
        & $\gasimax, \hasimax$ & $1.0$ \\
        & $\lambda_{i}^{\mathrm{AS,G}}, \lambda_{i}^{\mathrm{AS,H}}, \lambda_{i}^{\mathrm{AS,S}}$ 
        & $7.0,\ 7.0,\ 8.0$ \\
        & $\phi_i^{\mathrm{AS}}$ & $0.35$ \\
        & $\eta_{12}$ & $1.0$ \\
        & $\gamma_{ij}^{\mathrm{AS,base}}$ & $1$ for $i\neq j$, $0$ for $i=j$ \\

        \midrule

        \multirow{2}{*}{Monte Carlo}
        & MC runs & $20$ \\
        & seed & $1000$ \\

        \bottomrule
    \end{tabular}%
    }
\end{table}

\FloatBarrier
\clearpage
\onecolumn
 
\subsection{Sensitivity metrics}
\label{appendix:sensitivitymetrics}

\begin{table}[H]
    \centering
    \caption{Sensitivity and diagnostic metrics used in the numerical analysis. Here, $T$ denotes the simulation horizon, $N$ the number of sampled input or trajectory pairs, and $T^{\mathrm f}\leq T$ the length of the terminal averaging window. The mapping $\boldsymbol{\phi}(\cdot)$ denotes a generic static \gls{io} map, and $\boldsymbol{t}(k)$ denotes a generic trajectory.}
    \label{tab:sensitivity_metrics}
    \resizebox{\textwidth}{!}{%
    \begin{tabular}{llll}
    \toprule
    Metric & Definition & Used for & Interpretation \\
    \midrule
    
    Sampled \gls{io} gain
    &
    $\displaystyle
    \max_{j=1,\dots,N}
    \frac{
    \norm{\boldsymbol{\phi}(\bu_1^{(j)})-\boldsymbol{\phi}(\bu_2^{(j)})}
    }{
    \norm{\bu_1^{(j)}-\bu_2^{(j)}}
    }$
    &
    Static mappings
    &
    Sampled worst-case sensitivity; not an analytic Lipschitz constant.
    \\[2.0ex]
    
    Mean step increment
    &
    $\displaystyle
    \frac{1}{T-1}
    \sum_{k=0}^{T-2}
    \norm{\boldsymbol t(k+1)-\boldsymbol t(k)}
    $
    &
    All trajectories
    &
    Temporal responsiveness or smoothness.
    \\[2.0ex]
    
    Winner margin
    &
    $\displaystyle
    z^{\mathrm A}_{(1)}(k)-z^{\mathrm A}_{(2)}(k)
    $
    &
    Action selection
    &
    Decisiveness of the selected action.
    \\[2.0ex]
    
    Switch count
    &
    $\displaystyle
    \left|
    \left\{
    k:
    \arg\max_{i\in\mathcal G}\zai(k+1)
    \neq
    \arg\max_{i\in\mathcal G}\zai(k)
    \right\}
    \right|
    $
    &
    Action selection
    &
    Temporal instability of the selected action.
    \\[2.0ex]
    
    Mean \gls{npe} tracking error
    &
    $\displaystyle
    \frac{1}{T}
    \sum_{k=0}^{T-1}
    \abs{\enpe_{\ell}(k)}
    $
    &
    Predictive inference
    &
    Average mismatch between \gls{npe} state and attentional selection target.
    \\[2.0ex]
    
    Terminal \gls{npe} tracking error
    &
    $\displaystyle
    \frac{1}{T^{\mathrm{f}}}
    \sum_{k=T-T^{\mathrm{f}}}^{T-1}
    \abs{\enpe_{\ell}(k)}
    $
    &
    Predictive inference
    &
    Late-time residual tracking error.
    \\[2.0ex]
    
    Maximum state norm
    &
    $\displaystyle
    \max_{k=0,\dots,T-1}
    \norm{\bx(k)}
    $
    &
    Dynamic mappings
    &
    Practical excitation magnitude.
    \\[2.0ex]
    
    Sampled \gls{is} gain
    &
    $\displaystyle
    \max_{j=1,\dots,N}
    \frac{
    \sup_{0\leq k < T}\norm{\bx_1^{(j)}(k)-\bx_2^{(j)}(k)}
    }{
    \sup_{0 \leq k < T}\norm{\bu_1^{(j)}(k)-\bu_2^{(j)}(k)}
    }
    $
    &
    Dynamic mappings
    &
    Sampled propagation of input perturbations into state trajectories.
    \\
    Terminal steady-state estimate
    &
    $\displaystyle
    \hat{\bx}^{\star}
    =
    \frac{1}{T^{\mathrm f}}
    \sum_{k=T-T^{\mathrm f}}^{T-1}
    \bx(k)
    $
    &
    Constant-input cognition scans
    &
    Estimated steady state under fixed \gls{npe} state.
    \\[2.0ex]
    \bottomrule
    \end{tabular}%
    }
\end{table}

\end{document}